\documentclass[a4paper,UKenglish,cleveref, autoref]{lipics-v2019}

\usepackage{amsmath,amssymb,amsfonts}
\usepackage{graphicx,url,cite,algorithm,booktabs,algpseudocode}
\usepackage{hyperref,colortbl,longtable,dsfont}
\usepackage[usenames,dvipsnames]{xcolor}

\newcommand{\vc}[1]{\boldsymbol{#1}}
\newcommand{\mat}[1]{\vc{\mathord{\mathrm{#1}}}}

\newcommand{\No}{\mathcal{N}}

\newcommand{\cD}{\mathcal{D}}
\newcommand{\cL}{\mathcal{L}}
\newcommand{\cS}{\mathcal{S}}
\newcommand{\cV}{\mathcal{V}}
\newcommand{\cC}{\mathcal{C}}
\newcommand{\cH}{\mathcal{H}}
\newcommand{\cT}{\mathcal{T}}

\newcommand{\cR}{\mathcal{R}}
\newcommand{\dec}{\operatorname{dec}}
\newcommand{\midsep}{}

\newcommand{\R}{\mathbb{R}}

\newcommand{\vol}{\operatorname{Vol}}

\newtheorem{conjecture}[theorem]{Conjecture}

\title{Polytopes, lattices, and spherical codes\protect\\ for the nearest neighbor problem}
\titlerunning{Polytopes, lattices, and spherical codes for the nearest neighbor problem}

\author{Thijs Laarhoven}{Eindhoven University of Technology, Eindhoven, The Netherlands}{mail@thijs.com}{https://orcid.org/0000-0002-2369-9067}{}
\authorrunning{T. Laarhoven}

\Copyright{Thijs Laarhoven}

\ccsdesc[500]{Theory of computation~Nearest neighbor algorithms}
\ccsdesc[100]{Theory of computation~Random projections and metric embeddings}
\ccsdesc[100]{Theory of computation~Computational geometry}

\keywords{(approximate) nearest neighbor problem, spherical codes, polytopes, lattices, locality-sensitive hashing (LSH)}

\funding{The author is supported by an NWO Veni grant with project number 016.Veni.192.005. Part of this work was done while the author was visiting the Simons Institute for the Theory of Computing at the University of California, Berkeley.}

\EventEditors{Artur Czumaj, Anuj Dawar, and Emanuela Merelli}
\EventNoEds{3}
\EventLongTitle{47th International Colloquium on Automata, Languages, and Programming (ICALP 2020)}
\EventShortTitle{ICALP 2020}
\EventAcronym{ICALP}
\EventYear{2020}
\EventDate{July 8--11, 2020}
\EventLocation{Saarbrücken, Germany (virtual conference)}
\EventLogo{}
\SeriesVolume{168}
\ArticleNo{76}

\nolinenumbers
\hideLIPIcs

\begin{document}

\maketitle

\begin{abstract}
We study locality-sensitive hash methods for the nearest neighbor problem for the angular distance, focusing on the approach of first projecting down onto a random low-dimensional subspace, and then partitioning the projected vectors according to the Voronoi cells induced by a well-chosen spherical code. This approach generalizes and interpolates between the fast but asymptotically suboptimal hyperplane hashing of Charikar [STOC 2002], and asymptotically optimal but practically often slower hash families of e.g.\ Andoni--Indyk [FOCS 2006], Andoni--Indyk--Nguyen--Razenshteyn [SODA 2014] and Andoni--Indyk--Laarhoven--Razenshteyn--Schmidt [NIPS 2015]. We set up a framework for analyzing the performance of any spherical code in this context, and we provide results for various codes appearing in the literature, such as those related to regular polytopes and root lattices. Similar to hyperplane hashing, and unlike e.g.\ cross-polytope hashing, our analysis of collision probabilities and query exponents is \textit{exact} and does not hide any order terms which vanish only for large $d$, thus facilitating an easier parameter selection in practical applications.

For the two-dimensional case, we analytically derive closed-form expressions for arbitrary spherical codes, and we show that the equilateral triangle is optimal, achieving a better performance than the two-dimensional analogues of hyperplane and cross-polytope hashing. In three and four dimensions, we numerically find that the tetrahedron and $5$-cell (the $3$-simplex and $4$-simplex) and the $16$-cell (the $4$-orthoplex) achieve the best query exponents, while in five or more dimensions orthoplices appear to outperform regular simplices, as well as the root lattice families $A_k$ and $D_k$ in terms of minimizing the query exponent. We provide lower bounds based on spherical caps, and we predict that in higher dimensions, larger spherical codes exist which outperform orthoplices in terms of the query exponent, and we argue why using the $D_k$ root lattices will likely lead to better results in practice as well (compared to using cross-polytopes), due to a better trade-off between the asymptotic query exponent and the concrete costs of hashing.
\end{abstract}


\section{Introduction}

Given a large database of high-dimensional vectors, together with a target data point which does not lie in this database, a natural question to ask is: which item in the database is the most similar to the query? And can we somehow preprocess and store the database in a data structure that allows such queries to be answered faster? These and related questions have long been studied in various contexts, such as machine learning, coding theory, pattern recognition, and cryptography~\cite{duda00, shakhnarovich05, bishop06, dubiner10, laarhoven15crypto, may15}, under the headers of \emph{similarity search} and \emph{nearest neighbor searching}. Observe that a naive solution might consist of simply storing the data set in a big list, and to search this list in linear time to find the nearest neighbor to a query point. This solution requires an amount of time and space scaling linearly in the size of the data set, and solutions we are interested in commonly require more preprocessed space (and time), but achieve a sublinear query time complexity to find the nearest neighbor.

Depending on the context of the problem, different solutions for these problems have been proposed and studied. For the case where the dimensionality of the original problem is constant, efficient solutions are known to exist~\cite{arya94}. Throughout the remainder of the paper, we will therefore assume that the dimensionality of the data set is superconstant. In the late 1990s, Indyk--Motwani~\cite{indyk98} proposed the \emph{locality-sensitive hashing} framework, and until today this method remains one of the most prominent and popular methods for nearest neighbor searching in high-dimensional vector spaces, both due to its asymptotic performance when using theoretically optimal hash families~\cite{andoni14, andoni15cp}, and due to its practical performance when instantiated with truly efficient locality-sensitive hash functions~\cite{achlioptas01, charikar02, annbench1}. And whereas many other methods scale poorly as the dimensionality of the problem increases, locality-sensitive hashing remains competitive even in high-dimensional settings.

Although solutions for both asymptotic and concrete settings have been studied over the years, there is often a separation between both worlds: some methods work well in practice but do not scale optimally when the parameters increase (e.g.\ Charikar's hash family~\cite{charikar02}); while some other methods are known to achieve a superior performance for sufficiently large parameter sizes, but may not be quite as practical in some applications due to large hidden order terms in the asymptotic analysis (e.g.\ hash families studied in~\cite{andoni06, andoni14, andoni15cp} and filter families from~\cite{becker16lsf, andoni17, christiani17}). Moreover, the latter methods are often not as easy to deploy in practice due to these unspecified hidden order terms, making it hard to choose the scheme parameters that optimize the performance. A key problem in this area thus remains to close the gap between theory and practice, and to offer solutions that interpolate between \textit{quick--and--dirty} simple approaches that might not scale well with the problem size, and more sophisticated methods that only start outperforming these simpler methods as the parameters are sufficiently large.

\subsection{Related work} 

In this paper we will focus on methods for solving the nearest neighbor problem for the \emph{angular distance}, which is closely related to the nearest neighbor problem in various $\ell_p$-norms: as shown by Andoni and Razenshteyn~\cite{andoni15}, a solution for the nearest neighbor problem for the angular distance (or the Euclidean distance on the sphere) can be optimally extended to a solution for the $\ell_2$-norm for all of $\mathbb{R}^d$. Solutions for the $\ell_2$-norm can further be translated to e.g.\ solutions for the $\ell_1$-norm via appropriate embeddings. 

For the angular distance, perhaps the most well-known and widely deployed approach for finding similar items in a large database is to use the hyperplane hash family of Charikar~\cite{charikar02}. For spherically-symmetric, random data sets on the sphere of size $n$, it can find a nearest neighbor at angle at most e.g.\ $\theta = \frac{\pi}{3}$ in sublinear time $\tilde{O}(n^{\rho})$ and space $\tilde{O}(n^{1 + \rho})$, with $\rho \approx 0.5850$. Various improvements later showed that smaller query exponents $\rho$ can be achieved in sufficiently high dimensions~\cite{andoni06, andoni14}, the most practical of which is based on cross-polytopes or orthoplices~\cite{terasawa07, andoni15cp, kennedy17}: for large $d$ and for the same target angle $\theta = \frac{\pi}{3}$, the query time complexity scales as $n^{\rho + o(1)}$ with $\rho = 1/3$. Note that the convergence to the limit is rather slow~\cite[Theorem 1]{andoni15cp} and depends on both $d$ and $n$ being sufficiently large -- for fixed $d$ and large $n$, the exponent is still larger than $1/3$. In certain practical applications with moderate-sized data sets, hyperplane hashing still appears to outperform cross-polytope hashing and other advanced hashing and filtering schemes~\cite{mariano15, mariano17, becker16cp, albrecht19, becker16lsf} due to its low cost for hashing, and the absence of lower order terms in the exponent.

Related to the topic of this paper, various other works have further observed the relation between finding good locality-sensitive hash families for the angular distance and finding ``nice'' spherical codes that partition the sphere well, and allow for efficient decoding~\cite{andoni06, terasawa07, terasawa09, andoni15cp, becker16lsf, chandrasekaran18}. The requirements on a spherical code to be a \emph{good} spherical code are somewhat intricate to state, and to date it is an open problem to exactly quantify which spherical codes are the most suited for nearest neighbor searching. It may well be possible to improve upon the state-of-the-art orthoplex (cross-polytope) locality-sensitive hash family~\cite{andoni15cp, falconn} in practice with a method achieving the same asymptotic scaling, but with a faster convergence to the limit, and thus potentially a better performance in practice for large problem instances.

\subsection{A framework for evaluating spherical codes}

As a first contribution of this paper, we provide a framework for analyzing the performance of arbitrary spherical codes in the context of locality-sensitive hashing for the angular distance, where we focus on the approach of (1) projecting down onto a random low-dimensional subspace, and (2) partitioning the resulting subspace according to the Voronoi cells induced by a spherical code. More specifically, we relate the collision probabilities appearing in the analysis of these hash functions to so-called \emph{orthant probabilities} of multivariate normal distributions with non-trivial correlation matrices. Below we informally state this relation for general spherical codes, which provides us a recipe for computing the collision probabilities (and the query exponent) as a function of the set of vertices $\cC$ of the corresponding spherical code. Here a hash family being $(\theta, p_1, p_2)$-sensitive means that uniformly random vectors on the sphere collide with probability at most $p_2$, and target nearest neighbors at angle at most $\theta$ from a random query vector collide in a random hash function with probability at least $p_1$.
\begin{theorem}[Spherical code locality-sensitive hashing]
Let $\cC \subset \cS^{k-1}$ be a $k$-dimensional spherical code, and consider a hash family where:
\begin{itemize}
\item We first project onto a $k$-dimensional subspace using a random matrix $\mat{A} \sim \No(0,1)^{k \times d}$;
\item We then assign hash values based on which $\vc{c} \in \cC$ is nearest to the projected vector. 
\end{itemize}
Then for any $\theta \in (0, \frac{\pi}{2})$, this family is $(\theta, p_1, p_2)$-sensitive, where $p_2$ can be expressed in terms of the relative volumes of the Voronoi cells induced by $\cC$, and $p_1$ can be expressed as a sum of orthant probabilities $\Pr_{\vc{z} \sim \No(\vc{0},\mat{\Sigma}_i)}\left(\vc{z} \geq \vc{0}\right)$, where each $\mat{\Sigma}_i$ has size $(2k) \times (2k)$.
\end{theorem}
The locality-sensitive hashing exponent $\rho = \log p_1 / \log p_2$ describes the distinguishing power of a hash family, as for large $n$ this tells us that we can solve the (average-case) nearest neighbor problem with target angle $\theta$ in query time $\tilde{O}(n^{\rho})$ with space $\tilde{O}(n^{1 + \rho})$. The theorem above essentially provides us a recipe which, given any spherical code $\cC$ and target angle $\theta$ as input, tells us how to compute these probabilities $p_1$ and $p_2$ (and thus $\rho$) exactly. 

While the above theorem is somewhat abstract, we show how to explicitly construct the correlation matrices $\mat{\Sigma}_i$ appearing in the theorem, allowing us to always at least obtain numerical estimates on the performance of different spherical codes in the context of nearest neighbor searching. We further study how to reduce the dimensionality of the problem (and in particular, the sizes of the matrices $\mat{\Sigma}_i$) when the spherical code exhibits many symmetries, such as being isogonal, and we show how using only the \emph{relevant vectors} of each code word can further simplify the computations. In most cases the resulting orthant probabilities will still remain too complex to evaluate analytically, but in some cases we can obtain exact expressions this way.

\subsection{A survey of known spherical codes}

\subparagraph{One and two dimensions.} Using this new framework, we then apply it to various spherical codes, starting in very low dimensions. For the one-dimensional case we rederive the celebrated result of Charikar~\cite{charikar02} for one-dimensional spherical codes, noting that the collision probability analysis in fact dates back to an old result from the late 1800s~\cite{sheppard99}. Applying the same framework to two-dimensional codes, among others we establish a general formula for collision probabilities and query exponents for arbitrary polygons, as well as a more compact description of the collision probabilities for regular polygons.
\begin{theorem}[Collision probabilities for regular polygons]
Let $\cC \subset \cS^1$ consist of the vertices of the regular $c$-gon, for $c \geq 2$, and let $\theta \in (0, \frac{\pi}{2})$. Then the corresponding project--and--partition hash family $\cH$ is $(\theta, p_1, p_2)$-sensitive, with:
\begin{align}
p_1 &= \frac{1}{c} + c \left(\frac{\pi - \theta}{2 \pi}\right)^2 - c \left(\frac{\arccos(- \cos \theta \cos \frac{2 \pi}{c})}{2 \pi}\right)^2, \qquad \qquad p_2 = 1/c. \label{eq:2d}
\end{align}
\end{theorem}
As a direct corollary of the above theorem, we establish that \emph{triangular locality-sensitive hashing} achieves a superior asymptotic performance to hyperplane hashing, and achieves the lowest query exponents $\rho$ among all regular polygons.
\begin{corollary}[Triangular locality-sensitive hashing]
Consider the hash family where we first project onto a random plane using a random projection matrix $\mat{A} \sim \No(0,1)^{2 \times d}$, and then decode to the nearest corner of a fixed equilateral triangle centered at $(0,0)$. Then, for target angles $\theta \in (0, \frac{\pi}{2})$, this hash family achieves query exponents $\rho$ of the form:
\begin{align}
\rho = \ln\left[\frac{1}{3} + 3 \left(\frac{\pi - \theta}{2 \pi}\right)^2 - 3 \left(\frac{\arccos(\frac{1}{2} \cos \theta)}{2 \pi}\right)^2\right] \Big/ \ln \left[ \frac{1}{3}\right].
\end{align}
Among all such project--and--partition hash families based on regular $k$-gons, this family achieves the lowest values $\rho$ for any target angle $\theta \in (0, \frac{\pi}{2})$.
\end{corollary}
Note again that the above statement of $\rho$ is exact, and does not hide any order terms in $d$ or $n$ -- in fact, the collision probabilities do not depend on $d$ at all, due to our choice of $\mat{A}$ being Gaussian. Further note that the $2$-gon in the plane corresponds to the one-dimensional antipodal code of Charikar, and triangular hashing therefore strictly improves upon hyperplane hashing for any $\theta$ in terms of the query exponent $\rho$. For instance, we can find neighbors at angle at most $\frac{\pi}{3}$ in time $\tilde{O}(n^{\rho})$ with $\rho \approx 0.56996$, offering a concrete but minor improvement over the hyperplane hashing approach of Charikar~\cite{charikar02} with query exponent $\rho \approx 0.5850$. This is the best we can do with any two-dimensional isogonal spherical code, and we numerically predict that this code achieves the lowest values $\rho$ among all (not necessarily isogonal) two-dimensional spherical codes as well.

\subparagraph{Three and four dimensions.} For three-dimensional codes, through numerical integration of the resulting orthant probabilities we conclude that the \emph{tetrahedron} appears to minimize the query exponent $\rho$ out of all three-dimensional codes, beating the three-dimensional analogues of e.g.\ the cross-polytope and hypercube, as well as various sphere packings and other regular polytopes appearing in the literature, such as the Platonic and Archimedean solids. In four dimensions an interesting phenomenon occurs: the so-called \emph{$5$-cell} ($4$-simplex) and \emph{$16$-cell} ($4$-orthoplex) are optimal in different regimes, with the $5$-cell inducing a more coarse-grained partition of the space and achieving a better performance when the nearest neighbor lies relatively far away from the data set, and the $16$-cell inducing a more fine-grained partition, and working better when the nearest neighbor lies relatively close to the target vector. This crossover effect is also visualized in Figure~\ref{fig:comp}, and it strengthens our intuition that as the dimensionality goes up, or as the nearest neighbor lies closer to the query, more fine-grained partitions are necessary to obtain the best performance. An overview of some of the exponents $\rho$ for various low-dimensional spherical codes, for different target angles $\theta \in \frac{\pi}{12} \{1, 2, 3, 4, 5\}$, is given in Table~\ref{tab:overview}. The best exponents $\rho$ are highlighted in bold.

\subparagraph{Five and more dimensions.} For higher-dimensional spherical codes, we obtain further improvements in the query exponents $\rho$ through the use of suitable spherical codes, as shown in Table~\ref{tab:overview} and Figures~\ref{fig:comp}--\ref{fig:prac}. For dimensions $5$ and $6$, the corresponding orthoplices achieve a better performance than the simplices, and the exotic polytopes $1_{21}$ and $2_{21}$, related to the root lattice $E_6$, seem useful in the context of nearest neighbor searching as well. 

We further study the performance of the following five non-trivial infinite families of spherical codes. Color codes below correspond to the same colors used in Table~\ref{tab:overview} and Figures~\ref{fig:comp}--\ref{fig:prac} to differentiate these families of codes.
\begin{itemize}
\item The \textbf{simplices} {\color{red}$S_k$} on $k + 1$ vertices;
\item The \textbf{orthoplices} {\color{green!50!black}$O_k$} on $2k$ vertices (also known as cross-polytopes~\cite{terasawa09, andoni15cp});
\item The \textbf{hypercubes} {\color{blue}$C_k$} on $2^k$ vertices (as studied in~\cite{laarhoven17hypercube});
\item The \textbf{expanded simplices} {\color{orange}$A_k$} on $k(k + 1)$ vertices;
\item The \textbf{rectified orthoplices} {\color{purple}$D_k$} on $2k (k - 1)$ vertices.
\end{itemize}
The latter two families are connected to the root lattices $A_k$ and $D_k$, while the first three are related to the lattice $\mathbb{Z}^k$. We conjecture that, should other nice families of dense lattices be found (e.g.\ the recent \cite{vladut19}), these may give rise to suitable spherical codes in nearest neighbor applications as well. For each of the above five families we give closed-form expressions on the correlation matrices $\mat{\Sigma}$, but the resulting orthant probabilities that need to be evaluated for computing $\rho$ do not appear to admit simple closed-form expressions. Apart from the family of hypercubes, these are all asymptotically optimal (with $\rho \to (1 - \cos \theta) / (1 + \cos \theta)$ as $k \to \infty$), although the convergence to the limit may differ for each family.

\subsection{Lower bounds via spherical caps}

Although this work tries to be exhaustive in covering as many (families of) spherical codes as possible, better spherical codes may exist, achieving even lower query exponents $\rho$. As these codes may be hard to find, and as it may be difficult to rule out the existence of other, better spherical codes, the next best thing one might hope for is a somewhat tight lower bound on the performance of any $k$-dimensional spherical code in our framework, which hopefully comes close to the performance of the spherical codes we have considered in this survey.

As has been established in several previous works on nearest neighbor searching on the sphere~\cite{andoni06, andoni14, becker16lsf, laarhoven15nns, andoni17, christiani17}, ideally we would like the hash regions induced by the partitions to take the shape of a spherical cap. Such a shape minimizes the angular radius, given that the region has a fixed volume, and in a sense it is the most natural and smoothest shape that a region on a sphere can take. So in a utopian world, one might hope that a spherical code partitions the sphere into $c$ regions, and each region corresponds exactly to a spherical cap of volume $\vol(\cS^{k-1}) / c$. Clearly such spherical codes do not exist for $k > 2$ and $c > 2$, but such an extremal example does give us an indication on the limits of what might be achievable in dimension $k$, and how the optimal $\rho$ decreases with $k$. Note that random spherical codes might approach this utopian setting in high dimensions.

Following the above reasoning, and using a classic result of Baernstein--Taylor~\cite{baernstein76}, we state a formal lower bound on the performance parameter $\rho$ of any $k$-dimensional spherical code of size $c$, and by minimizing over $c$ for given $k$, one obtains a lower bound for any spherical code living in $k$ dimensions. Here $I_x(a,b)$ denotes the regularized incomplete beta function, which comes from computing the volume of a sphere in $k$ dimensions, while $\|\cdot\|$ denotes the Euclidean norm.
\begin{theorem}[Spherical cap lower bounds] \label{thm:lb}
Let $\cC \subset \cS^{k-1}$ be a spherical code, and let $\cH$ be the associated project--and--partition hash family. Then the parameter $\rho$ for $\cC$, for target angle $\theta \in (0, \frac{\pi}{2})$, must satisfy:
\begin{align}
\rho(\theta) \geq \rho_k(\theta) := \min_{c \geq 2} \rho_k(c, \theta),
\end{align}
where, with $\alpha_k(c)$ denoting the solution $\alpha$ to $\frac{1}{2} \cdot I_{1 - \alpha^2}(\frac{k-1}{2}, \frac{1}{2}) = \frac{1}{c}$, $\rho_k(c, \theta)$ is given by:
\begin{align*}
\rho_k(c, \theta) := \log\left\{\Pr_{\vc{x}, \vc{y} \sim \No(0,1)^k} \left(\frac{x_1}{\|\vc{x}\|} \geq \alpha_k(c), \frac{x_1 \cos \theta + y_1 \sin \theta}{\|\vc{x} \cos \theta + \vc{y} \sin \theta\|} \geq \alpha_k(c) \right)\right\} \bigg/ \log\left(\frac{1}{c}\right).
\end{align*}
\end{theorem}
The above minimization over $c \geq 2$ concerns the possible code sizes, and $\rho_k(c,\theta)$ describes the parameter $\rho$ one would obtain when indeed, the code consisted of $c$ equivalent regions of equal volume, and each region was shaped like a spherical cap. Equivalently, one could state that any spherical code of size exactly $c$ must satisfy $\rho(\theta) \geq \rho_k(c, \theta)$.

Numerical evaluation of these expressions $\rho_k(c,\theta)$, and the resulting minimization over $c$, leads to the values in Table~\ref{tab:overview} in the rows indicated by \textit{spherical caps}. The superscripts denote the values $c$ that numerically solve the minimization problems. These results in low dimensions suggest that, especially for small angles, the optimal code size should increase superlinearly with the dimension. This inspires the following conjecture, stating that the use of orthoplices is likely not optimal in higher dimensions.
\begin{conjecture}[Orthoplices are suboptimal for large $k$] For arbitrary $\theta \in (0, \frac{\pi}{2})$, there exists a dimension $k_0 \in \mathbb{N}$ such that, for all dimensions $k \geq k_0$, there exist spherical codes $\cC \subset \cS^{k-1}$ whose query exponents $\rho$ in the project--and--partition framework are smaller than the exponents $\rho$ of the $k$-orthoplex.
\end{conjecture}
Actually finding such spherical codes, or finding families of spherical codes that outperform orthoplices may again be closely related to the problem of finding nice families of dense lattices with efficient decoding algorithms. We informally conjecture that in high dimensions, and for sufficiently large code sizes $c$, random spherical codes may be close to optimal. If we care only about minimizing $\rho$, then a further study might focus on (1) getting a better grip of the optimal scaling of $c = c(k)$ with $k$, and (2) estimating the asymptotic performance of using random spherical codes of size $c(k)$, for large $k$. We predict this will lead to better values $\rho$ than those obtained with cross-polytope hashing.

\subsection{Selecting spherical codes in practice}

Although the exponent $\rho$, and therefore the probabilities $p_1$ and $p_2$, directly imply the main performance parameters to assess the asymptotic performance of a locality-sensitive hash family, in practice we are always dealing with concrete, non-asymptotic values $n$, $d$, and $\theta$ -- if the convergence to the optimal asymptotic scaling is slow, or if the hash functions are too expensive to evaluate in practice, then hash families with lower query exponents $\rho$ may actually be less practical for small, concrete values of $n$ and $d$ than fast hash families with larger $\rho$. For example, the hash family from~\cite{andoni06} may be ``optimal'', but appears to be only of theoretical interest. In practice one needs to find a balance between decreasing $\rho$ and using hash functions which are fast to evaluate.

As a more concrete example, consider how hyperplane hashing \cite{charikar02} allows us to partition a sphere in $2^k$ regions with a single random projection matrix $\mat{A} \in \mathbb{R}^{k \times d}$ (or equivalently $k$ matrices $\mat{A}_1, \dots, \mat{A}_k \in \mathbb{R}^{1 \times d}$ merged into one large matrix). For cross-polytope hashing \cite{andoni15cp} a $k$-dimensional projection would only divide the $k$-dimensional sphere into $2k$ regions. As the total required number of hash buckets in the data structure is often roughly the same, regardless of the chosen hash family,\footnote{While a hash family with lower query exponents $\rho$ does require a smaller number of hash buckets per hash table, this effect is marginal compared to the increase in the number of projections/rotations required by e.g.\ cross-polytope hashing compared to hyperplane hashing.} this means that if we wish to divide the sphere into $2^k$ hash regions with cross-polytope hashing, we would need $m = k / \log_2(2k)$ independent random projection matrices $\mat{A}_1, \dots, \mat{A}_m \in \mathbb{R}^{k \times d}$ to hash a vector to one of $2^k$ buckets. So even though the resulting partitions for cross-polytope hashing generate smaller exponents $\rho$, in practice this improvement may not offset the additional costs of computing the projections/rotations, which is almost a linear factor $O(d)$ more than for hyperplane hashing. So especially for data sets of small/moderate sizes, hyperplane hashing may be more practical than cross-polytope hashing, as also observed in e.g.\ \cite{laarhoven15crypto, becker16cp, mariano15, mariano17}.

We explicitly quantify the trade-off between the complexity of the hash functions and the asymptotic query exponents $\rho$ in Figure~\ref{fig:to}, where on the vertical axis we plotted the query exponents $\rho$ as computed in this paper, and on the horizontal axis we plotted the number of bits extracted per row of a projection matrix; for hyperplane hashing we extract $1$ bit per projection, while e.g.\ for cross-polytope hashing in dimension $k$ we only extract $\log(2k)/k$ bits per projection. Depending on the application, it may be desirable to choose hash families with slightly larger values $\rho$, if this means the cost of the hashing becomes less. Figure~\ref{fig:to} makes a case for using hashes induced by the root lattices $D_k$, as well as those induced by exotic polytopes such as the $2_{21}$-polytope; compared to e.g. the $5$-orthoplex, the $D_{10}$ lattice achieves lower values $\rho$ and extracts more bits per projection, while the $2_{21}$-polytope further improves upon $D_{10}$ by extracting more hash bits per projection.

To further illustrate how these trade-offs might look in a real-world setting, Figure~\ref{fig:prac} describes a case study for average-case nearest neighbor searching with different sizes $n$ for the data sets. For small data sets, hyperplane hashing is quite competitive, and the best other spherical codes are those induced by the $D_k$ lattices (for small $k$), and spherical codes generated by the polytopes $1_{21}, 1_{31}$ and $2_{21}$. This case study matches the recommendations from Figure~\ref{fig:to}. For large $n$, the role of $\rho$ becomes more prominent, and higher-dimensional orthoplices and $D_k$ lattices attain the best performance. We further describe how one might choose the best codes in practice.

\subsection{Summary and open problems}

With the project--and--partition framework outlined in this paper, we can now easily formalize and analyze the performance of any spherical code in this framework, and analyze the collision probabilities and query exponents either analytically (by attempting to simplify the corresponding multivariate orthant probabilities) or numerically (by evaluating these multivariate integrals with mathematical software, such as the mvtnorm package~\cite{genz19}). We observed that already in two dimensions, it is possible to improve upon hyperplane hashing with a smaller exponent $\rho$ by using triangular partitions, and we described closed-form formulas for the resulting parameters $\rho$. In three and four dimensions the improvements in the exponent $\rho$ become more significant, with e.g.\ the tetrahedron, the $5$-cell, and the $16$-cell achieving the best exponents $\rho$ in these dimensions. For higher dimensions the simplices and orthoplices seem to achieve the best performance in theory, thus making a strong case for the use of these partitions as advertised in~\cite{terasawa07, terasawa09, andoni15cp, kennedy17}, and as observed in nearest neighbor benchmarks~\cite{annbench1, annbench2, falconn}. The family of $D_k$ lattices and the generalized $m$-max hash functions however seem to offer better practical trade-offs between the asymptotic performance in terms of $\rho$ and the costs of computing hashes, as discussed later on.

\subparagraph{Finding better spherical codes.} An important open problem, both for our project--and--partition framework and for arbitrary hash families for the angular distance, is to find (if they exist) other nice families of spherical codes which achieve an even better performance than the family of orthoplices. Numerics suggest that as $k$ increases, the optimal code size $c$ should increase faster than the linear scaling of $c = 2k$ offered by orthoplices. Finding the optimal scaling of $c$ as a function of $k$, and finding nicer spherical codes closely matching this appropriate scaling of $c$ is left for future work.

\subparagraph{Faster pseudorandom projections.} To make e.g.\ hyperplane and cross-polytope hashing more practical, various ideas were proposed in e.g.~\cite{achlioptas01, andoni15cp} to work with sparse projections and pseudorandom rotation matrices. Similar ideas can be used for any hash family, to reduce the cost of multiplication by a random Gaussian matrix $\mat{A}$. Concretely, one can often replace it with a sparse, ``sufficiently random'' pseudorandom matrix $\mat{A}$ while still achieving good query exponents $\rho$ in practice. Depending on how these pseudorandom projections are instantiated, this may lead to a different practical evaluation than what we described in our practical case analyses. In particular, as this reduces the relative cost of the hashing, this will further favor schemes which reduce $\rho$ at the cost of increasing the (naive) complexity of hashing.

\subparagraph{Using orthogonal projection matrices.} In our framework, we focused on projecting down onto a low-dimensional subspace using Gaussian matrices $\mat{A} \sim \No(0,1)^{k \times d}$. For $k \ll d$, using Gaussian matrices or orthogonal matrices (i.e.\ $\mat{A} \mat{A}^T = I_k$) is essentially equivalent~\cite{diaconis87, jiang06}, but for large $k$ it may be beneficial to use proper rotations induced by orthogonal matrices. For instance, recent work~\cite{laarhoven17hypercube} showed that for hypercube hashing in the ambient space, there is a clear gap between using random or orthogonal matrices; using orthogonal matrices generally works better than using random Gaussian projection matrices. 

The main issue when analyzing the same framework with orthogonal matrices is that the dependence on $d$ then does not disappear; the distribution of $\mat{A}\vc{x}$ then relies on both $k$ and $d$, rather than only on $\vc{k}$. Our framework is in a sense dimensionless, as the performance can be computed without knowledge of $d$, and for $k = 2$ this even allowed us to obtain explicit analytic expressions for the collision probabilities for arbitrary $k$-gons. Using orthogonal matrices, the collision probabilities will likely become complicated functions of $d$, and comparing different spherical codes may then become a more difficult task. We leave it as an open problem to adjust the above framework to the setting where $\mat{A}$ is orthogonal, and to see how much can be gained by using orthogonal rather than Gaussian matrices\footnote{Note that one obtains different asymptotics when analyzing the hypercube for Gaussian~\cite{charikar02} and orthogonal projection/rotation matrices~\cite{laarhoven17hypercube}. For constant $k = O(1)$ and large $d$ both approaches are asymptotically equivalent, but for large $k = O(d)$ we expect orthogonal matrices to give better results.}.




\subsection*{Outline}

The remainder of the paper is structured as follows. In Section~\ref{sec:preliminaries} we first describe preliminary results and notation. Section~\ref{sec:framework} defines the model for the nearest neighbor problem, and the framework for analyzing the project--and--partition approach for arbitrary spherical codes. Sections~\ref{sec:1d}--\ref{sec:4d} continue with applying this framework to spherical codes in dimensions 1 through 4, and Section~\ref{sec:kd} describes results for spherical codes in higher dimensions. Section~\ref{sec:lb} continues with (conjectured) lower bounds on the parameter $\rho$ in this framework. Section~\ref{sec:practice} concludes with a practical analysis, weighing the costs of hashing against the exponent $\rho$.


\begin{table}[p]
\footnotesize
\centering
\begin{tabular}{p{0.1cm}p{3.8cm}p{1.1cm}p{1.1cm}p{1.1cm}p{1.1cm}p{1.1cm}p{1.1cm}}
\toprule[2pt]
$k$ & spherical code & $c$ & $\rho(\frac{\pi}{12})$ & $\rho(\frac{\pi}{6})$ \ \ & $\rho(\frac{\pi}{4})$ \ \ & $\rho(\frac{\pi}{3})$ & $\rho(\frac{5\pi}{12})$ \\
\midrule[2pt]
$1$ & \emph{spherical caps} 	& 		& $\emph{0.1255}^{(2)}$	& $\emph{0.2630}^{(2)}$ 	& $\emph{0.4150}^{(2)}$ 	& $\emph{0.5850}^{(2)}$ 	& $\emph{0.7776}^{(2)}$ \\
\midsep
 & hyperplane & $2$ 	& $\mathbf{0.1255}$	& $\mathbf{0.2630}$	& $\mathbf{0.4150}$	& $\mathbf{0.5850}$	& $\mathbf{0.7776}$ \\
\midrule
$2$ & \emph{spherical caps} 		& 		& $\emph{0.1194}^{(3)}$	& $\emph{0.2518}^{(3)}$	& $\emph{0.4005}^{(3)}$ 	& $\emph{0.5700}^{(3)}$ 	& $\emph{0.7666}^{(3)}$ \\
\midsep
 & triangle ($\color{red}S_2$) 	& $3$    & $\mathbf{0.1194}$	& $\mathbf{0.2518}$	& $\mathbf{0.4005}$	& $\mathbf{0.5700}$	& $\mathbf{0.7666}$ \\
 & square ($\color{green!50!black}O_2$, $\color{blue}C_2$, $\color{purple}D_2$) & $4$  & $0.1255$	& $0.2630$	& $0.4150$	& $0.5850$	& $0.7776$ \\
 & pentagon 			& $5$   & $0.1343$	& $0.2788$	& $0.4346$	& $0.6040$	& $0.7905$ \\
 & hexagon ($\color{orange}A_2$) 		& $6$   & $0.1438$	& $0.2954$	& $0.4544$	& $0.6222$	& $0.8022$ \\
\midrule
$3$ & \emph{spherical caps} 	& 				& $\emph{0.1117}^{(5)}$ & $\emph{0.2389}^{(4)}$ & $\emph{0.3846}^{(4)}$ & $\emph{0.5548}^{(4)}$ & $\emph{0.7561}^{(4)}$ \\
\midsep
 & tetrahedron ($\color{red}S_3$) 	& $4$  	& $\mathbf{0.1155}$	& $\mathbf{0.2445}$	& $\mathbf{0.3910}$	& $\mathbf{0.5600}$	& $\mathbf{0.7592}$ \\
 & sphere packing 				& $5$  	& $0.1170$	& $0.2481$	& $0.3965$	& $0.5664$	& $0.7644$ \\
 & octahedron ($\color{green!50!black}O_3$) 	& $6$  	& $0.1159$	& $0.2465$	& $0.3952$	& $0.5661$	& $0.7649$ \\
 & sphere packing 				& $7$ 	& $0.1207$	& $0.2554$	& $0.4065$	& $0.5772$	& $0.7725$ \\
 & cube ($\color{blue}C_3$) 		& $8$ 	& $0.1255$	& $0.2630$	& $0.4150$	& $0.5850$	& $0.7776$ \\
 & sphere packing 				& $9$ 	& $0.1217$	& $0.2591$	& $0.4129$	& $0.5850$	& $0.7786$ \\
 & icosahedron 			& $12$  	& $0.1255$	& $0.2678$	& $0.4255$	& $0.5983$	& $0.7883$ \\
 & cuboctahedron ($\color{orange}A_3$, $\color{purple}D_3$)& $12$ 	& $0.1294$	& $0.2728$	& $0.4301$	& $0.6017$	& $0.7900$ \\
 & dodecahedron 		& $20$ 	& $0.1509$	& $0.3077$	& $0.4692$	& $0.6360$	& $0.8112$ \\
\midrule
$4$ & \emph{spherical caps} 	& 		& $\emph{0.1050}^{(7)}$ & $\emph{0.2285}^{(6)}$ & $\emph{0.3730}^{(6)}$ & $\emph{0.5433}^{(5)}$ & $\emph{0.7466}^{(5)}$ \\
\midsep
 & $5$-cell ($\color{red}S_4$) 	& $5$  	& $0.1126$ 	& $0.2392$	& $0.3840$	& $\mathbf{0.5527}$	& $\mathbf{0.7537}$\\
 & sphere packing 				& $6$ 	& $0.1128$ 	& $0.2401$	& $0.3861$	& $0.5555$	& $0.7564$\\
 & sphere packing				& $7$ 	& $0.1120$ 	& $0.2392$	& $0.3852$	& $0.5553$	& $0.7567$\\
 & $16$-cell ($\color{green!50!black}O_4$) 	& $8$ 	& $\mathbf{0.1107}$ 	& $\mathbf{0.2368}$	& $\mathbf{0.3822}$	& $0.5528$	& $0.7553$\\
 & sphere packing 				& $10$ 	& $0.1136$ 	& $0.2426$	& $0.3909$	& $0.5623$	& $0.7622$\\
 & sphere packing 				& $13$ 	& $0.1133$ 	& $0.2439$	& $0.3939$	& $0.5666$	& $0.7663$\\
 & tesseract ($\color{blue}C_4$) 	& $16$  	& $0.1255$ 	& $0.2630$	& $0.4150$	& $0.5850$	& $0.7776$\\
 & runcinated $5$-cell ($\color{orange}A_4$) & $20$ & $0.1216$ & $0.2586$ & $0.4128$ & $0.5855$ & $0.7795$ \\
 & octacube ($\color{purple}D_4$)			& $24$ 	& $0.1202$ 	& $0.2577$	& $0.4140$	& $0.5877$	& $0.7823$ \\
\midrule
$5$ & \emph{spherical caps} 	& 		& $\emph{0.0997}^{(13)}$ & $\emph{0.2203}^{(10)}$ & $\emph{0.3630}^{(8)}$ & $\emph{0.5342}^{(7)}$ & $\emph{0.7415}^{(6)}$ \\
\midsep
 & $5$-simplex ($\color{red}S_5$) 	& $6$ & $0.1105$	& $0.2354$	& $0.3785$	& $0.5469$	& $0.7493$  \\
 & $5$-orthoplex ($\color{green!50!black}O_5$)& $10$ 	& $\mathbf{0.1076}$	& $\mathbf{0.2299}$	& $\mathbf{0.3733}$	& $\mathbf{0.5433}$	& $\mathbf{0.7483}$  \\
 & $1_{21}$-polytope & $16$  & $0.1080$ & $0.2330$ & $0.3794$ & $0.5516$ & $0.7554$ \\
 & expanded $5$-simplex ($\color{orange}A_5$) & $30$  & $0.1167$ & $0.2494$ & $0.4007$ & $0.5735$ & $0.7713$ \\
 & $5$-hypercube ($\color{blue}C_5$) 	& $32$ 	& $0.1255$	& $0.2630$	& $0.4150$	& $0.5850$	& $0.7776$  \\
 & rectified $5$-orthoplex ($\color{purple}D_5$) & $40$  & $0.1139$ & $0.2471$ & $0.4009$ & $0.5757$ & $0.7739$ \\ 
\midrule
$6$ & \emph{spherical caps} 	& 		& $\emph{0.0955}^{(18)}$ 	& $\emph{0.2135}^{(15)}$ 	& $\emph{0.3552}^{(11)}$ 	& $\emph{0.5263}^{(9)}$ 	& $\emph{0.7357}^{(8)}$ \\
\midsep
 & $6$-simplex ($\color{red}S_6$) 	& $7$ 	& $0.1089$	& $0.2319$	& $0.3742$	& $0.5422$	& $0.7458$ \\
 & $6$-orthoplex ($\color{green!50!black}O_6$)& $12$ 	& $0.1065$	& $0.2260$	& $\mathbf{0.3670}$	& $\mathbf{0.5361}$ 	& $\mathbf{0.7431}$ \\
 & $2_{21}$-polytope & $27$ &  $\mathbf{0.1038}$ & $\mathbf{0.2258}$ & $0.3712$ & $0.5442$ & $0.7510$ \\ 
 & $1_{31}$-polytope & $32$ & $0.1062$ & $0.2314$ & $0.3788$ & $0.5520$ & $0.7564$ \\
 & expanded $6$-simplex ($\color{orange}A_6$) & $42$ & $0.1133$ & $0.2427$ & $0.3917$ & $0.5642$ & $0.7647$ \\
 & rectified $6$-orthoplex ($\color{purple}D_6$) & $60$  & $0.1107$ & $0.2404$ & $0.3915$ & $0.5661$ & $0.7673$ \\ 
 & hypercube ($\color{blue}C_6$) 	& $64$	& $0.1255$	& $0.2630$	& $0.4150$	& $0.5850$	& $0.7776$ \\
\midrule[2pt]
$d$ & \emph{lower bound} 		& 	 	& $\emph{0.0173}$	& $\emph{0.0718}$	& $\emph{0.1716}$	& $\emph{0.3333}$	& $\emph{0.5888}$     \\
 & simplex ($\color{red}S_d$) & $d+1$ 	& $\mathbf{0.0173}$	& $\mathbf{0.0718}$	& $\mathbf{0.1716}$	& $\mathbf{0.3333}$	& $\mathbf{0.5888}$     \\
 & orthoplex ($\color{green!50!black}O_d$) & $2d$ & $\mathbf{0.0173}$	& $\mathbf{0.0718}$	& $\mathbf{0.1716}$	& $\mathbf{0.3333}$	& $\mathbf{0.5888}$     \\
 & expanded simplex ($\color{orange}A_d$)	& $d(d+1)$ 	& $\mathbf{0.0173}$	& $\mathbf{0.0718}$	& $\mathbf{0.1716}$	& $\mathbf{0.3333}$	& $\mathbf{0.5888}$     \\
 & rectified orthoplex	($\color{purple}D_d$) & $2d(d-1)$ &	$\mathbf{0.0173}$	& $\mathbf{0.0718}$	& $\mathbf{0.1716}$	& $\mathbf{0.3333}$	& $\mathbf{0.5888}$     \\
 & hypercube ($\color{blue}C_d$; $\mat{A}$ orth.) & $2^d$ & $0.0799$	& $0.1800$	& $0.3151$	& $0.5201$	& $0.7686$     \\
 & hypercube ($\color{blue}C_d$; $\mat{A}$ Gaussian) & $2^d$ & $0.1255$	& $0.2630$	& $0.4150$	& $0.5850$	& $0.7776$     \\
 \bottomrule[2pt]
\end{tabular}
\vspace{0.1cm}
\caption{Parameters $\rho$ for various spherical codes. Sphere packings are from~\cite{sloanepackings}. Rows labeled ``spherical caps'' are lower bounds (Theorem~\ref{thm:lb}). Superscripts refer to the number of caps minimizing $\rho_{\theta}$. Bold values indicate the best values $\rho$ encountered up to this dimension. Results for $k = 1, 2, d$, were obtained analytically; for $k = 3, 4, 5, 6$ most results were obtained through numerical integration and Monte Carlo simulation. \label{tab:overview}}
\end{table}


\begin{figure}[p]
\includegraphics[width=13.5cm]{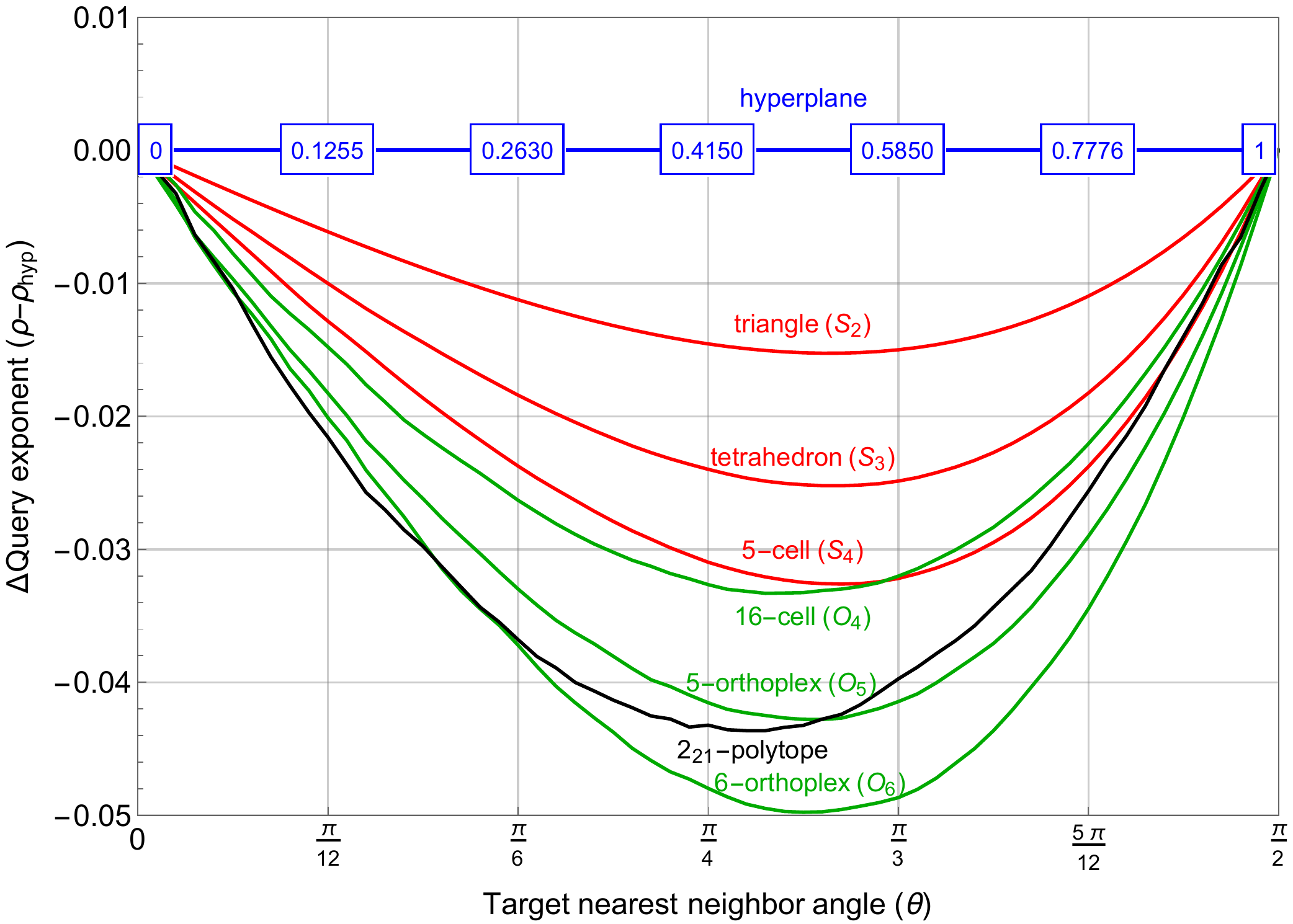}
\caption{Comparison of the improvements in the query exponent $\rho$ over hyperplane hashing~\cite{charikar02} with query exponents $\rho_{\text{hyp}}$. The horizontal blue line denotes the baseline hyperplane hashing approach, with $\rho - \rho_{\text{hyp}} \equiv 0$. Lower curves denote improvements to $\rho_{\text{hyp}}$ using various spherical codes. The measurements at the five vertical gridlines correspond to the values in Table~\ref{tab:overview}. For instance, at $\theta = \pi/6$ the triangle has query exponent approximately $0.011$ lower than hyperplane hashing (which has exponent $0.2630$), leading to $\rho \approx 0.2520$; from Table~\ref{tab:overview} we get the more precise estimate $\rho \approx 0.2518$; while the theory in the full version allows us to compute $\rho$ exactly through a closed-form expression. The polytopes in this figure are those achieving the lowest exponents $\rho$ in their respective dimensions at one of these grid lines, as highlighted in boldface in Table~\ref{tab:overview}.\label{fig:comp}}
\end{figure}


\begin{figure}[p]
\includegraphics[width = 6.8cm]{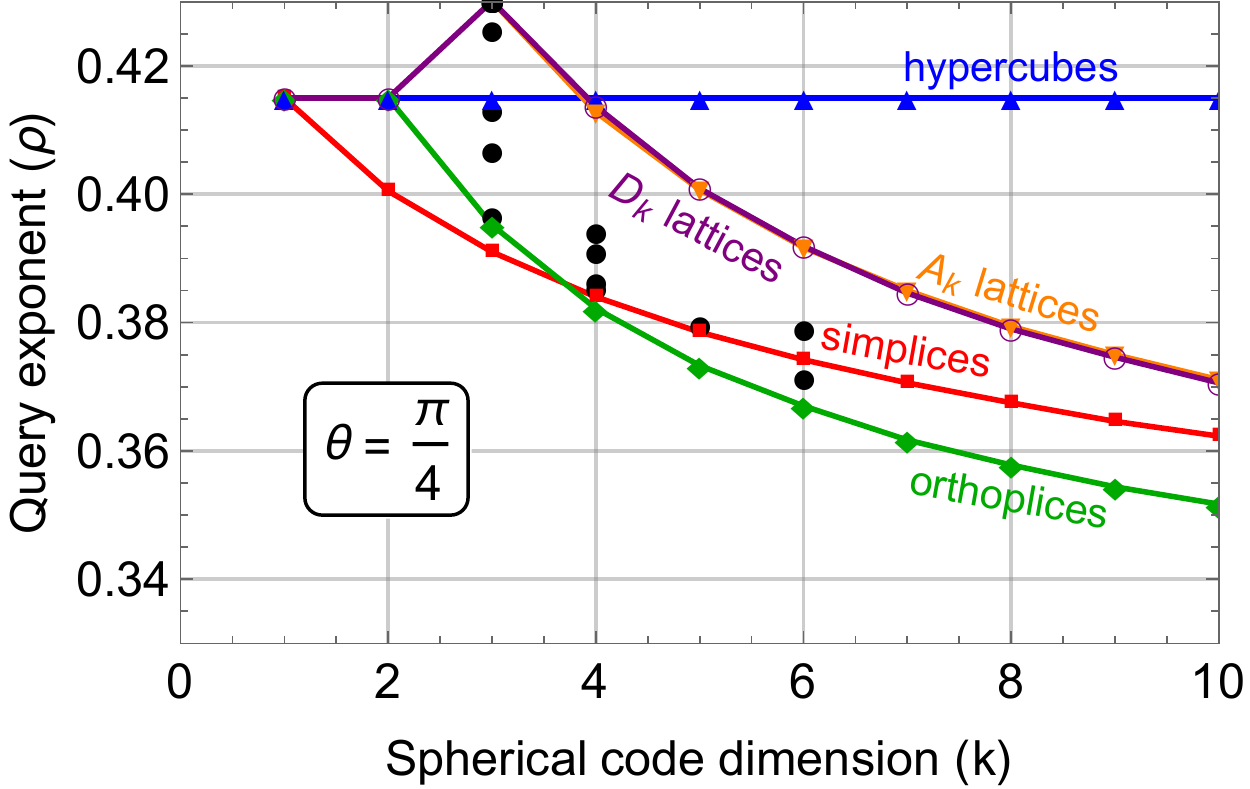} \ \includegraphics[width = 6.8cm]{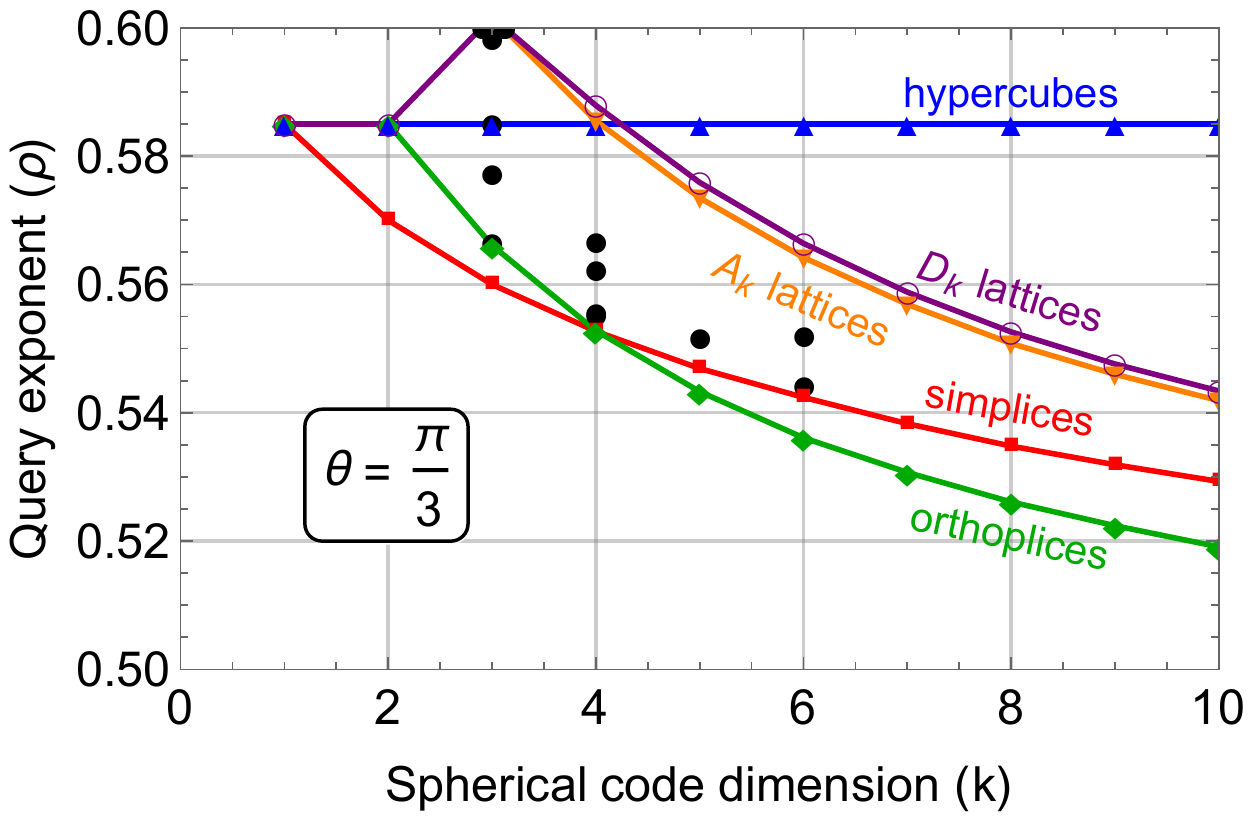}
\caption{Query exponents $\rho$ for project--and--partition hash families based on the spherical codes listed in Table~\ref{tab:overview}. Lower query exponents $\rho$ are generally better, and for these hash families the best exponents $\rho$ are achieved by the simplex for $k \leq 3$ and by the orthoplex for $k \geq 5$. For $k = 4$ the simplex and orthoplex are close, and which of the two achieves a better performance depends on the target nearest neighbor angle $\theta$. \label{fig:high1}}
\end{figure}


\begin{figure}
\includegraphics[width=13.8cm]{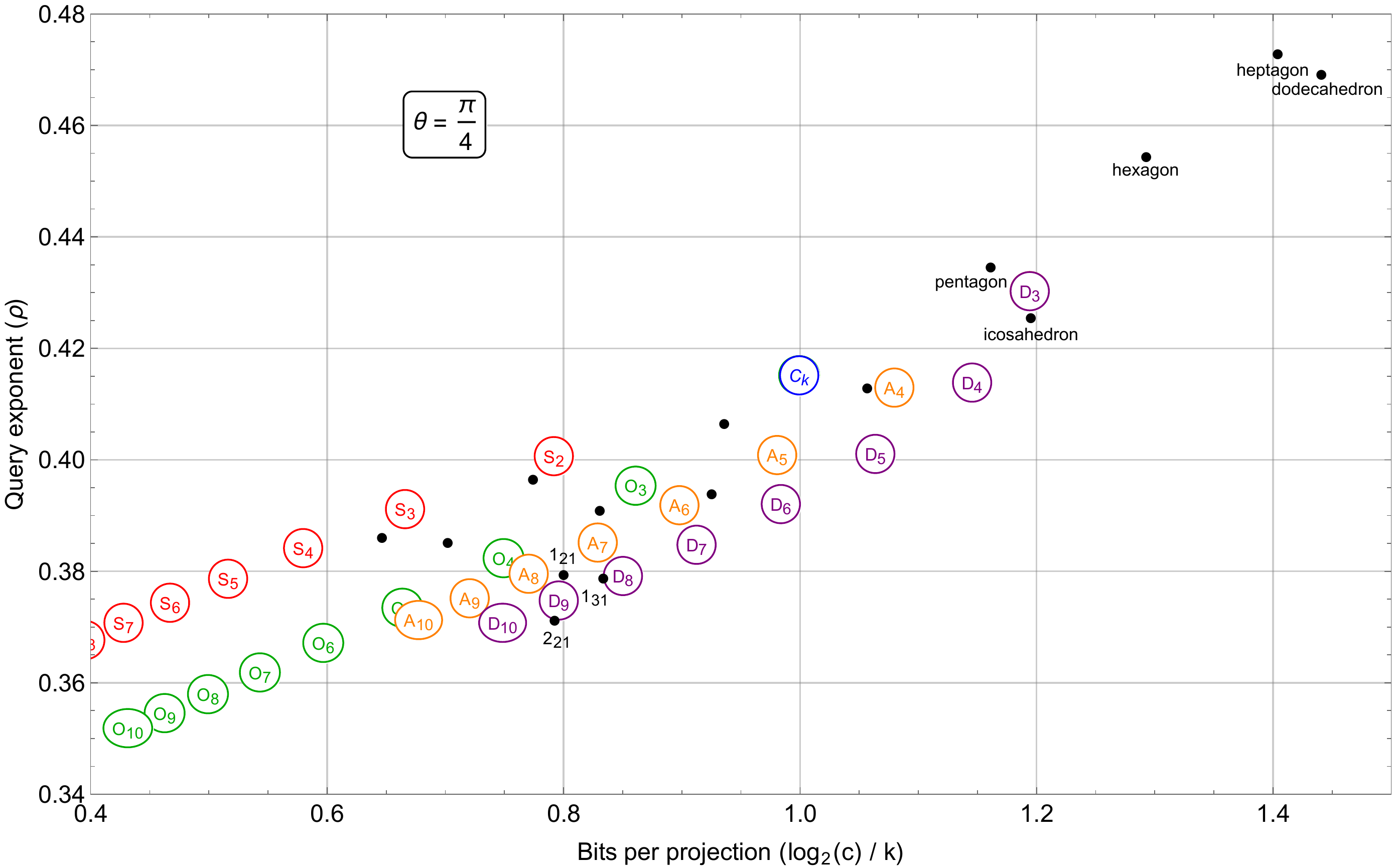} \\
\includegraphics[width=13.8cm]{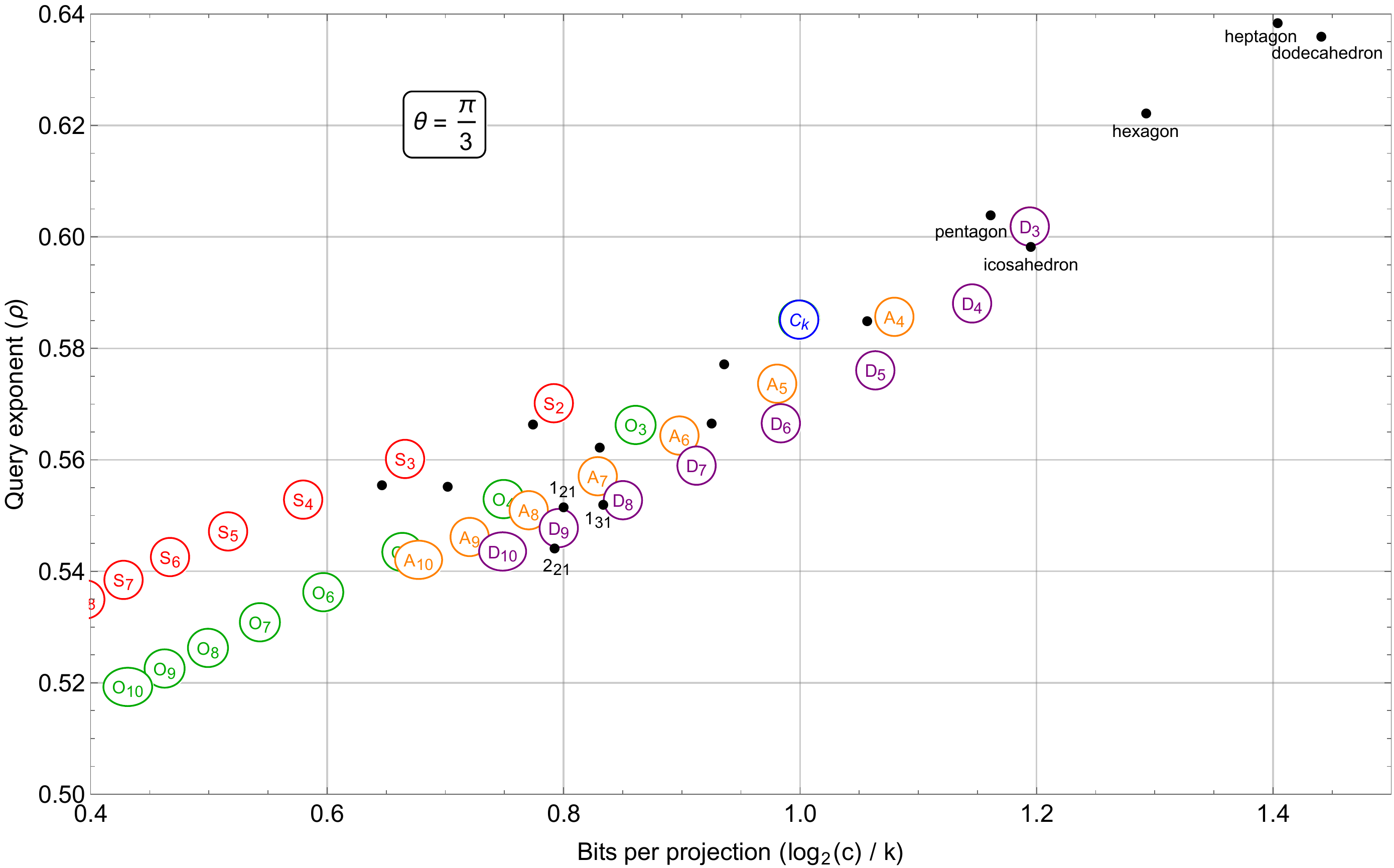}
\caption{A comparison between different spherical codes in terms of the query exponent $\rho$ (vertical axis) and the bits of information extracted from each row of the projection matrix $\mat{A}$ (horizontal axis). The top figure corresponds to target angle $\theta = \frac{\pi}{4}$ (or approximation factor $c = \sqrt{2 + \sqrt{2}}$), the bottom figure to target angle $\theta = \frac{\pi}{3}$ (or $c = \sqrt{2}$). Codes further down generate hash functions with a more discriminative power (a lower value $\rho$), while codes further to the right extract more bits per projection, therefore requiring fewer inner product computations to compute hash values. So for our purposes, the best codes would be as far to the right and as far down as possible. Hypercubes $C_k$ all achieve the same value $\rho$ and the same value $\log_2(c) / k = 1$.\label{fig:to}}
\end{figure}


\begin{figure}
\includegraphics[width=6.8cm]{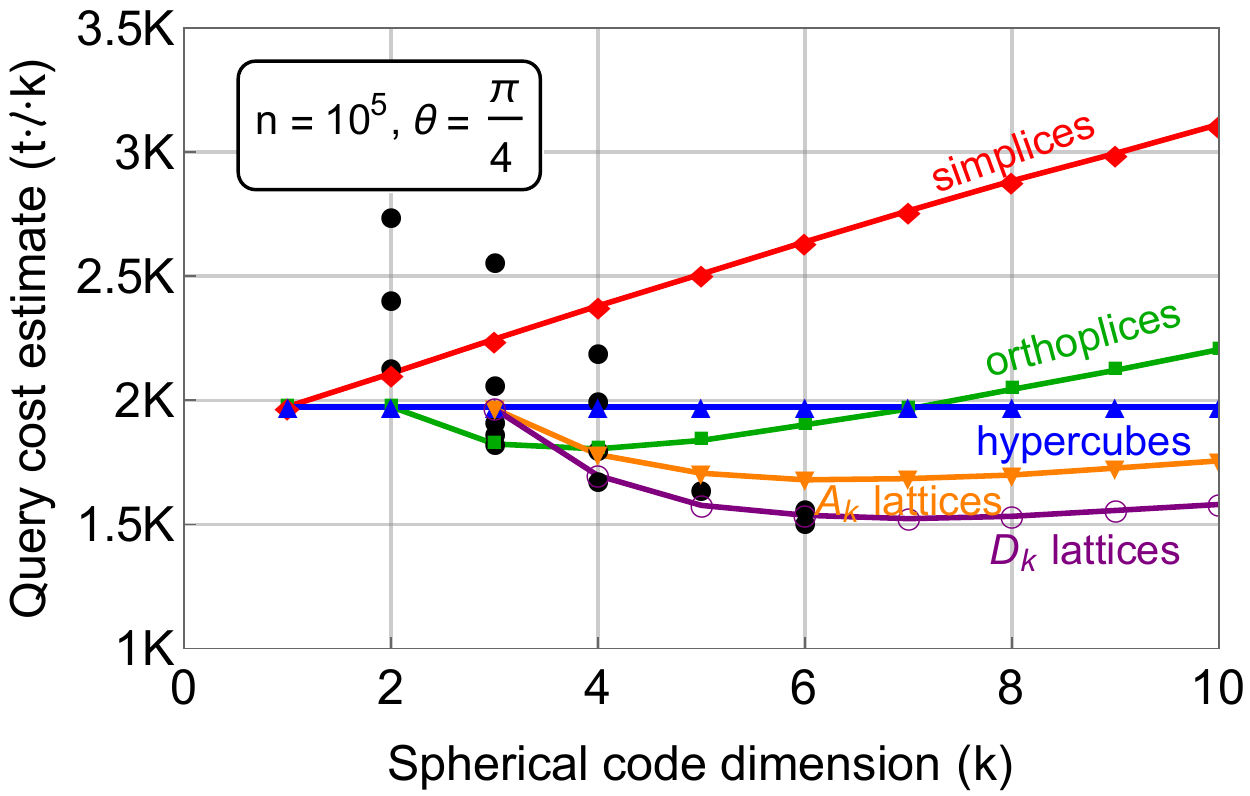} \ \includegraphics[width=6.8cm]{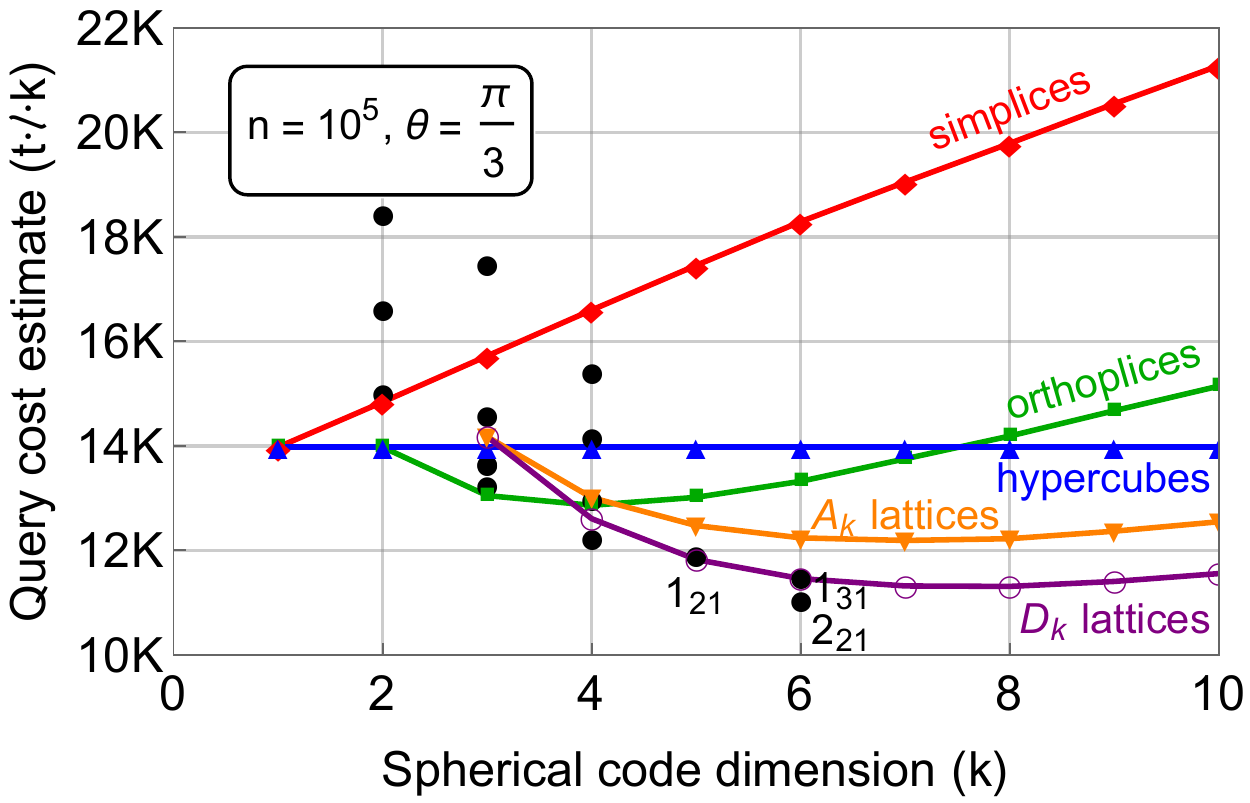} \\
\includegraphics[width=6.8cm]{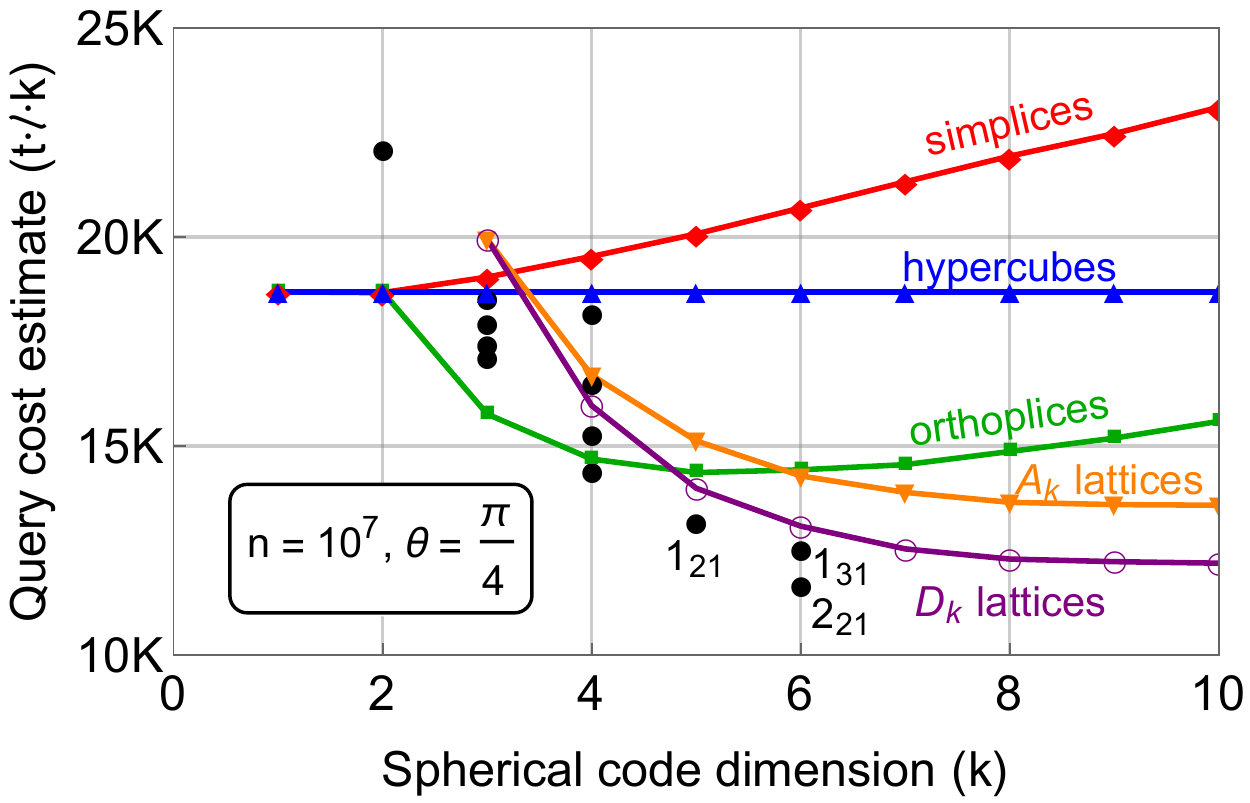} \ \includegraphics[width=6.8cm]{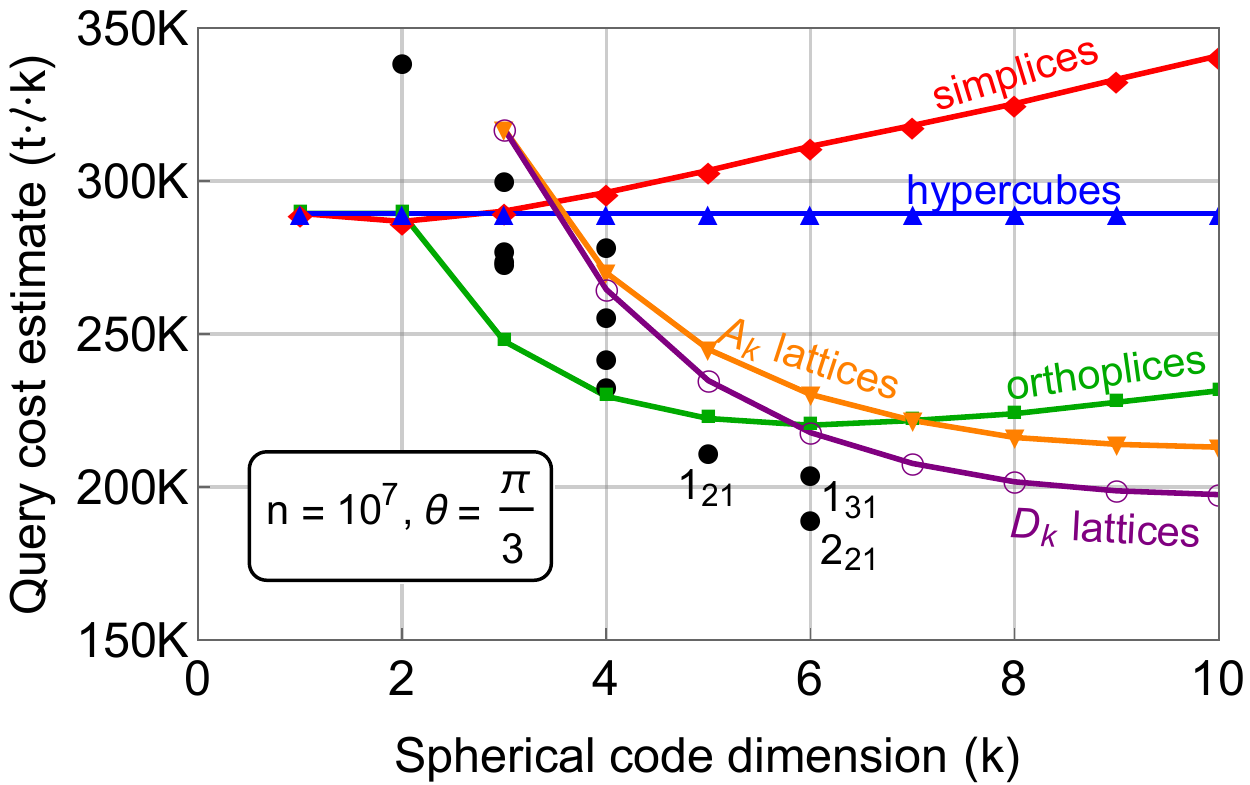} \\
\includegraphics[width=6.8cm]{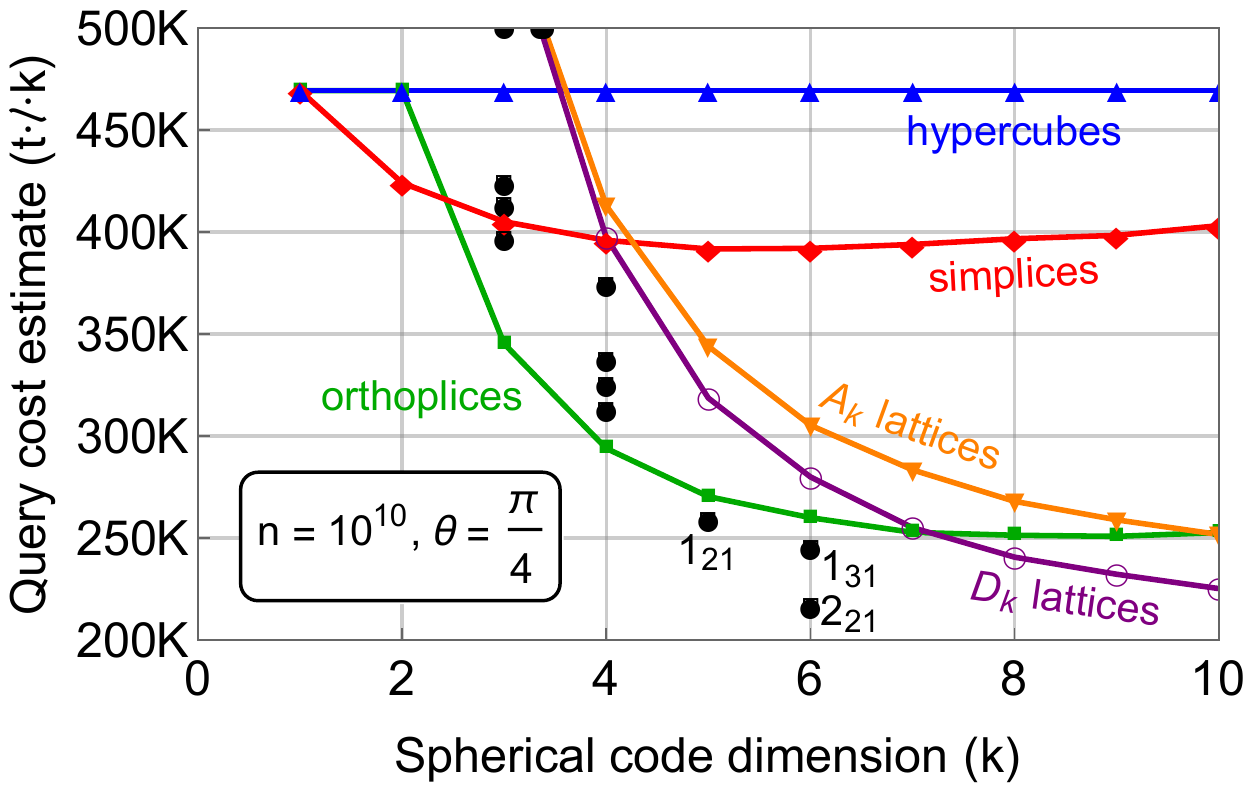} \ \includegraphics[width=6.8cm]{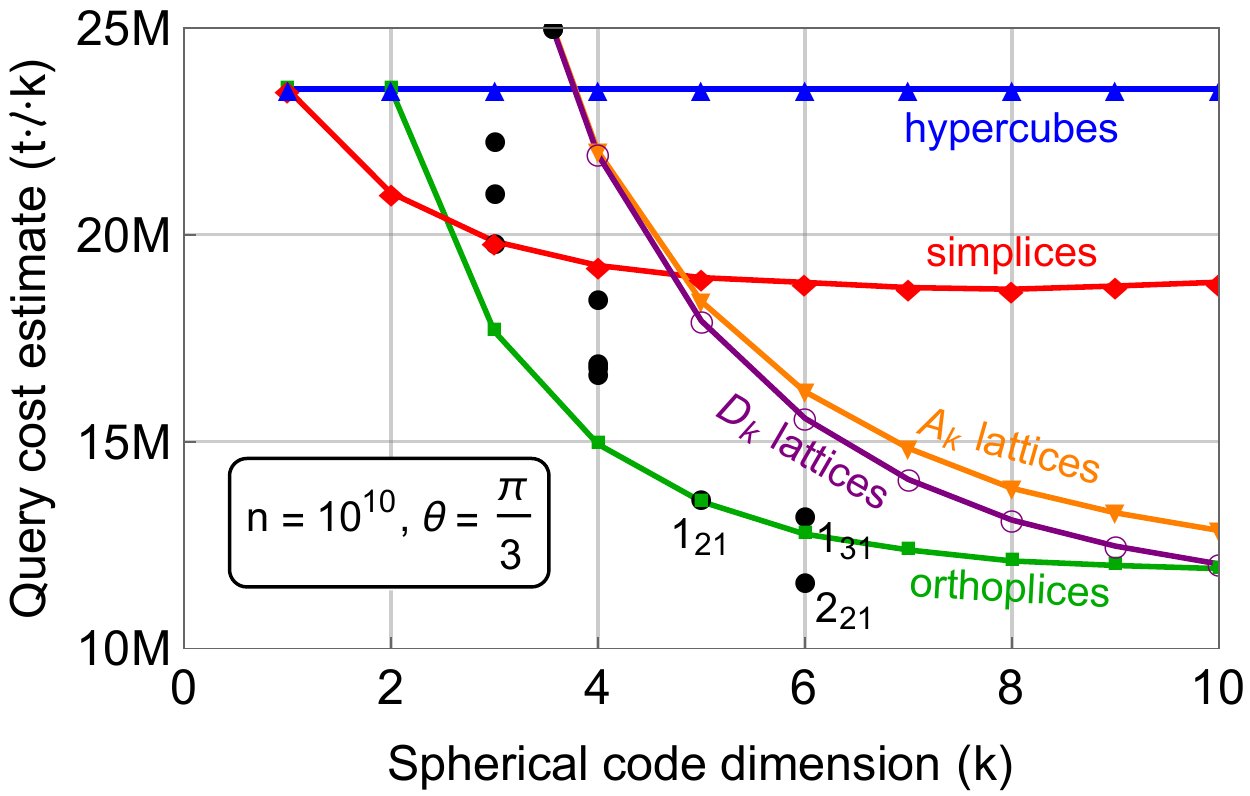}
\caption{Comparison of the query cost estimates $t \cdot \ell \cdot k$ for different parameters $n \in \{10^5, 10^7, 10^{10}\}$ and $\theta \in \{\frac{\pi}{4}, \frac{\pi}{3}\}$, when using $t = n^{\rho}$ hash tables and a hash length $\ell = \log(n) / \log(1/p_2)$. The curves correspond to the regular convex polytope families and the root lattice families $A_k$ and $D_k$, while single black points for $k \leq 6$ correspond to spherical codes described in Table~\ref{tab:overview}. As $n$ increases, the role of a smaller $\rho$ becomes bigger, and so larger spherical codes with smaller values $\rho$ become more suitable choices than those with bigger values $\rho$ but with lower hash costs. The cost of the projections increases linearly with $k$, while $\rho$ decreases rather slowly with $k$, suggesting that for each $n$ and $\theta$ there is an optimal spherical code dimension $k \ll d$ to project down to, and an optimal spherical code in this dimension to use. Note that other spherical codes than those from the five families of codes sometimes achieve better query cost estimates; in particular, the $2_{21}$-polytope (which is connected to the $E_6$ lattice) might be useful in practice as well, and we cannot rule out the existence of other exotic spherical codes with good properties for nearest neighbor searching. \protect\\ One of the conclusions one might draw from these figures is that indeed orthoplices (cross-polytopes) perform very well~\cite{andoni15cp}, but depending on the parameters it may be better to use the $D_k$ lattices instead -- the trends suggest that as $k$ increases further, the query costs of using the $D_k$ lattices will be smaller than those of the orthoplices.\label{fig:prac}}
\end{figure}


\section{Preliminaries}
\label{sec:preliminaries}

\subparagraph{Basic notation.} Throughout the paper, we denote vectors in lower-case boldface (e.g.\ $\vc{x}$), and matrices in upright, upper-case boldface (e.g. $\mat{A}$). Coordinates of vectors and matrices are written without boldface (e.g. $x_i$ and $\mathord{\mathrm{A}}_{i,j}$). We write $\vc{x} \geq \vc{y}$ to denote the event that $x_i \geq y_i$ for all $i$. We write $\vc{e}_i$ for the $i$th unit vector with a $1$ on position $i$, and a $0$ in all other coordinates. We write $\mat{I}_k$ for the $k \times k$ identity matrix, and $\mat{J}_k$ for the $k \times k$ matrix with all entries equal to $1$. We write $\mat{K} = \mat{A} \otimes \mat{B}$ for the standard Kronecker (tensor) product of matrices; if $\mat{A} \in \mathbb{R}^{r_a \times c_a}$ and $\mat{B} \in \mathbb{R}^{r_b \times c_b}$, then $\mat{K} \in \mathbb{R}^{(r_a r_b) \times (c_a c_b)}$ is a matrix with entries $\mathord{\mathrm{A}}_{i,j} \mathord{\mathrm{B}}_{k,l}$. We write $\|\vc{x}\| = (\sum_i x_i^2)^{1/2}$ for the Euclidean norm of $\vc{x}$, and we denote the standard dot product in $\mathbb{R}^k$ by $\langle \vc{x}, \vc{y} \rangle = \sum_i x_i y_i$. Given two vectors $\vc{x}, \vc{y} \in \mathbb{R}^k$, we further denote their mutual angle by 
\begin{align}
\phi(\vc{x}, \vc{y}) = \arccos \left\langle \frac{\vc{x}}{\|\vc{x}\|}, \frac{\vc{y}}{\|\vc{y}\|} \right\rangle.
\end{align}
Angles between vectors are invariant under scalar multiplication of either vector (i.e.\ $\phi(\lambda \vc{x}, \mu \vc{y}) = \phi(\vc{x}, \vc{y})$ for all $\lambda, \mu > 0$), and if $\mat{A}$ is a rotation matrix ($\mat{A}$ is orthogonal), then $\phi(\mat{A} \vc{x}, \mat{A} \vc{y}) = \phi(\vc{x}, \vc{y})$. 

\subparagraph{Geometry on the sphere.} We write $\cS^{k-1} = \{\vc{x} \in \mathbb{R}^k: \|\vc{x}\| = 1\}$ for the Euclidean unit sphere in $k$ dimensions. We write $\mu$ for the standard Haar measure, so that the relative volume of $\mathcal{A} \subseteq \cS^{k-1}$ on the sphere can be expressed as $\mu(\mathcal{A}) / \mu(\cS^{k-1})$. Let $\cC_{\alpha} = \{\vc{x} \in \cS^{k-1}: x_1 > \alpha\}$ denote a \emph{spherical cap} of height $1 - \alpha$. Writing $I_x(a,b)$ for the regularized incomplete beta function, we can compute the volume of $\cC_{\alpha}$ as~\cite{lee14}:
\begin{align}
\frac{\mu(\cC_{\alpha})}{\mu(\cS^{k-1})} = \frac{1}{2} \, I_{1 - \alpha^2}\left(\frac{k-1}{2}, \frac{1}{2}\right).
\end{align}
Asymptotically (for large $k$), the above ratio scales as $k^{\Theta(1)} \cdot (1 - \alpha^2)^{k/2}$~\cite{shannon59, becker16lsf, andoni17}.

\subparagraph{Spherical codes and Voronoi cells.} Given a set of points (code words) $\cC = \{\vc{c}_1, \dots, \vc{c}_c\} \subset \mathbb{R}^k$, the \emph{Voronoi cells} of $\cC$ are subsets of $\mathbb{R}^k$, associating to each point $\vc{c}_i \in \cC$ a Voronoi region 
\begin{align}
\cV_i = \cV_i(\cC) = \left\{\vc{x} \in \mathbb{R}^k: \|\vc{x} - \vc{c}_i\| \leq \|\vc{x} - \vc{c}_j\|, \, \forall j = 1, \dots, c\right\}. 
\end{align}
Note that $\bigcup_i \cV_i = \mathbb{R}^k$, while $\sum_{i \neq j} \mu(\cV_i \cap \cV_j) / \mu(\cS^{k-1}) = 0$; except for the boundaries between cells, which together have measure $0$, this forms a proper partition of $\mathbb{R}^k$. For a partition of a space into Voronoi cells induced by $\cC$, we write $\cR_i$ for the set of \emph{relevant vectors} of $\vc{c}_i \in \cC$, i.e.\ the vectors $\vc{c}_j \in \cC$ such that $\cV_i, \cV_j$ share a non-trivial boundary. If $\cC \subset \cS^{k-1}$, we call $\cC$ a \emph{spherical code}, and in that case the Voronoi cells take conical shapes, and are invariant under (positive) scalar multiplication; $\vc{x} \in \cV_i$ iff $\vc{x} / \|\vc{x}\| \in \cV_i$. Each pair of cells $\cV_i, \cV_j$ then shares a common point $\vc{0}$; for $\vc{c}_j$ to be considered a relevant vector of $\vc{c}_i$, the intersection between $\cV_i, \cV_j$ must contain more than only $\vc{0}$. We further say a spherical code $\cC$ is \emph{uniform} if the associated Voronoi partition satisfies $\mu(\cV_i) = \mu(\cV_j)$ for all $i, j$, and we say $\cC$ is \emph{isogonal} (vertex-transitive) if, for each $\vc{c}_i, \vc{c}_j \in \cC$, there exists an isometry $f$ with $f(\vc{c}_i) = f(\vc{c}_j)$ and $f(\cC) = \cC$.

\subparagraph{Lattices and root systems.} Given a set $\mat{B} = \{\vc{b}_1, \dots, \vc{b}_k\}$ of linearly independent vectors, we define the lattice $\cL = \cL(\mat{B}) = \{\sum_{i=1}^k \lambda_i \vc{b}_i: \vc{\lambda} \in \mathbb{Z}^k\}$. A lattice $\cL$ is an \emph{integral lattice} if for all $\vc{v}, \vc{w} \in \cL$ we have $\langle \vc{v}, \vc{w} \rangle \in \mathbb{Z}$. A lattice $\cL$ is a \emph{root lattice} if it is integral and generated by a set of roots $\vc{a} \in \cL$ satisfying $\|\vc{a}\|^2 = 2$. Besides the three exceptional lattices $E_6, E_7$ and $E_8$, the only irreducible root lattices are those from the following two families~\cite{brouwer02}:
\begin{itemize}
\item $A_k$: Described as a subset of $\mathbb{R}^{k+1}$, this lattice has $k (k+1)$ roots $\vc{e}_i - \vc{e}_j$ for $i \neq j$.\footnote{As all vertices lie on a hyperplane, the vertices can be projected onto $\mathbb{R}^k$ preserving mutual distances.}
\item $D_k$: This lattice has roots $\pm \vc{e}_i \pm \vc{e}_j$, for $i \neq j$. There are $2 k (k - 1)$ such roots.
\end{itemize}
Note that the roots of the root systems $B_k$ and $C_k$ are of different norms, and these sets of roots therefore do not define proper spherical codes.

\subparagraph{Regular polytopes.} Throughout the paper we will study spherical codes generated by \emph{regular polytopes}, i.e.\ polytopes with regular facets and regular vertex figures. Besides several regular polytopes in low dimensions, there are only three families of regular polytopes in high dimensions.
\begin{itemize}
\item $S_k$: Described as a subset of $\mathbb{R}^{k+1}$, the $k$-simplex has vertices $\vc{e}_i$ for $i = 1, \dots, k+1$. \footnotemark[1]
\item $O_k$: The $k$-orthoplex or $k$-cross polytope has $2k$ vertices of the form $\pm \vc{e}_i$, for $i = 1, \dots, k$. 
\item $C_k$: The $k$-hypercube has $2^k$ vertices $\pm \vc{e}_1 \pm \dots \pm \vc{e}_k$, taking all possible signs.
\end{itemize}
Besides these infinite families, there are several regular polytopes in dimension up to four, such as the Platonic and Archimedian solids.

\subparagraph{Probabilities.} Given a set $\mathcal{A}$, we write $a \sim \mathcal{A}$ to denote the process of drawing $a$ uniformly at random from $\mathcal{A}$. With slight abuse of notation, we similarly write $a \sim \xi$ to denote that $a$ is drawn from a probability distribution $\xi$; from the context it will be clear which case we are referring to. We write $\No(\mu, \sigma^2)$ for the (univariate) normal distribution with mean $\mu$ and variance $\sigma^2$, and we write $\No(\vc{\mu}, \mat{\Sigma})$ for the multivariate normal distribution with mean $\vc{\mu}$ and correlation matrix $\mat{\Sigma}$. We refer to $\No(0,1)$ and $\No(\vc{0}, \mat{I}_k)$ as the univariate and multivariate standard normal distributions.





\section{A framework for nearest neighbor searching with spherical codes}
\label{sec:framework}

In this section we will describe a framework for solving a natural average-case version of the nearest neighbor problem with well-chosen spherical codes, and we will show how the performance of these codes can be estimated and compared, to find the best spherical codes. In subsequent sections we will then apply this framework to different spherical codes.

\subsection{The nearest neighbor problem}

Let us start with a broad definition of the nearest neighbor problem, the topic of this paper.
\begin{definition}[Nearest neighbor problem]
Given a data set $\cD = \{\vc{p}_1, \dots, \vc{p}_n\} \subset \mathbb{R}^d$, the nearest neighbor problem asks to index $\cD$ in a data structure such that, when later given a query vector $\vc{q} \in \mathbb{R}^d$, one can quickly find a nearest neighbor $\vc{p}^* \in \cD$ to $\vc{q}$. 
\end{definition}
Depending on the similarity measure, which defines when points are considered nearby, and depending on how the above problem is modeled, different solutions exist. In this paper we will be concerned with the nearest neighbor problem for the angular distance, where the similarity between two vectors is measured by their angle:
\begin{align}
\text{dist}(\vc{x}, \vc{y}) = \phi(\vc{x}, \vc{y}) = \arccos \left\langle \frac{\vc{x}}{\|\vc{x}\|}, \frac{\vc{y}}{\|\vc{y}\|} \right\rangle.
\end{align}
With this notion of similarity, two vectors are considered nearby if they have a small angle, and so for the nearest neighbor problem we are looking for the point $\vc{p}^* \in \cD$ which has the smallest angle with the query vector $\vc{q}$. As the angle between two vectors is independent of their norms, this can equivalently be viewed as the nearest neighbor problem for the Euclidean distance on the unit sphere ($\cD \subset \cS^{d-1}$) by scaling all vectors to have norm $1$: a small angle then corresponds to a small distance on the unit sphere through the following one-to-one relation between Euclidean distances and angles:
\begin{align}
\vc{x}, \vc{y} \in \cS^{d-1} \qquad \implies \qquad \tfrac{1}{2} \|\vc{x} - \vc{y}\|^2 = 1 - \cos \phi(\vc{x}, \vc{y}).
\end{align}

As described in e.g.\ \cite{andoni15, andoni15cp, andoni16lb, andoni17}, the angular nearest neighbor problem (or the Euclidean nearest neighbor problem on the unit sphere) is in a sense the canonical hard problem; through various reductions, a solution to the nearest neighbor problem for this problem would immediately yield an (efficient) algorithm for the nearest neighbor problem in the $\ell_1$ and $\ell_2$-norms in all of $\mathbb{R}^d$. Studying this problem therefore may also be of interest when considering the nearest neighbor problem for other metrics.

Since it has long been known that solving this problem exactly for worst-case instances cannot be done in sublinear time in the size of the data set $n = |\cD|$~\cite{indyk98, andoni16lb}, a large body of work has focused on \emph{approximation} versions of this problem under \emph{worst-case} assumptions, where an approximate solution $\hat{\vc{p}} \in \cD$ with $\text{dist}(\hat{\vc{p}}, \vc{q}) \leq c \cdot \text{dist}(\vc{p}^*, \vc{q})$ suffices for solving the problem, given some approximation factor $c > 1$ and the exact nearest neighbor $\vc{p}^* \in \cD$ to $\vc{q}$ in $\cD$. Here instead we will focus on methods for solving this problem \emph{exactly} in sublinear time for \emph{average-case} instances, by making suitable, natural assumptions on the distribution of the data points $\vc{p}_i$ and the query vector $\vc{q}$. 
\begin{definition}[Nearest neighbor problem on the sphere] Given a data set $\cD$ of size $n$ drawn uniformly at random from $\cS^{d-1}$, and a parameter $\theta \in (0, \frac{\pi}{2})$, the nearest neighbor problem on the sphere asks to index $\cD$ in a data structure such that, when later given a uniformly random query $\vc{q} \in \cS^{d-1}$, we can quickly find a neighbor $\vc{p}^* \in \cD$ with $\phi(\vc{p}^*, \vc{q}) \leq \theta$, or conclude that with high probability no such vector $\vc{p}^* \in \cD$ exists.
\end{definition}
This description covers many applications where the data set and the query are indeed (close to) uniformly random, and we wish to detect highly similar objects in the database, if they exist. Note that there is often a one-to-one correspondence between such average-case, exact solutions and worst-case, approximate solutions -- we can often substitute $\cD$ with a (worst-case) data set $\cD'$ with the guarantee that no vectors in $\cD' \setminus \{\vc{p}^*\}$ have angle smaller than $\theta' > \theta$ with the query vector, and use the same algorithms in both cases.

\subsection{Locality-sensitive hashing}

To solve the nearest neighbor problem on the unit sphere, we will follow the celebrated \emph{locality-sensitive hashing} framework, first introduced in~\cite{indyk98}. Let us first restate the folklore definition of locality-sensitive hash function, slightly adjusted to fit our model of exact nearest neighbor searching for the angular distance for random data points.
\begin{definition}[Locality-sensitive hash functions]
Let $c \in \mathbb{N}$. A set $\cH$ of hash functions $h: \mathbb{R}^d \to \{1, \dots, c\}$ is called $(\theta, p_1, p_2)$-sensitive if (1) two vectors at an angle at most $\theta$ collide with probability at least $p_1$; and (2) two random vectors collide with probability at most $p_2$:
\begin{align}
\Pr_{\substack{h \sim \cH \\ \vc{x}, \vc{y} \sim \cS^{d-1}}}(h(\vc{x}) = h(\vc{y}) \mid \phi(\vc{x}, \vc{y}) \leq \theta) \geq p_1, \quad \qquad \Pr_{\substack{h \sim \cH \\ \vc{x}, \vc{y} \sim \cS^{d-1}}}(h(\vc{x}) = h(\vc{y})) \leq p_2.
\end{align}
\end{definition}
With slight abuse of notation, we will sometimes use $p_1$ and $p_2$ to refer to the actual probabilities above (rather than their lower and upper bounds, respectively). Now, assuming such hash functions exist, and they are efficiently evaluable\footnote{Being able to compute $h(\vc{x})$ given $\vc{x}$ in $n^{o(1)}$ time may be sufficient to guarantee the stated asymptotic time complexities, but in practice a superpolynomial decoding cost $d^{\omega(1)}$ may already render the method impractical.}, we can build efficient nearest neighbor data structures as follows. After selecting suitable parameters $\ell$ and $t$, we construct $t$ hash tables $T_1, \dots, T_t$, each consisting of $c^{\ell}$ hash buckets, corresponding to the concatenated outputs of $\ell$ randomly chosen hash functions $h_1, \dots, h_k \sim \cH$ from our hash family. For each vector $\vc{p} \in \cD$ and each hash table $T_i$, we compute its combined hash value, and insert $\vc{p}$ in the corresponding hash bucket with label $(h_{i,1}(\vc{p}), \dots, h_{i,\ell}(\vc{p}))$. Then, given a query $\vc{q}$, for each hash table $T_i$ we compute its corresponding hash bucket $(h_{i,1}(\vc{q}), \dots, h_{i,\ell}(\vc{q}))$, look up vectors stored in this hash bucket, and check if any of them have angle less than $\theta$ with $\vc{q}$. After going through all the hash tables, we hope to have encountered the nearest neighbor $\vc{p}^* \in \cD$ in one of these hash tables, while hopefully the total amount of work done at this point (computing hashes, doing hash table look-ups, going through potential nearest neighbors in the hash buckets) is still sublinear in $n$. 

Besides the initialization costs of locality-sensitive hashing (e.g.\ setting up the hash functions, and initializing and populating the hash tables), the query costs can be divided in two main costs: (1) computing $t \cdot \ell$ hash values and doing $t$ lookups in memory to locate the right hash buckets; and (2) going through all potential nearest neighbors encountered in these hash buckets, to find a solution. By choosing $\ell$ and $t$ carefully, as a function of $n, \theta, p_1, p_2$, we can balance these costs and minimize the query costs as stated in the folklore result below.
\begin{proposition}[Locality-sensitive hashing for the nearest neighbor problem]
Let $\cH$ be \\an efficiently evaluable $(\theta, p_1, p_2)$-sensitive hash family, and let
\begin{align}
\rho = \frac{\log 1/p_1}{\log 1/p_2}, \qquad \ell = \frac{\log n}{\log 1/p_2}, \qquad t = O(n^{\rho}). \label{eq:lshpars}
\end{align}
Then, using locality-sensitive hashing with these parameters, we can solve the nearest neighbor problem with the following costs:
\begin{enumerate}
\item Initializing and populating the data structure can be done in time $\tilde{O}(n^{1 + \rho})$; 
\item The data structure requires $\tilde{O}(n^{1 + \rho})$ memory to store; 
\item With prob.\ $\geq 0.99$, within time $\tilde{O}(n^{\rho})$ the algorithm finds a neighbor at angle $\leq \theta$, if it exists.
\end{enumerate}
\end{proposition} 
Here $\tilde{O}$ hides multiplicative factors which are polynomial in $d$ and $\log n$, while the success probability can be fine-tuned through the leading constant of $t$; a larger $t$ leads to a higher success probability.

Note that the memory and query costs scale as $n^{1 + \rho}$ and $n^{\rho}$ respectively, and so the goal is to find a locality-sensitive hash family $\cH$ minimizing the associated parameter $\rho$, for the parameter $\theta$ of interest, preferably with an efficient decoding algorithm. Beyond this point, we will mostly ignore all other aspects of locality-sensitive hashing, and focus on finding families $\cH$ minimizing $\rho$; in Section~\ref{sec:practice} we will come back to these practical aspects of locality-sensitive hashing.

\subsection{Space partitions}

Various locality-sensitive hash families achieving sublinear query times have been proposed over time. Although some are strictly more efficient than others, often there is a trade-off between e.g.\ the different performances for different angles $\theta$ (approximation factors $c$); the asymptotic efficiency for large $d$ and $n$; and the practical performance for small/moderate values $d$ and $n$. Examples of hash families for the angular distance from the literature include:
\begin{itemize}
\item The fast and practical hyperplane-based approach of Charikar~\cite{charikar02};
\item The asymptotically optimal cross-polytope hash family~\cite{terasawa07, eshghi08, terasawa09, andoni15cp, kennedy17};
\item The closely related (asymptotically optimal) simplex hash family~\cite{terasawa07, terasawa09};
\item The hypercube hash family~\cite{terasawa07, terasawa09, laarhoven17hypercube};
\item Hash families based on spherical caps~\cite{andoni06, andoni14, becker16lsf, andoni17};
\item Approaches based on using suitable integer lattices~\cite{andoni06, chandrasekaran18}. 
\end{itemize}
More recently, it has been shown that deviating from the locality-sensitive hashing framework (by relaxing the condition that $h$ induces a proper partition of the space) may lead to better asymptotic results in certain applications~\cite{becker16cp, becker16lsf, andoni17, christiani17}. However, for data sets of subexponential size ($n = 2^{o(d)}$), the best locality-sensitive hash methods achieve the same asymptotic performance as these locality-sensitive filters, and appear to perform better in practice~\cite{andoni15cp, falconn, annbench1, aumueller17, annbench2}. Here we restrict our attention to locality-sensitive hash families generated by proper partitions of the unit sphere.

Out of the above methods, the hyperplane-based approach of Charikar~\cite{charikar02} is often (one of) the most efficient in practice, due to its simple decoding algorithm which can be made extremely fast with heuristic tweaks~\cite{achlioptas01}. For sufficiently large $d$ and $n$, the cross-polytope hash family will perform better, but e.g.\ in applications in cryptography, the cross-over point between hyperplane hashing and ``optimal'' hash families was found to be quite high~\cite{laarhoven15crypto, becker16cp, becker16lsf, mariano17}. Rather than only focusing on asymptotic performance, we will search for hash families with similar properties to the hyperplane-based approach, ideally allowing us to interpolate between the fast hyperplane LSH and the asymptotically superior cross-polytope LSH.

\subsection{Project--and--partition} 

A common approach for locality-sensitive hashing encompasses two steps:
\begin{description}
\item[Project.] Given vectors $\vc{x} \in \mathbb{R}^d$, we first project down onto a $k$-dimensional subspace through a random projection matrix $\mat{A} \in \mathbb{R}^{k \times d}$, to obtain $\mat{A} \vc{x} \in \mathbb{R}^k$. This projection is not lossless, and mutual distances/angles generally cannot be perfectly preserved through this transformation, but ideally these distances are preserved as best as possible. For this, here we generate $\mat{A}$ by choosing each entry independently as a standard normal random variable. Note that with high probability, for constant $k = O(1)$ and large $d$ the $k$ rows of $\mat{A}$ will be approximately orthogonal and have approximately the same norms~\cite{jiang06}. Through the above definition of $\mat{A}$ however, we eliminate the dependence on $d$ in the analysis of the collision probabilities $p_1$ and $p_2$ -- the distribution of any particular column of $\mat{A}$ is independent of $d$. 
\item[Partition.] Using a partition of $\mathbb{R}^k$ into $c$ regions, we assign a hash value to $\vc{x}$ according to the region $\mat{A} \vc{x}$ landed in, using a decoding function $\mathrm{dec}: \mathbb{R}^k \to \{1, \dots, c\}$. A natural choice for the partition is to use the Voronoi cells induced by a suitable code $\cC \subset \mathbb{R}^k$ of size $|\cC| = c$. When all code words $\vc{c} \in \cC$ have the same (unit) length, $\cC$ forms a spherical code, and in the remainder we will exclusively focus on such codes $\cC \subset \cS^{k-1}$.
\end{description} 
Given a spherical code $\cC$, we thus obtain the following hash function family $\cH$.
\begin{definition}[Spherical code hash families] Let $\cC \subset \cS^{k-1}$ be a spherical code of size $|\cC| = c$, and let $\dec: \mathbb{R}^k \to \{1, \dots, c\}$ be a decoding algorithm for the Voronoi cells of $\cC$. Then we associate to $\cC$ hash functions $h = h_{\mat{A}}$ as follows:
\begin{align}
h_{\mat{A}}(\vc{x}) = \dec(\mat{A} \vc{x}). 
\end{align}
For the corresponding hash family $\cH$, sampling $h \sim \cH$ corresponds to sampling $\mat{A} \sim \No(0,1)^{k \times d}$.
\end{definition}
Observe that in all cases the decoding step can be done in time $O(c k)$, by computing all distances between $\mat{A} \vc{x}$ and code words $\vc{c} \in \cC$, but for some codes with additional structure, faster decoding algorithms may exist. Projections can naively be computed in time $O(d k)$, for a total hash time complexity of $O(d k + c k)$.  

As an example, for hyperplane LSH~\cite{charikar02} we have $k = 1$ (i.e.\ we project onto a line), and the matrix $\mat{A} \in \mathbb{R}^{1 \times d}$ is a vector drawn from a spherically-symmetric distribution (e.g.\ a Gaussian distribution). The partition is then defined by the Voronoi cells of the one-dimensional spherical code $\cC = \{-1, 1\}$; the sign of the inner product $\mat{A} \vc{x} \in \mathbb{R}$ determines whether the hash value equals $0$ or $1$.

\subsection{Orthant probabilities}

Without loss of generality, let $\vc{x} = (1, 0, \dots, 0)$ and, in case of a predetermined angle $\theta$ with $\vc{x}$, let $\vc{y} = (\cos \theta, \sin \theta, 0, \dots, 0)$. Denoting the $k$-dimensional columns of $\mat{A}$ by $\vc{a}_i$, observe that $\mat{A} \vc{x} = \vc{a}_1$ and $\mat{A} \vc{y} = (\cos \theta) \vc{a}_1 + (\sin \theta) \vc{a}_2$. Note that while $\vc{x}$ and $\vc{y}$ have angle \textit{exactly} $\theta$, $\mat{A} \vc{x}$ and $\mat{A} \vc{y}$ only have angle \textit{approximately} $\theta$, depending on how orthogonal the two random vectors $\vc{a}_1, \vc{a}_2 \sim \No(0,1)^k$ are. This means that we can compute $p_1$ as follows:
\begin{align}
p_1 &= \Pr_{\substack{h \sim \cH \\ \vc{x}, \vc{y} \sim \cS^{d-1}}}(h(\vc{x}) = h(\vc{y}) \mid \phi(\vc{x}, \vc{y}) = \theta) \\
 &= \sum_{i=1}^c \Pr_{\vc{a}_1, \vc{a}_2 \sim \No(\vc{0},\mat{I}_k)}(\dec(\vc{a}_1) = \dec(\vc{a}_1 \cos \theta + \vc{a}_2 \sin \theta) = \vc{c}_i).
\end{align}
For two independently sampled vectors $\vc{x}, \vc{y} \sim \cS^{d-1}$, also $\mat{A} \vc{x}$ and $\mat{A} \vc{y}$ are independent. Given any spherical code $\cC = \{\vc{c}_1, \dots, \vc{c}_c\} \in \cS^{k-1}$, we can thus describe the probability $p_2$ as follows:
\begin{align}
p_2 &= \Pr_{\substack{h \sim \cH \\ \vc{x}, \vc{y} \sim \cS^{d-1}}}(h(\vc{x}) = h(\vc{y})) = \sum_{i=1}^c \Pr_{\vc{a}_1 \sim \No(\vc{0},\mat{I}_k)}(\dec(\vc{a}_1) = \vc{c}_i)^2 = \sum_{i=1}^c \frac{\mu(\cV_i)^2}{\mu(\cS^{d-1})^2} \, . \label{eq:p2}
\end{align}
Here $\cV_i \subset \mathbb{R}^k$ denotes the Voronoi region associated to $\vc{c}_i$, and $\mu(\cV_i)$ its measure. Note that for a size-$c$ partition of $\mathbb{R}^k$, $p_2$ is smallest if all regions have the same size, in which case $p_2 = 1/c$.

To further study $p_1$, for general codes $\cC$ by definition we have $\dec(\vc{x}) = \vc{c}_i$ if and only if $\|\vc{x} - \vc{c}_i\| \leq \|\vc{x} - \vc{c}_j\|$ for all $j = 1, \dots, c$.\footnote{For simplicity we ignore cases where $\vc{x}$ lies on the boundary between two Voronoi cells and decoding is not unique, as the boundaries together have measure $0$.} Since we assume that $\cC$ is a spherical code, this is equivalent to $\langle \vc{x}, \vc{c}_i \rangle \geq \langle \vc{x}, \vc{c}_j \rangle$ for all $j = 1, \dots, c$, or equivalently $\langle \vc{x}, \vc{c}_i - \vc{c}_j \rangle \geq 0$ for all $j = 1, \dots, c$. The event that $\vc{x} \in \mathbb{R}^k$ is mapped to the hash region corresponding to $\vc{c}_i$ (i.e.\ $\vc{x} \in \cV_i$) can thus equivalently be described by the following matrix inequality:
\begin{align}
\begin{bmatrix}
\text{--- } (\vc{c}_i - \vc{c}_1) \text{ ---} \\
\text{--- } (\vc{c}_i - \vc{c}_2) \text{ ---} \\
\dots \\
\text{--- } (\vc{c}_i - \vc{c}_c) \text{ ---} \\
\end{bmatrix} \cdot \begin{bmatrix} | \\ \vc{x} \\ | \end{bmatrix} \geq \begin{bmatrix} | \\ \vc{0} \\ | \end{bmatrix}.
\end{align}
Let $\mat{C}_i$ denote the matrix on the left hand side above, with rows $\vc{c}_i - \vc{c}_j$ for $j = 1, \dots, c$. Recall that for the event of two $\theta$-correlated random vectors to both be mapped to $\vc{c}_i$, we need $\mat{C}_i \vc{x} \geq \vc{0}$ to hold for both $\vc{x} = \vc{a}_1$ and $\vc{x} = (\cos \theta) \vc{a}_1 + (\sin \theta) \vc{a}_2$, where $\vc{a}_1, \vc{a}_2 \sim \No(\vc{0}, \mat{I}_k)$. The probability $p_1$ can thus be described by a $2k$-dimensional system of inequalities:
\begin{align}
p_1 &= \sum_{i=1}^c \Pr_{\vc{a}_1, \vc{a}_2 \sim \No(\vc{0},\mat{I}_k)}\left(\begin{bmatrix} 
\mat{C}_i & \mat{O} \\ (\cos \theta) \mat{C}_i & (\sin \theta) \mat{C}_i
\end{bmatrix} \begin{bmatrix} \vc{a}_1 \\ \vc{a}_2 \end{bmatrix} \geq \begin{bmatrix} \vc{0} \\ \vc{0} \end{bmatrix}\right) \\
 &= \sum_{i=1}^c \Pr_{\vc{z} \sim \No(\vc{0},\mat{I}_{2k})}\left(\left\{\begin{bmatrix} 1 & 0 \\ \cos \theta & \sin \theta \end{bmatrix} \otimes \mat{C}_i \right\} \vc{z} \geq \vc{0}\right).
\end{align}
Here $\otimes$ denotes the standard Kronecker product for matrices. Note that if $\vc{z} \sim \No(\vc{0}, \mat{I}_{2k})$, then $\vc{z}' = \mat{B} \vc{z} \sim \No(\vc{0}, \mat{B} \mat{B}^T)$. Furthermore, using $(\mat{A} \otimes \mat{B}) \cdot (\mat{A} \otimes \mat{B})^T = (\mat{A} \mat{A}^T) \otimes (\mat{B} \mat{B}^T)$ for arbitrary $\mat{A}$ and $\mat{B}$, we obtain:
\begin{align}
p_1 &= \sum_{i=1}^c \Pr_{\vc{z} \sim \No(\vc{0},\mat{\Sigma}_i)}\left(\vc{z} \geq \vc{0}\right), \qquad \mat{\Sigma}_i = \begin{bmatrix} 1 & \cos \theta \\ \cos \theta & 1 \end{bmatrix} \otimes (\mat{C}_i \mat{C}_i^T).
\end{align}
The probabilities above are also known as \textit{orthant probabilities} for the multivariate normal distribution with covariance matrix $\mat{\Sigma}_i$. For completeness, note that $p_2$ can equivalently be described as an orthant probability, as:
\begin{align}
p_2 &= \sum_{i=1}^c \Pr_{\vc{z} \sim \No(\vc{0},\mat{\Pi}_i)}\left(\vc{z} \geq \vc{0}\right), \qquad \mat{\Pi}_i = \begin{bmatrix} 1 \ &  \ 0 \ \\ 0 \ & \ 1 \quad \end{bmatrix} \otimes \mat{C}_i \mat{C}_i^T, \\
 &= \sum_{i=1}^c \Pr_{\vc{z} \sim \No(\vc{0},\mat{\Theta}_i)}\left(\vc{z} \geq \vc{0}\right)^2, \qquad \mat{\Theta}_i = \mat{C}_i \mat{C}_i^T,
\end{align} 
This description is computationally more difficult to work with than the approach using volumes of the corresponding Voronoi cells, as in Equation~\eqref{eq:p2}. To summarize, the above derivation can be seen as a proof of the following theorem. Here $\dec(\cdot) = \dec_{\cC}(\cdot)$ denotes a decoding function for the corresponding spherical code, in terms of its Voronoi cells $\cV_i$: $\dec(\vc{x}) = i$ if $\vc{x} \in \cV_i$ lies closer to $\vc{c}_i$ than to all other code words $\vc{c}_j \in \cC$.
\begin{theorem}[Spherical code locality-sensitive hashing] \label{thm:main}
Let $\cC = \{\vc{c}_1, \dots, \vc{c}_c\} \subset \cS^{k-1}$ be a $k$-dimensional spherical code, and for $i = 1, \dots, c$, let $\mat{C}_i \in \mathbb{R}^{c \times k}$ denote the matrix whose $j^{\text{th}}$ row is the vector $\vc{c}_i - \vc{c}_j$. Let $\cH$ consist of functions $h = h_{\mat{A}}$ of the form:
\begin{align}
h(\vc{x}) = \dec(\mat{A} \vc{x}) \in \{1, \dots, c\}, \qquad \mat{A} \sim \No(0,1)^{k \times d}.
\end{align} 
Then for any $\theta \in (0, \frac{\pi}{2})$,  $\cH$ is a $(\theta, p_1, p_2)$-sensitive hash family, with:
\begin{align}
p_1 &= \sum_{i=1}^c \Pr_{\vc{z} \sim \No(\vc{0},\mat{\Sigma}_i)}\left(\vc{z} \geq \vc{0}\right), \ \mat{\Sigma}_i = \begin{bmatrix} 1 & \cos \theta \\ \cos \theta & 1 \end{bmatrix} \otimes (\mat{C}_i \mat{C}_i^T), \ p_2 = \sum_{i=1}^c \frac{\mu(\cV_i)^2}{\mu(\cS^{d-1})^2} \, .
\end{align}
\end{theorem}
Note again that, by our choice of $\mat{A}$, the parameters $p_1, p_2$ are \emph{independent} of $d$. For large $d \gg k$, choosing $\mat{A} \sim \No(0,1)^{k \times d}$ or choosing $\mat{A}$ to have orthogonal, unit-norm rows is essentially equivalent~\cite{diaconis87, jiang06}. For smaller $d \approx k$ however, these definitions are not the same, and it may be worthwhile to use orthogonal rows instead, to decrease the distortion of mutual distances by the projection $\mat{A}$. For instance, when using dense $d$-dimensional spherical codes for partitions, it may be impossible for two almost-orthogonal vectors to be in the same hash region, and $p_2$ may quickly approach $0$ while $p_1$ remains large. Although our study may lead to advances and insights in the regime of $k \approx d$ as well, the main topic of interest here however is to first project down onto a $k$-dimensional space with $k \ll d$, and then use a suitable spherical code there for partitions.

\subsection{Special spherical codes}

For special classes of spherical codes, the expressions of Theorem~\ref{thm:main} can be further simplified. We state two specific cases of interest below. First, if $\cC$ is such that the Voronoi cells of the vertices each cover an equal part of the sphere, i.e.\ $\mu(\cV_i) = \mu(\cV_j)$ for all $i$ and $j$, then we get the following slightly simplified corollary. Note that given a fixed code size $c$, the value $p_2$ is minimized exactly when all Voronoi cells are of equal size. As we want to minimize $\rho = \rho(p_1, p_2)$ which is an increasing function of $p_2$, this suggests that using such codes may well lead to the best results.\footnote{The stated intuition does not provide a formal proof that $\rho$ is minimized when all Voronoi cells have the same volume, as there is also an intricate dependence of $p_1$ (and $\rho$) on the shape of the Voronoi cells. If all Voronoi cells are of equal size, but their shapes are ``bad'', then the resulting code may not be useful in our framework.}
\begin{corollary}[Collision probabilities for uniform spherical codes] \label{thm:equalweight}
Let $\cC = \{\vc{c}_1, \dots, \vc{c}_c\} \subset \cS^{k-1}$ be a $k$-dimensional \emph{uniform} spherical code, in the sense that $\mu(\cV_i) = \mu(\cV_j)$ for all $i, j$. Let $\mat{C}_i$ and $\cH$ be as in Theorem~\ref{thm:main}. Then for any $\theta \in (0, \frac{\pi}{2})$, $\cH$ is a $(\theta, p_1, p_2)$-sensitive hash family, with
\begin{align}
p_1 &= \sum_{i=1}^c \Pr_{\vc{z} \sim \No(\vc{0},\mat{\Sigma}_i)}\left(\vc{z} \geq \vc{0}\right), \qquad \mat{\Sigma}_i = \begin{bmatrix} 1 & \cos \theta \\ \cos \theta & 1 \end{bmatrix} \otimes (\mat{C}_i \mat{C}_i^T), \qquad p_2 = 1/c.
\end{align}
\end{corollary}
In case $\cC$ has many symmetries, and in particular is \textit{isogonal} (vertex-transitive), then Theorem~\ref{thm:main} can be further simplified, as in that case all hash regions are equivalent, and only one of the orthant probabilities needs to be evaluated to determine $p_1, p_2$ and $\rho$. Note that by definition, isogonal codes are a subset of the uniform spherical codes above.
\begin{corollary}[Collision probabilities for isogonal spherical codes] \label{thm:isogonal}
Let $\cC = \{\vc{c}_1, \dots, \vc{c}_c\} \subset \cS^{k-1}$ be a $k$-dimensional \emph{isogonal} spherical code, in the sense that for each $\vc{c}_i, \vc{c}_j \in \cC$, there exists an isometry $f$ with $f(\vc{c}_i) = \vc{c}_j$ and $f(\cC) = \cC$. Let $\mat{C}_i$ and $\cH$ be as in Theorem~\ref{thm:main}. Then for any $\theta \in (0, \frac{\pi}{2})$, $\cH$ is a $(\theta, p_1, p_2)$-sensitive hash family, with:
\begin{align}
p_1 &= c \cdot \Pr_{\vc{z} \sim \No(\vc{0},\mat{\Sigma}_1)}\left(\vc{z} \geq \vc{0}\right), \qquad \mat{\Sigma}_1 = \begin{bmatrix} 1 & \cos \theta \\ \cos \theta & 1 \end{bmatrix} \otimes (\mat{C}_1 \mat{C}_1^T), \qquad p_2 = 1/c.
\end{align}
\end{corollary}
For isogonal codes, computing the query exponent $\rho$ comes down to computing a single orthant probability.

\subsection{Relevant vectors}

Finally, the computations of the collision probabilities $p_1$ and $p_2$ for a given spherical code can sometimes be simplified by reducing the number of conditions that define when a point is decoded to a given code word. Recall that the matrix $\mat{C}_i$, containing difference vectors $\vc{c}_i - \vc{c}_j$ as its rows, originates from the equivalence below:
\begin{align}
\dec(\vc{x}) = i \ \Longleftrightarrow \ \forall \vc{c}_j \in \cC: \|\vc{x} - \vc{c}_i\| \leq \|\vc{x} - \vc{c}_j\| \ \Longleftrightarrow \ \forall \vc{c}_j \in \cC: \langle \vc{x}, \vc{c}_i - \vc{c}_j \rangle \geq 0. 
\end{align}
In many cases however, we can already conclude that $\dec(\vc{x}) = i$ by only computing some of the distances $\|\vc{x} - \vc{c}_j\|$ -- if $\vc{x}$ is closer to $\vc{c}_i$ than to any of the closest neighbors of $\vc{c}_i$ in $\cC$, then it must be closest to $\vc{c}_i$. More generally, we may define $\cR_i \subseteq \cC$ as the set of relevant vectors for $\vc{c}_i$, i.e.\ those vectors $\vc{c}_j \in \cC$ whose Voronoi cells $\cV_j$ share a non-trivial boundary with $\cV_i$. The conditions to verify that decoding $\vc{x}$ results in $i$ may then be simplified:
\begin{align}
\dec(\vc{x}) = i \ \Longleftrightarrow \ \forall \vc{c}_j \in \cR_i: \|\vc{x} - \vc{c}_i\| \leq \|\vc{x} - \vc{c}_j\| \ \Longleftrightarrow \ \forall \vc{c}_j \in \cR_i: \langle \vc{x}, \vc{c}_i - \vc{c}_j \rangle \geq 0. 
\end{align}
Computationally, for evaluating the orthant probabilities in $p_1$ in Theorem~\ref{thm:main}, this allows us to replace the matrix $\mat{C}_i$ (with $c$ rows $\vc{c}_i - \vc{c}_j$ for $\vc{c}_j \in \cC$), with a smaller matrix $\hat{\mat{C}}_i$ (consisting of $|\cR_i|$ rows $\vc{c}_i - \vc{c}_j$ for $\vc{c}_j \in \cR_i$). This can significantly reduce the dimensionality of the resulting orthant probability problem, and this may assist in evaluating $p_1$ more accurately via numerical integration when no closed form expression is available. 
\begin{proposition}[Computing collision probabilities using relevant vectors] \label{prop:relevant}
Theorem~\ref{thm:main} \\and Corollaries~\ref{thm:equalweight} and \ref{thm:isogonal} remain valid when we replace the matrices $\mat{C}_i \in \mathbb{R}^{c \times k}$ by $\hat{\mat{C}}_i \in \mathbb{R}^{|\cR_i| \times k}$, consisting of the rows $\vc{c}_i - \vc{c}_j$ with $\vc{c}_j \in \cR_i$. 
\end{proposition}
With this general framework in hand, we can now compute collision probabilities for arbitrary spherical codes, and compute the resulting nearest neighbor exponents $\rho$, given $\cC$ and $\theta$.


\section{One-dimensional codes}
\label{sec:1d}

Let us first illustrate the above framework by rederiving the hyperplane locality-sensitive hashing results of Charikar~\cite{charikar02}. For the $1$-dimensional unit sphere $\cS^0 = \{(-1), (1)\}$, there is only one choices for a spherical code $\cC \subseteq \cS^0$ of size at least two, namely $\cC = \{(-1), (1)\}$. This code is an isogonal code, and $\cR_{(1)} = \{(-1)\}$ and $\cR_{(-1)} = \{(1)\}$ are the relevant vectors. As in Proposition~\ref{prop:relevant}, $\hat{\mat{C}}_1$ consists of one row, containing the vector $(1) - (-1) = (2)$. By Corollary~\ref{thm:isogonal} and Proposition~\ref{prop:relevant}, we can compute the collision probabilities, after projecting down using a random vector $\mat{A} \sim \No(0,1)^{1 \times d}$, as:
\begin{align}
p_1 &= 2 \cdot \Pr_{\vc{z} \sim \No(\vc{0},\mat{\Sigma}_1)}\left(\vc{z} \geq \vc{0}\right), \qquad \mat{\Sigma}_1 = \begin{bmatrix} 1 & \cos \theta \\ \cos \theta & 1 \end{bmatrix} \otimes (4), \qquad p_2 = 1/2.
\end{align}
For $p_1$, observe that $\Pr_{\vc{z} \sim \No(\vc{0},\mat{\Sigma}_1)}\left(\vc{z} \geq \vc{0}\right) = \Pr_{\vc{z} \sim \No(\vc{0}, \mu \cdot \mat{\Sigma}_1)}\left(\vc{z} \geq \vc{0}\right)$ for all scalars $\mu$. We can thus eliminate the factor $4$, and we are left with a standard orthant probability for the bivariate normal distribution with covariance $\cos \theta$. The derivation of a closed-form formula for the orthant probability of the bivariate normal distribution is often attributed to Sheppard in the late 1800s~\cite{sheppard99}, who showed that the orthant probability for the bivariate normal distribution with Pearson correlation coefficient $a$ is:
\begin{align}
\mat{\Sigma} = \begin{bmatrix} 1 \ \ & a \\ a \ \ & 1 \end{bmatrix} \qquad \implies \qquad \Pr_{(z_1, z_2) \sim \No(\vc{0},\mat{\Sigma})}\left(z_1 \geq 0, z_2 \geq 0\right) = \frac{\pi - \arccos a}{2 \pi} \, . 
\end{align}
For our application, this implies the well-known result of Charikar~\cite{charikar02}, stated below. 
\begin{proposition}[Collision probabilities for antipodal codes] \label{prop:1d}
Let $\cC$ consist of two antipodal points on the unit sphere. Then the corresponding project--and--partition hash family $\cH$ is $(\theta, p_1, p_2)$-sensitive, with:
\begin{align}
p_1 = 1 - \frac{\theta}{\pi} \, , \qquad p_2 = \frac{1}{2} \, .
\end{align}
\end{proposition}
For the nearest neighbor query exponent $\rho = \log p_1 / \log p_2$, a simple computation leads to e.g.\ $\rho = \log_2(3/2) \approx 0.5850$ for $\theta = \frac{\pi}{3}$. Table~\ref{tab:overview} lists several query exponents $\rho$ for this hash family, for different parameters $\theta \in \frac{\pi}{12} \{1, 2, 3, 4, 5\}$. Note that for $\theta \to 0$ we have $p_1 \to 1$ and $\rho \to 0$, while for $\theta \to \frac{\pi}{2}$ we have $p_1 \to p_2 = \frac{1}{2}$ and $\rho \to 1$, as expected.


\section{Two-dimensional codes}
\label{sec:2d}

The simplest previously unstudied case is to project down onto a two-dimensional subspace through a random projection matrix $\mat{A} \sim \No(0,1)^{2 \times d}$, and to use a suitable spherical code on the two-dimensional sphere (circle) to divide these projected vectors into hash buckets. Using two (antipodal) points on the circle would be equivalent to using the one-dimensional hyperplane hashing described above, so the non-trivial cases start at spherical codes of size at least three. 

\subsection{Two-dimensional isogonal spherical codes}

The following theorem summarizes the collision probabilities for what is arguably the most natural choice for a spherical code on the circle, which is to use any number of equidistributed points on the circle (i.e.\ to use the vertices of a regular polygon to partition the sphere).
\begin{theorem}[Collision probabilities for regular polygons] \label{thm:2d}
Let $\cC \subset \cS^1$ consist of the vertices of the regular $c$-gon, for $c \geq 2$, and let $\theta \in (0, \frac{\pi}{2})$. Then the corresponding project--and--partition hash family $\cH$ is $(\theta, p_1, p_2)$-sensitive, with:
\begin{align}
p_1 &= \frac{1}{c} + c \left(\frac{\pi - \theta}{2 \pi}\right)^2 - c \left(\frac{\arccos(- \cos \theta \cos \frac{2 \pi}{c})}{2 \pi}\right)^2, \qquad \qquad p_2 = 1/c. \label{eq:2d-b}
\end{align}
\end{theorem}
To prove the above result, we will first derive the shape of the correlation matrix below.
\begin{lemma}[The correlation matrix for regular polygons] \label{lem:2d1}
Let $\cC \subset \cS^1$ consist of the vertices of the regular $c$-gon, for $c \geq 2$. Then the correlation matrix $\mat{\Sigma}_1$ in Theorem~\ref{thm:isogonal} (Prop.~\ref{prop:relevant}) satisfies:
\begin{align}
\mat{\Sigma}_1 = 2 (1 - \cos \tfrac{2 \pi}{c}) \cdot \begin{bmatrix} 1 & \cos \theta \\ \cos \theta & 1 \end{bmatrix} \otimes \begin{bmatrix} 1 & -\cos \frac{2 \pi}{c} \\ -\cos \frac{2 \pi}{c} & 1 \end{bmatrix}.
\end{align}
\end{lemma}
\begin{proof}
Without loss of generality\footnote{Note that although the matrices $\mat{C}_i$ depend on the absolute positioning of the code words, the dot products $\mat{C}_i \mat{C}_i^T$ only depends on the relative positions between code words, and is invariant under rotations of $\cC$.}, let us fix one vertex of $\cC$ at $(1, 0)$. For ease of notation, let us write $c_j = \cos \frac{2 j \pi}{c}$ and $s_j = \sin \frac{2 j \pi}{c}$, so that $\cC = \{(c_j, s_j): j = 1, \dots, c\}$. Now, a point in $\mathbb{R}^2$ is closest to $\vc{c}_c = (1,0)$ out of all code words in $\cC$ iff it lies closer to $(1,0)$ than to the relevant vectors of $(1,0)$, which are its left and right neighbors on the circle, $\cR_{(1,0)} = \{(c_1, s_1), (c_{c-1}, s_{c-1})\} = \{(c_1, s_1), (c_1, -s_1)\}$. Let $\hat{\mat{C}}_i$ as in Proposition~\ref{prop:relevant} denote the matrix $\mat{C}_i$ reduced to its relevant rows. For the matrix $\hat{\mat{D}}_c = \hat{\mat{C}}_c \hat{\mat{C}}_c^T \in \mathbb{R}^{2 \times 2}$, we therefore obtain:
\begin{align}
\hat{\mat{D}}_c = \begin{bmatrix} 1 - c_1 \ \ \ & 0 - s_1 \\ 1 - c_1 \ \ \ & 0 + s_1 \end{bmatrix} \cdot \begin{bmatrix} 1 - c_1 \ \ \ & 0 - s_1 \\ 1 - c_1 \ \ \ & 0 + s_1 \end{bmatrix}^T = 2 (1 - \cos \tfrac{2 \pi}{c}) \cdot \begin{bmatrix} 1 & -c_1 \\ -c_1 & 1 \end{bmatrix}.
\end{align}
Due to the spherical/circular symmetry of $\cC$, all matrices $\hat{\mat{D}}_i = \hat{\mat{C}}_i \hat{\mat{C}}_i^T$ are equal to $\hat{\mat{D}}_c$. Together with Theorem~\ref{thm:isogonal}, this leads to the given expression for $\mat{\Sigma} = \mat{\Sigma}_1$. \qed
\end{proof}
Orthant probabilities for general quadrivariate normal distributions are difficult to solve, but for certain specific forms of $\mat{\Sigma}$, explicit formulas for the resulting orthant probabilities in terms of the off-diagonal entries of $\mat{\Sigma}$ are known. In the 1960s, Cheng~\cite{cheng68, cheng69} used a clever path integration technique, together with Plackett's reduction formula~\cite{plackett54}, to ultimately end up with the following expressions for the orthant probabilities for matrices of the above form.
\begin{lemma}[Orthant probabilities for a class of quadrivariate normal distributions] {\cite[Appendix]{cheng68}} \label{lem:2d2}
Let $a, b \in \mathbb{R}$, and let the correlation matrix $\mat{\Sigma} = \mat{\Sigma}_{a,b}$ be of the following form:
\begin{align}
\mat{\Sigma} = \begin{bmatrix} 1 \ \ & a \\ a \ \ & 1 \end{bmatrix} \otimes \begin{bmatrix} 1 \ \ & b \\ b \ \ & 1 \end{bmatrix} = \begin{bmatrix} 1 & a & b & ab \\ a & 1 & ab & b \\ b & ab & 1 & a \\ ab & b & a & 1 \end{bmatrix}.
\end{align}
Then the orthant probability $\Phi(a,b) = \Pr_{\vc{z} \sim \No(\vc{0},\mat{\Sigma})}\left(\vc{z} \geq \vc{0}\right)$ for the correlation matrix $\mat{\Sigma}$ satisfies:
\begin{align}
\Phi(a,b) &= \frac{1}{16} + \frac{1}{4 \pi} \left(\arcsin a + \arcsin b + \arcsin ab\right) + \frac{1}{4 \pi^2} \left(\arcsin^2 a + \arcsin^2 b - \arcsin^2 ab\right).
\end{align}
\end{lemma}
The orthant probability $\Phi(a,b)$ above can be equivalently written in the following form:
\begin{align}
\Phi(a,b) &= \left(\frac{\arccos (-a)}{2 \pi}\right)^2 + \left(\frac{\arccos (-b)}{2 \pi}\right)^2 - \left(\frac{\arccos (a b)}{2 \pi}\right)^2. 
\end{align}
Note that $a, b, ab$ are the off-diagonal entries of $\mat{\Sigma}$, or equivalently, the arguments of the arccosines $-a, -b, ab$ are the off-diagonal entries of $\mat{\Sigma}^{-1}$, thus showing similarities with the orthant probabilities for $2 \times 2$ matrices related to antipodal spherical codes, where $\Phi(a) = \arccos(-a)/(2 \pi)$. An interesting open problem, as we will see later, is to establish a pattern in the orthant probabilities for the matrices $\mat{\Sigma}$ appearing in Theorem~\ref{thm:main} and Corollary~\ref{thm:isogonal}.

Now, with the above lemma in hand, we are ready to prove Theorem~\ref{thm:2d}.
\begin{proof}[Proof of Theorem~\ref{thm:2d}]
From Theorem~\ref{thm:isogonal}, and the observation that regular polygons form isogonal codes on the sphere, we obtain the correlation matrix $\mat{\Sigma}$ as described by Lemma~\ref{lem:2d1}. Recall that orthant probabilities are invariant under scalar multiplication of the random vector $\vc{z}$; removing the leading scalar factor of $\mat{\Sigma}$ in Lemma~\ref{lem:2d1}, the collision probability $p_1$ remains the same. After removing this factor, we obtain a matrix $\mat{\Sigma}$ of the form of Lemma~\ref{lem:2d2}, where $a = \cos \theta$ and $b = - \cos \frac{2 \pi}{c}$. Through some elementary trigonometric reductions, we then obtain the stated result, where we obtain an additional leading factor $c$ as $\Pr(h(\vc{x}) = h(\vc{y})) = c \cdot \Pr(h(\vc{x}) = h(\vc{y}) = 1)$. \qed
\end{proof}
Let us draw some observations from Theorem~\ref{thm:2d}. First, observe that if we substitute $c = 2$, i.e.\ using a regular $2$-gon, we obtain $p_1 = 1 - \theta/\pi$ and $p_2 = 1/2$, as we expect from Proposition~\ref{prop:1d}. Similarly, substituting $c = 4$ into Theorem~\ref{thm:2d} we obtain $p_1 = (1 - \theta/\pi)^2$ and $p_2 = (1/2)^2$, matching the combined collision probability of two independent antipodal codes. As the nearest neighbor exponents $\rho = \log p_1 / \log p_2 = \log p_1^2 / \log p_2^2$ are invariant under applying the same exponentiation to both $p_1$ and $p_2$, both the regular $2$-gon and $4$-gon (square) achieve the same performance when using Gaussian projection matrices $\mat{A}$.

For other values of $c \geq 2$, the collision probabilities do not simplify much further than the given expression in Theorem~\ref{thm:2d} with $c$ replaced by the corresponding number. Looking at the nearest neighbor exponents $\rho = \log p_1 / \log p_2$ for e.g.\ $\theta = \frac{\pi}{3}$ and $\theta = \frac{\pi}{4}$, we obtain the two graphs in Figure~\ref{fig:2d1}. As we can see, the query exponents for $c = 2$ and $c = 4$ are the same, and using very large polygons does not improve the performance; any polygon with more than $4$ vertices achieves worse query exponents $\rho$ than antipodal locality-sensitive hashing. For $c = 3$ however we observe a non-trivial minimum: using equilateral triangles, the query time complexity for nearest neighbor searching is lower than for $c = 2, 4$.

\begin{figure}[t]
\includegraphics[width = 6.8cm]{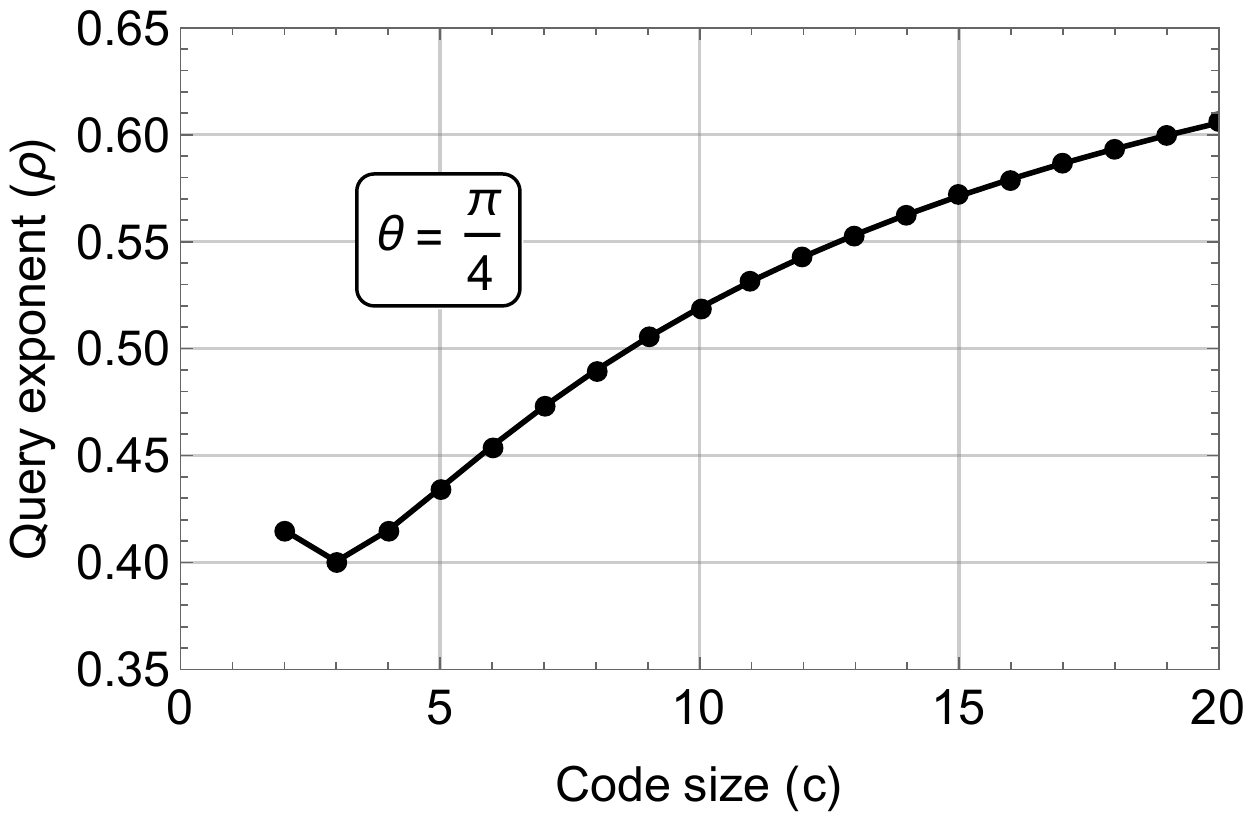} \ \includegraphics[width = 6.8cm]{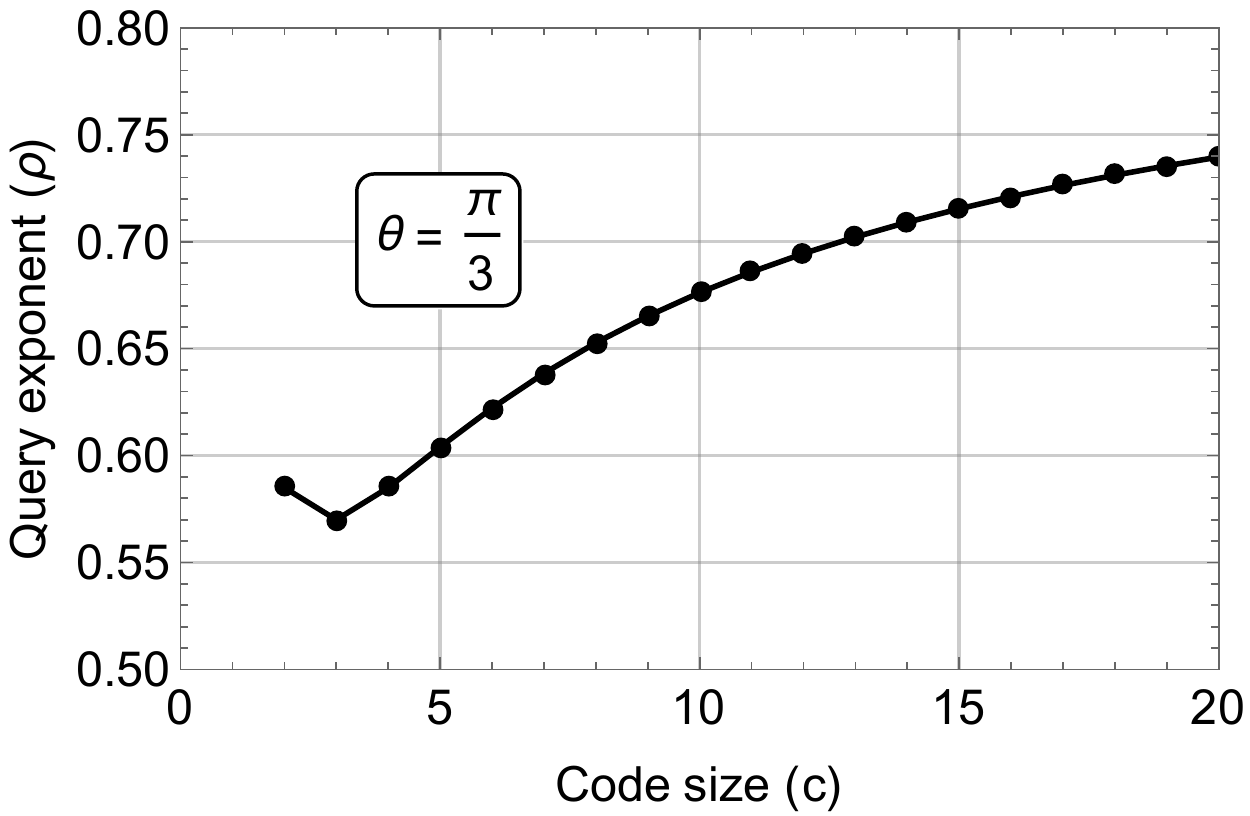}
\caption{Query exponents $\rho$ for project--and--partition hash families based on two-dimensional regular polygons. Lower query exponents $\rho$ are better, and the best performance is achieved by the equilateral triangle ($c = 3$). The case $c = 2$ corresponds to hyperplane hashing, while $c = 4$ corresponds to using the vertices from a square, which is asymptotically equivalent to using two (orthogonal) hyperplanes. \label{fig:2d1}}
\end{figure}

\begin{theorem}[In two dimensions, the triangle is optimal]
Let $\theta \in (0, \frac{\pi}{2})$. Then, among all project--and--partition hash families induced by two-dimensional isogonal spherical codes, the query exponent $\rho = \log p_1 / \log p_2$ is minimized by the equilateral triangle.
\end{theorem}
\begin{proof}
By inspecting the function $f(c, \theta) = \log p_1 / \log p_2$ with $p_1$ and $p_2$ as in \eqref{eq:2d}, we see that on the positive reals, its derivative with respect to $c$ always has a single root $c_0 = c_0(\theta) \in (2, 3)$. For $c \leq c_0$ the derivative is negative and for $c \geq c_0$ the derivative is positive, hence on the positive reals, $f(c, \theta)$ is minimized at $c = c_0$, where the exact value $c_0$ depends on $\theta$. On the natural numbers, $f(c, \theta)$ is therefore minimized either at $c = 2$ or at $c = 3$. Since $f(2, \theta) = f(4, \theta) > f(3, \theta)$, the value at $c = 3$ is the unique global minimum, and therefore for all target nearest neighbor angles $\theta \in (0, \frac{\pi}{2})$, the equilateral triangle achieves the lowest values $\rho$. \qed 
\end{proof}
For example, when $\theta = \frac{\pi}{3}$ we obtain a query exponent $\rho = 0.56996\dots$ for the equilateral triangle, compared to $\rho = 0.58496\dots$ for the square or for two antipodal points. Recall again that this improvement holds unconditionally, for all $d$; these collision probabilities are exact, and do not hide order terms in $d$. A further overview of query exponents $\rho$ for regular polygons is given in Table~\ref{tab:overview}, in the rows labeled $k = 2$.

\subsection{Arbitrary two-dimensional spherical codes}

For arbitrary (non-isogonal) two-dimensional spherical codes, where code word $\vc{c}_i$ covers a fraction $f_i$ of the circle, the probability that two vectors collide in one of the $c$ hash regions, given their mutual angle $\theta$, can be computed in terms of the vector $\vc{f}$ as:
\begin{align}
p_1(\theta, \vc{f}) &= \|\vc{f}\|_2^2 + c \left(\frac{\pi - \theta}{2 \pi}\right)^2 - \frac{1}{4 \pi^2} \sum_{i=1}^c \arccos^2(- \cos \theta \cos 2 \pi f_i), \quad p_2(\vc{f}) = \|\vc{f}\|_2^2.
\end{align}
The set of query exponents $\rho$ that can be achieved by two-dimensional codes $\cC \subset \cS^1$ in the project--and--partition framework, given a target angle $\theta$, is then equal to:
\begin{align}
R_{\theta} = \left\{\frac{\log p_1(\theta, \vc{f})}{\log p_2(\vc{f})} : \vc{f} \in \mathbb{R}^c, c \geq 2, \vc{f} \geq \vc{0}, \|\vc{f}\|_1 = 1\right\}.
\end{align}
Here $\|\vc{f}\|_2$ and $\|\vc{f}\|_1$ denote the $\ell_2$-norm and $\ell_1$-norm of $\vc{f}$. Although a numerical evaluation indicates that, for arbitrary $\theta \in (0, \frac{\pi}{2})$, the minimum over $R_{\theta}$ is obtained for $\vc{f} = (\frac{1}{3}, \frac{1}{3}, \frac{1}{3})$, we leave a formal proof of the optimality of the equilateral triangle over all two-dimensional codes as an open problem.


\section{Three-dimensional codes}
\label{sec:3d}

As the dimensionality increases, the dimensions of the matrices $\mat{C}_i$ in Theorem~\ref{thm:main} and Corollary~\ref{thm:isogonal} generally also increase, even when using the relevant vectors optimization of Proposition~\ref{prop:relevant}. In particular, as each Voronoi region in three dimensions is defined by at least three relevant vectors, the matrix $\mat{\Sigma}$ will be at least six-dimensional with non-trivial off-diagonal entries. The bad news is that for none of the three-dimensional codes we considered, we could derive analytic, closed-form expressions for the collision probabilities; numerical integration seems to be necessary to evaluate their performance. The good news however is that Section~\ref{sec:framework} provides us with a computational framework to evaluate the suitability of a spherical code for nearest neighbor searching, given the vertices defining the spherical code.

Table~\ref{tab:overview} summarizes the results of a brief survey of various spherical codes appearing in the literature in different contexts. The three-dimensional codes are given in the rows labeled $k = 3$. The codes we considered and evaluated numerically are:
\begin{itemize}
\item The \textbf{tetrahedron}, or the $3$-simplex, consisting of $4$ equidistributed points on the sphere;
\item The \textbf{octahedron}, or the $3$-orthoplex (cross-polytope), consisting of the $6$ vectors $\pm \vc{e}_i$;
\item The \textbf{cube}, consisting of $8$ vertices of the form $\frac{1}{3} \sqrt{3} \cdot (\pm 1, \pm 1, \pm 1)$;
\item The \textbf{icosahedron}, one of the five Platonic solids, with $12$ vertices;
\item The \textbf{cuboctahedron}, the vertex figure of the root lattices $A_3$ and $D_3$, with $12$ vertices;
\item The \textbf{dodecahedron}, the largest of the Platonic solids, with $20$ vertices;
\item Three \textbf{sphere packings} of $5, 7, 9$ vertices from the sphere packing database of Sloane~\cite{sloanepackings}.
\end{itemize}
Numerically evaluating the resulting query exponents $\rho$, it seems that for all $\theta \in (0, \frac{\pi}{2})$, the tetrahedron gives the best query exponents $\rho$. Similar to the two-dimensional case, using large, dense spherical codes does not seem to help; a fine-grained partition may make $p_2$ smaller, but at the same time $p_1$ then also becomes so much smaller that $\rho$ still increases. Note that, similar to the triangle, the tetrahedron belongs to the class of simplices.
\begin{conjecture}[In three dimensions, the tetrahedron is optimal]
Let $\theta \in (0, \frac{\pi}{2})$. Then, \\ among all project--and--partition hash families induced by three-dimensional (isogonal) spherical codes, the query exponent $\rho = \log p_1 / \log p_2$ is minimized by the regular tetrahedron.
\end{conjecture}
For the tetrahedron, the collision probabilities are $p_2 = 1/4$ and:
\begin{align}
p_1 &= 4 \cdot \Pr_{\vc{z} \sim \No(\vc{0},\mat{\Sigma})}\left(\vc{z} \geq \vc{0}\right), \qquad \mat{\Sigma} = \begin{bmatrix} 1 & \cos \theta \\ \cos \theta & 1 \end{bmatrix} \otimes \begin{bmatrix} 1 & 1/2 & 1/2 \\ 1/2 & 1 & 1/2 \\ 1/2 & 1/2 & 1 \end{bmatrix}.
\end{align}
A further simplification of $p_1$, i.e.\ finding a closed-form expression for this six-dimensional orthant probability is left as an open problem \cite{orthantmo}.

For numerically computing the values $\rho$ in Table~\ref{tab:overview} for three-dimensional codes, we used the statistical software R together with the mvtnorm-package~\cite{genz19}, which specifically implements fast and precise evaluation of orthant probabilities. For instance, for the tetrahedron with $\theta = \frac{\pi}{3}$, we numerically compute $\rho = 0.559998$. Using an independent C implementation, we performed Monte Carlo experiments, which for $\theta = \frac{\pi}{3}$ led to the estimate $\rho = 0.559997$ based on $10^8$ random trials. 

Note that the (numerical) integration task becomes more complex, and likely less precise, as the dimensionality of the problem increases (i.e.\ the size of $\Sigma$ increases). For isogonal codes, the dimensions of $\Sigma$ are not determined by the size of the code but by the number of relevant vectors for each code word. The entries in Table~\ref{tab:overview} may therefore be less precise for e.g.\ the icosahedron, with $|\cR| = 5$, and for large irregular sphere packings from~\cite{sloanepackings}.



\section{Four-dimensional codes}
\label{sec:4d}

As for the three-dimensional case, the orthant probabilities in the four-dimensional case generally appear too hard to evaluate analytically. Instead, we numerically computed parameters $\rho$ for different angles $\theta$ for the following list of four-dimensional spherical codes, by estimating the corresponding multivariate integrals with the R-library~\cite{genz19}:
\begin{itemize}
\item The \textbf{pentatope} or \textbf{$\vc{5}$-cell}, also known as the $4$-simplex, with $5$ vertices;
\item The \textbf{$\vc{16}$-cell}, the four-dimensional version of the orthoplex, with $8$ vertices;
\item The \textbf{tesseract}, or the four-dimensional hypercube, with $16$ vertices;
\item The \textbf{octacube} or \textbf{$\vc{24}$-cell}, the vertex figure of the $D_4$ root lattice, with $24$ vertices;
\item The runcinated $5$-cell, corresponding to the $20$ roots of the $A_4$-lattice;
\item Four \textbf{sphere packings} of $6, 7, 10, 13$ vertices from the sphere packing database of Sloane~\cite{sloanepackings}.
\end{itemize}
The only other regular convex polytopes not mentioned above are the $120$-cell and $600$-cell which, as their names suggest, are rather large; not only will the numerical evaluation (or Monte Carlo estimation) of the corresponding orthant probabilities likely be imprecise, but we also expect that such large codes are unlikely to offer an improvement, if we extrapolate the intuition obtained from the results in two and three dimensions to the four-dimensional setting. 

As shown in Table~\ref{tab:overview}, an interesting phenomenon occurs in dimension $4$. For two- and three-dimensional spherical codes we saw that the simplices give the best performance, beating other codes such as the orthoplices. In four dimensions however, for small angles $\theta$, the more fine-grained partitioning offered by the orthoplex gives a better performance for the associated locality-sensitive hashing scheme. Note that the best parameters $\rho$ are achieved by the simplex and orthoplex; these polytopes seem well-suited for locality-sensitive hashing, and even the nicest other spherical codes cannot compete with these highly regular and symmetric shapes.
\begin{conjecture}[In four dimensions, the $5$-cell and $16$-cell are optimal]
Let $\theta \in (0, \frac{\pi}{2})$. \\There exists a parameter $\theta_0 \in (0, \frac{\pi}{2})$ such that, among all project--and--partition hash families induced by four-dimensional (isogonal) spherical codes, the query exponent $\rho$ is minimized by:
\begin{itemize}
\item the $16$-cell, if $\theta \in (0, \theta_0)$;
\item the $5$-cell, if $\theta \in (\theta_0, \frac{\pi}{2})$.
\end{itemize}
\end{conjecture}
Numerically, we estimate the cross-over point between the $5$-cell and $16$-cell in terms of the query exponent $\rho$ to lie at $\theta_0 \approx 0.33 \pi$, slightly below $\theta = \frac{\pi}{3}$. See also Figure~\ref{fig:comp}, showing the improvements over hyperplane hashing for both the $5$-cell and the $16$-cell, and showing that the cross-over point between both four-dimensional codes lies close to $\theta = \frac{\pi}{3}$. From a practical point of view however, as discussed in Section~\ref{sec:practice}, using larger spherical codes (in a fixed dimension) is generally better, as the hash computations are then cheaper. In that sense, using the $16$-cell may be a better choice than using a $5$-cell for partitions even for angles $\theta > \theta_0$.

For completeness, for the $5$-cell and $16$-cell we performed independent Monte Carlo estimations of the collision probabilities using $10^8$ trials for each value $\theta$ in Table~\ref{tab:overview}, to verify that these numbers are accurate. The first three decimals in Table~\ref{tab:overview} were always an exact match, and small deviations in the fourth decimal were occasionally observed for small $\theta$.

Again, we stress that the list of spherical codes given in Table~\ref{tab:overview} is by no means exhaustive, and the framework provided in Section~\ref{sec:framework} may be useful for evaluating any other spherical codes, if one suspects better spherical codes exist for partitioning the space.



\section{Higher-dimensional codes}
\label{sec:kd}

For high dimensions $k$, the orthant probabilities become increasingly difficult to evaluate, and the number of known, \emph{nice} spherical codes quickly decreases. A well-known result is that in five or more dimensions, there exist only three classes of regular convex polytopes, which have also been studied before in the context of locality-sensitive hashing~\cite{terasawa07, terasawa09, andoni15cp, laarhoven17hypercube, kennedy17}. These are the simplices, orthoplices, and hypercubes. Besides these families of polytopes, we will also study the families of expanded simplices and rectified orthoplices, which correspond to the vertex figures of the $A_k$ and $D_k$ root lattices. Together with these families, we further study three more special polytopes in five and six dimensions, listed in Table~\ref{tab:overview}:
\begin{itemize}
\item The \textbf{$\mathbf{1_{21}}$-polytope} or the $5$-demicube, obtained by taking $16$ vertices of the $5$-cube;\footnote{Viewing the hypercube as a bipartite graph, the vertices of the demicube are all vertices from one of the two sets.}
\item The \textbf{$\mathbf{1_{31}}$-polytope}, or the $6$-dimensional demicube on $32$ vertices;
\item The \textbf{$\mathbf{2_{21}}$-polytope} or \textbf{Sch\"{a}fli polytope} on $27$ vertices, with symmetries of the $E_6$ lattice. 
\end{itemize}
The other polytopes in Table~\ref{tab:overview} for dimensions higher than four are all from the families $S_k$, $O_k$, $C_k$, $A_k$ and $D_k$ which we further discuss separately below.

\subsection{Simplices ($S_k$)}

Recall that simplices can most easily be described as the subset in $\mathbb{R}^{k+1}$ containing all standard unit vectors $\vc{e}_1, \dots, \vc{e}_{k+1}$. We are unable to obtain closed-form expressions for $p_1$ and $\rho$ for $k \geq 3$, beyond the multivariate integral representation from Corollary~\ref{thm:isogonal}, and we only obtain a closed-form expression of the correlation matrices appearing in the computation of $p_1$ below. In this case the matrix $\mat{D} \in \mathbb{R}^{k \times k}$ has the simple shape of having ones on the diagonal and all off-diagonal entries being equal to $1/2$. (Observe that orthant probabilities for such matrices, with all off-diagonal entries equal to the same value, were studied in e.g. \cite{steck62}.) Yet, due to the Kronecker product, an exact evaluation of the resulting orthant probabilities for $k \geq 3$ seems difficult. In the following proposition, recall that $\mat{I}_k$ denotes the identity matrix, and $\mat{J}_k$ the all-$1$ matrix.

\begin{proposition}[Collision probabilities for simplices] \label{prop:simplex}
Let $\cC \subset \cS^{k-1}$ consist of the $k + 1$ vertices of the simplex, and let $\theta \in (0, \frac{\pi}{2})$. Then the corresponding project--and--partition hash family $\cH$ is $(\theta, p_1, p_2)$-sensitive, with $p_2 = 1/(k+1)$ and:
\begin{align}
p_1 &= (k+1) \cdot \Pr_{\vc{z} \sim \No(\vc{0},\mat{\Sigma})}\left(\vc{z} \geq \vc{0}\right), \qquad \mat{\Sigma} = \begin{bmatrix} 1 & \cos \theta \\ \cos \theta & 1 \end{bmatrix} \otimes \tfrac{1}{2}(\mat{J}_{k} + \mat{I}_k).
\end{align}
\end{proposition}

\begin{proof}
For ease of computations, let $\cC = \{\vc{e}_i: i = 1, \dots, k+1\}$. Although this code is not centered around the origin, and lies in $\mathbb{R}^{k+1}$ rather than $\mathbb{R}^k$, none of these things matter when computing $\mat{D} = \hat{\mat{C}}_i \hat{\mat{C}}_i^T$, which is a function only of the difference vectors $\vc{c}_i - \vc{c}_j$. Now, a vector lies in the Voronoi cell of the vertex $\vc{e}_1$ if and only if it lies closer to $\vc{e}_1$ than to all its neighbors, which in this case are all other vectors in the code. The difference vectors $\vc{e}_1 - \vc{e}_j$ all have norm $\sqrt{2}$ and pairwise inner products $1$, hence after dividing $\mat{D}$ by $2$ we obtain a matrix with ones on the diagonal and $1/2$ everywhere off the diagonal. The result then follows from Theorem~\ref{thm:isogonal}. \qed
\end{proof}

Asymptotically, as $d \to \infty$, using simplices is equivalent to using cross-polytopes, and the last rows of Table~\ref{tab:overview} illustrate this asymptotic scaling as $k$ increases. 

\subsection{Orthoplices ($O_k$)}

For cross-polytopes or orthoplices, consisting of the $2k$ code words $\pm \vc{e}_1, \dots, \pm \vc{e}_k$, unfortunately we also cannot easily derive an analytic expression for the collision probabilities for $k \geq 3$. We state the reduced form of the correlation matrix below, where the matrix $\mat{D} \in \mathbb{R}^{(2k-2) \times (2k-2)}$ is a Toeplitz matrix with ones on the diagonal, $1/2$ on most of the other diagonals, and one diagonal strip above and below the diagonal with all zeros.
\begin{proposition}[Collision probabilities for orthoplices] \label{prop:orthoplex}
Let $\cC \subset \cS^{k-1}$ consist of the $2k$ vertices of the orthoplex, and let $\theta \in (0, \frac{\pi}{2})$. Then the corresponding project--and--partition hash family $\cH$ is $(\theta, p_1, p_2)$-sensitive, with $p_2 = 1/(2k)$ and:
\begin{align}
p_1 &= 2k \cdot \Pr_{\vc{z} \sim \No(\vc{0},\mat{\Sigma})}\left(\vc{z} \geq \vc{0}\right), \qquad \mat{\Sigma} = \begin{bmatrix} 1 & \cos \theta \\ \cos \theta & 1 \end{bmatrix} \otimes \frac{1}{2} \begin{bmatrix} \mat{J}_{k-1} + \mat{I}_{k-1} \ \ & \ \mat{J}_{k-1} - \mat{I}_{k-1} \\
\mat{J}_{k-1} - \mat{I}_{k-1} \ \ & \ \mat{J}_{k-1} + \mat{I}_{k-1} \end{bmatrix}.
\end{align}
\end{proposition}

\begin{proof}
Without loss of generality, let $\cC = \{\pm \vc{e}_j: j = 1, \dots, k\}$. For the vector $\vc{c}_1 = \vc{e}_1$, the relevant vectors are $\cR_1 = \cC \setminus \{\pm\vc{e}_1\}$, and taking differences with the relevant vectors, we observe that the matrix $\hat{\mat{C}}_1$ has $2k - 2$ rows $\vc{e}_1 \pm \vc{e}_j$, for $j = 2, \dots, k$. For the product $\mat{D} = \hat{\mat{C}}_1 \hat{\mat{C}}_1^T$ we get twos on the diagonal, ones for most off-diagonal entries, and each row is paired up with exactly one other row whose joint correlation is $0$. Through a suitable reordering of the rows/columns of $\mat{D}$, we thus get the claimed shape of $\mat{D}$. Finally, with $\cC$ being isogonal and of size $2k$, using Theorem~\ref{thm:isogonal} we get the claimed expressions for $p_1$ and $p_2$. \qed
\end{proof}

As $k \to \infty$, from~\cite{andoni15cp} we know that the parameters $\rho$ scale optimally as $\rho \to 1/(2c^2 - 1)$, where $c = \sqrt{1 / (1 - \cos \theta)}$ is the approximation factor for worst-case approximate nearest neighbor searching. Substituting values $\theta \in \frac{\pi}{12} \{1, 2, 3, 4, 5\}$, we obtain the values in the bottom rows of Table~\ref{tab:overview}.

\subsection{Hypercubes ($C_k$)}

The easiest of these three families of convex regular polytopes to analyze is the family of hypercubes. For arbitrary $k \geq 1$, these codes consist of $2^k$ vertices of the form $\frac{1}{\sqrt{k}} (\pm \vc{e}_1 \pm \dots \pm \vc{e}_k)$. When using normal matrices $\mat{A} \sim \No(0,1)^{k \times d}$ to project down onto a $k$-dimensional subspace, the following proposition shows that the query exponent does not change at all as $k$ increases.  

\begin{proposition}[Collision probabilities for hypercubes] \label{prop:hypercube}
Let $\cC \subset \cS^{k-1}$ consist of the $2^k$ vertices of the hypercube, and let $\theta \in (0, \frac{\pi}{2})$. Then the corresponding project--and--partition hash family $\cH$ is $(\theta, p_1, p_2)$-sensitive, with:
\begin{align}
p_1 &= \left(1 - \frac{\theta}{\pi}\right)^k, \qquad \qquad p_2 = \left(\frac{1}{2}\right)^k.
\end{align}
As a result, when using Gaussian projection matrices $\mat{A}$ the query exponents $\rho$ for hypercubes are equal to those for the antipodal spherical code of Proposition~\ref{prop:1d}.
\end{proposition}

\begin{proof}
As before, scaling $\cC$ by a constant factor does not affect collision probabilities, so w.l.o.g.\ we may let $\cC$ consist of the $2^k$ vectors $(\pm 1, \pm 1, \dots, \pm 1)$. A vector $\vc{x}$ is mapped to the code word $\vc{c}_0 = (1, 1, \dots, 1)$ if and only if it is closer to this code word than to its relevant vectors $\cR_0$, which in this case are the vectors with $(k - 1)$ entries equal to $+1$ and one entry equal to $-1$. For $j = 1, \dots, c$ let us denote by $\vc{c}_j$ the vector with ones everywhere, except at position $j$ where we have a $-1$. For the difference vectors $\vc{c}_0 - \vc{c}_j$ with $\vc{c}_j \in \cR_1$ we get $\vc{c}_0 - \vc{c}_j = 2 \vc{e}_j$, and as a result we obtain $\hat{\mat{C}}_0 \hat{\mat{C}}_0^T = 4 \cdot \mat{I}_k$, which is a scalar multiple of the identity matrix. Now, note that a correlation matrix $\mat{I}_k$ essentially means that the entries of the random vector are independent. It is therefore an easy exercise to see that, for all matrices $\mat{A}$:
\begin{align}
\Pr_{\vc{z} \sim \No(\vc{0},\mat{A} \otimes \mat{I}_k)}\left(\vc{z} \geq \vc{0}\right) = \Pr_{\vc{z} \sim \No(\vc{0},\mat{I}_k \otimes \mat{A})}\left(\vc{z} \geq \vc{0}\right) = \Pr_{\vc{z} \sim \No(\vc{0},\mat{A})}\left(\vc{z} \geq \vc{0}\right)^k.
\end{align}
Together with the derivation of the orthant probability for the one-dimensional antipodal spherical code, this leads to $p_1 = 2^k \cdot (\pi - \theta)^k / (2 \pi)^k = (1 - \theta/\pi)^k$. As $\cC$ is isogonal, we further obtain $p_2 = 1/2^k$ as claimed. \qed
\end{proof}

Note that the paper~\cite{laarhoven17hypercube} showed that, if $\mat{A} \in \mathbb{R}^{d \times d}$ is chosen as a random \textit{orthogonal} matrix, then in fact the performance of hypercube hashing gets better as $k$ increases. Here however we assumed $\mat{A} \sim \No(0,1)^{k \times d}$, which for small $k$ is essentially equivalent to sampling $\mat{A}$ as an orthogonal matrix, but is not quite the same for large $k$; see e.g.~\cite{diaconis87, jiang06} for more details on how the gap between orthogonal and Gaussian matrices changes with the relation between $k$ and $d$. This leads to two different asymptotic scalings, as indicated in the bottom rows of Table~\ref{tab:overview}.

\subsection{Expanded simplices ($A_k$)}

Similar to simplices, so-called expanded simplices (the vertex figures of the $A_k$ root lattices) can most easily be described as a subset in $\mathbb{R}^{k+1}$. These spherical codes contain all differences between unit vectors, $\vc{e}_i - \vc{e}_j$ for $i \neq j$. These vertices all lie on a $k$-dimensional hyperplane, and can thus be viewed as a $k$-dimensional spherical code. For these codes the matrix $\mat{D} \in \mathbb{R}^{(2k - 1) \times (2k - 1)}$ has a block structure described below. 

\begin{proposition}[Collision probabilities for expanded simplices] \label{prop:expsimplex}
Let $\cC \subset \cS^{k-1}$ consist of the $k (k + 1)$ vertices of the expanded simplex, and let $\theta \in (0, \frac{\pi}{2})$. Then the corresponding project--and--partition hash family $\cH$ is $(\theta, p_1, p_2)$-sensitive, with $p_2 = 1/k(k+1)$ and:
\begin{align}
p_1 &= k (k+1) \cdot \Pr_{\vc{z} \sim \No(\vc{0},\mat{\Sigma})}\left(\vc{z} \geq \vc{0}\right), \ \mat{\Sigma} = \begin{bmatrix} 1 & \cos \theta \\ \cos \theta & 1 \end{bmatrix} \otimes \frac{1}{2} \begin{bmatrix}
\mat{J}_{k-1} + \mat{I}_{k-1} & -\mat{I}_{k-1} \\
-\mat{I}_{k-1} & \mat{J}_{k-1} + \mat{I}_{k-1} \end{bmatrix}.
\end{align}
\end{proposition}

We omit the proof, which is analogous to the proofs of Propositions~\ref{prop:simplex} and \ref{prop:orthoplex}.

Observe that, as hashing in the expanded simplex corresponds to mapping to the two largest coordinates of a vector, for large $d$ up to order terms this is essentially equivalent to using two independent simplex hashes. For expanded simplices we therefore again get the optimal scaling for large $k$ also obtained with orthoplices and simplices, as illustrated in Table~\ref{tab:overview}.

\subsection{Rectified orthoplices ($D_k$)}

Rectified orthoplices, related to the root lattices $D_k$, contain as vertices all pairwise combinations of unit vectors; in contrast to $A_k$, where the signs must be opposite, here any combination of signs is included. As a subset of $\mathbb{R}^k$, this code thus contains the $2 k (k-1)$ vertices of the form $\pm \vc{e}_i \pm \vc{e}_j$ for $i \neq j$. For these codes the matrix $\mat{D} \in \mathbb{R}^{(2k - 1) \times (2k - 1)}$ has a $2 \times 2$ block structure described below. 
\begin{proposition}[Collision probabilities for rectified orthoplices] \label{prop:rectorthoplex}
Let $\cC \subset \cS^{k-1}$ consist of the $2 k (k - 1)$ vertices of the rectified orthoplex, and let $\theta \in (0, \frac{\pi}{2})$. Then the corresponding project--and--partition hash family $\cH$ is $(\theta, p_1, p_2)$-sensitive, with $p_2 = 1/2k(k-1)$ and:
\begin{align}
p_1 &= 2 k (k-1) \cdot \Pr_{\vc{z} \sim \No(\vc{0},\mat{\Sigma})}\left(\vc{z} \geq \vc{0}\right), \quad \mat{\Sigma} = \begin{bmatrix} 1 & \cos \theta \\ \cos \theta & 1 \end{bmatrix} \otimes \tfrac{1}{2} \begin{bmatrix}
\mat{J} + \mat{I} \ & \ \mat{J} - \mat{I} \ & \mat{I} & -\mat{I} \\
\mat{J} - \mat{I} \ & \ \mat{J} + \mat{I} \ & -\mat{I} & \mat{I} \\
\mat{I} & -\mat{I} & \ \mat{J} + \mat{I} \ & \ \mat{J} - \mat{I} \ \\
-\mat{I} & \mat{I} & \ \mat{J} - \mat{I} \ & \ \mat{J} + \mat{I} \ \\
 \end{bmatrix}.
\end{align}
All matrices $\mat{I} = \mat{I}_{k-2}$ and $\mat{J} = \mat{J}_{k-2}$ above are square matrices of size $(k - 2) \times (k - 2)$.
\end{proposition}
We omit the proof, which is analogous to the proofs of Propositions~\ref{prop:simplex} and \ref{prop:orthoplex}. Note again that, as we are mapping to the two largest coordinates after projection, the hash function for the rectified orthoplex has the same asymptotic scaling as the simplex and orthoplex.

\subsection{Comparison}

To compare these families of spherical codes in terms of their performance for nearest neighbor searching, Figure~\ref{fig:high1} depicts the resulting parameters $\rho$ for different values $k$, focusing on the two cases $\theta = \frac{\pi}{3}$ and $\theta = \frac{\pi}{4}$. As observed previously there is a cross-over between simplices and orthoplices around $k = 4$, after which the orthoplex consistently performs better than the simplex. The hypercube, as shown in Proposition~\ref{prop:hypercube}, remains stable at the same performance as the antipodal code, regardless of the choice of dimension $k$.


\section{Utopian spherical cap estimates}
\label{sec:lb}

To get an idea how good the spherical codes considered above are, ideally we would find a matching lower bound, stating that any spherical code of dimension at most $k$ must satisfy $\rho(\theta) \geq \rho_k(\theta)$, for some monotonically increasing function $\rho_k: (0, \frac{\pi}{2}) \to (0,1)$. However, ruling out the existence of exotic spherical codes which happen to perform better than the ones we considered seems difficult.

As has been observed several times before~\cite{andoni14, becker16lsf, andoni17}, using spherical caps to partition the sphere into regions seems to work best, so that the regions have nice, symmetric shapes that ultimately seem to minimize the exponent $\rho$. We can therefore get lower bounds on the parameter $\rho$ by assuming that the sphere can be perfectly divided into an integer number of spherical caps, and estimating the resulting parameter $\rho$ for these utopian\footnote{Note that for $k \geq 3$ it is impossible to partition the sphere into spherical caps, as regions either overlap, or part of the sphere remains uncovered. These conjectured lower bounds therefore do not correspond to actual, achievable bounds, and are likely not tight.} partitions of the sphere. This leads to the following result, where we assume such perfect spherical cap partitions exist, and where we compute collision probabilities $p_1, p_2$ and the query exponents $\rho$ accordingly.
\begin{theorem}[Spherical cap lower bounds] \label{thm:lb2}
Let $\cC \subset \cS^{k-1}$ be a spherical code, and let $\cH$ be the associated project--and--partition hash family. Then the parameter $\rho$ for $\cC$, for angle $\theta \in (0, \frac{\pi}{2})$, must satisfy:
\begin{align}
\rho(\theta) \geq \rho_k(\theta) := \min_{c \geq 2} \rho_k(c, \theta),
\end{align}
where $\rho_k(c, \theta)$ is given by:
\begin{align}
\rho_k(c, \theta) := \log\left\{\Pr_{\vc{x}, \vc{y} \sim \No(0,1)^k} \left(\frac{x_1}{\|\vc{x}\|} \geq \alpha_k(c), \frac{x_1 \cos \theta + y_1 \sin \theta}{\|\vc{x} \cos \theta + \vc{y} \sin \theta\|} \geq \alpha_k(c) \right)\right\} \bigg/ \log\left(\frac{1}{c}\right), \label{eq:47}
\end{align}
and $\alpha_k(c)$ is the solution to $\frac{1}{2} \cdot I_{1 - \alpha^2}(\frac{k-1}{2}, \frac{1}{2}) = \frac{1}{c}$.
\end{theorem}
\begin{proof}
Suppose we use a spherical code in $k$ dimensions, partitioning the $k$-sphere into $c$ regions $\cV_1, \dots, \cV_c$. First we will argue that the performance parameter $\rho$ of any such spherical code is at most the corresponding parameter where each region $\cV_i$ was replaced by a spherical cap of the same volume. Then we will argue that if indeed we used spherical caps to form our code, it is optimal to use spherical caps of equal size. Finally we will tie up the loose ends to derive the stated lower bounds.

\subparagraph{Proof, part 1: Spherical caps are optimal.} First, suppose we fix the size of the code $c$ as well as the volumes of each Voronoi region. It is clear that the collision probability $p_2 = \sum_{i=1}^c \vol(\cV_i)^2 / \vol(\cS^{k-1})^2$ remains invariant under any changes we make to $\cC$ adhering to these constraints, and so the only question is whether $p_1$ is largest when each region is a spherical cap. 

To prove this we use a classical result of Baernstein--Taylor~\cite{baernstein76}, which extends Riesz' inequality~\cite{riesz30} to the Euclidean sphere. Suppose we have functions $f,g: \cS^{d-1} \to \R$ and $h: [-1,1] \to \R$, and let the surface measure $\sigma$ on $\cS^{d-1}$ be normalized so that $\sigma(\cS^{d-1}) = 1$. Let $h$ further be non-decreasing, bounded, and measurable, and let $f,g \in L_1(\cS^{d-1})$. Let:
\begin{align}
    \cT(f,g,h) := \iint\limits_{\cS^{k-1} \times \cS^{k-1}} f(\vc{x}) g(\vc{y}) h(\vc{x} \cdot \vc{y}) d\sigma(\vc{x}) d\sigma(\vc{y}).
\end{align}
Then \cite[Theorem 2]{baernstein76} states:
\begin{align}
    \cT(f,g,h) \leq \cT(f^*, g^*, h),
\end{align}
where $f^* = f^*(x_1)$ is a non-decreasing function which only depends on the first coordinate $x_1$, $g^* = g^*(y_1)$ is a non-decreasing function which only depends on $y_1$, and $\sigma(\{f^* > \lambda\}) = \sigma(\{f > \lambda\})$ and $\sigma(\{g^* > \lambda\}) = \sigma(\{g > \lambda\})$ for all $\lambda \in \R$.

To instantiate it, for each Voronoi region $i = 1, \dots, c$, we apply \cite[Theorem 2]{baernstein76} with $f(\vc{x}) := \mathds{1}\{\vc{x} \in \cV_i\}$, $g(\vc{x}) := \mathds{1}\{\vc{y} \in \cV_i\}$, and $h(\vc{x} \cdot \vc{y}) = \mathds{1}\{\vc{x} \cdot \vc{y} \geq \cos \theta\}$. Note that for these functions $f$ and $g$ we have $f^*(\vc{x}) = \mathds{1}(\vc{x} \cdot \vc{e}_1 \geq \alpha_i\}$ and $g^*(\vc{x}) = \mathds{1}(\vc{x} \cdot \vc{e}_1 \geq \alpha_i\}$ with $\alpha_i$ depending on the volume of $\cV_i$: the functions $f^*$ and $g^*$ correspond to spherical caps $\cC_i$.
\begin{align}
\cT(f,g,h) &= \iint\limits_{\cS^{k-1} \times \cS^{k-1}} \mathds{1}\{\vc{x} \in \cV_i\} \mathds{1}\{\vc{y} \in \cV_i\} \mathds{1}\{\vc{x} \cdot \vc{y} \geq \cos \theta\} d\sigma(\vc{x}) d\sigma(\vc{y}) \\
&= \Pr_{\vc{x}, \vc{y} \sim \cS^{k-1}} \left(\vc{x} \in \cV_i, \vc{y} \in \cV_i, \vc{x} \cdot \vc{y} \geq \cos \theta\right).
\end{align}
On the other hand, for $\cT(f^*,g^*,h)$ we obtain:
\begin{align}
\cT^*(f,g,h) &= \iint\limits_{\cS^{k-1} \times \cS^{k-1}} \mathds{1}\{\vc{x} \cdot \vc{e}_1 \geq \alpha_i\} \mathds{1}\{\vc{y} \cdot \vc{e}_1 \geq \alpha_i\} \mathds{1}\{\vc{x} \cdot \vc{y} \geq \cos \theta\} d\sigma(\vc{x}) d\sigma(\vc{y}) \\
&= \Pr_{\vc{x}, \vc{y} \sim \cS^{k-1}} \left(\vc{x} \in \cC_i, \vc{y} \in \cC_i, \vc{x} \cdot \vc{y} \geq \cos \theta\right).
\end{align}
By the inequality $\cT(f,g,h) \leq \cT(f^*,g^*,h)$ it follows that each such probability is upper bounded by the corresponding spherical cap probability. Thus, if our original hash function $h$ uses a spherical code $\cC$ with Voronoi regions $\cV_i$, and a utopian spherical code $\cC^*$ were to consist of induced Voronoi regions in the shape of spherical caps $\cC_i$ and generated a hash function $h^*$, then the previous results would show that:
\begin{align}
\Pr_{\vc{x}, \vc{y} \sim \cS^{k-1}} \left(h(\vc{x}) = h(\vc{y}) \mid \vc{x} \cdot \vc{y} \geq \cos \theta\right) &= \sum_{i=1}^c \Pr_{\vc{x}, \vc{y} \sim \cS^{k-1}} \left(\vc{x}, \vc{y} \in \cV_i \mid \vc{x} \cdot \vc{y} \geq \cos \theta\right) \\
& \leq \sum_{i=1}^c \Pr_{\vc{x}, \vc{y} \sim \cS^{k-1}} \left(\vc{x}, \vc{y} \in \cC_i \mid \vc{x} \cdot \vc{y} \geq \cos \theta\right) \\
& = \Pr_{\vc{x}, \vc{y} \sim \cS^{k-1}} \left(h^*(\vc{x}) = h^*(\vc{y}) \mid \vc{x} \cdot \vc{y} \geq \cos \theta\right).
\end{align}
In other words, in a utopian world we would always replace hash regions with spherical caps of the same volume if we could, as the resulting parameter $\rho$ can only get smaller.

\subparagraph{Proof, part 2: Equal spherical caps are optimal.} The previous derivation shows that the parameter $\rho$ can be lower-bounded by replacing each region with a spherical cap of the same volume. The next question is: if we do use such spherical cap-shaped hash regions, does it ever make sense to use regions of different volumes? 

In the above derivation, where each spherical region is independent of the others (we assume in our utopian setting that the spherical caps are perfectly disjoint, and the only boundary conditions for regions $\cV_i$ are defined by $\alpha_i$), the performance parameter $\rho$ of the entire code is a weighted average of the individual performance parameters $\rho$ of each of these spherical caps on parts of the sphere. This weighted average is minimized when we choose all the regions to be of the size that minimizes $\rho$, i.e.\ we find the spherical cap size that offers the best balance between the collision probabilities $p_1$ and $p_2$ and use that cap size for all our spherical caps.\footnote{Note that the resulting optimal code size may technically not be integral, and so the optimal integral-sized code may not consist of equal regions. However, minimizing over real-valued $c$ will always provide a lower bound on the parameter $\rho$ of any $k$-dimensional spherical code.}

\subparagraph{Proof, part 3: Obtaining the final expression.} Finally, observe that w.l.o.g.\ we may model two vectors $\vc{v}, \vc{w}$ at angle $\theta$ in $d$-dimensional space by the two vectors $\vc{v} = \vc{e}_1$ and $\vc{w} = \vc{e}_1 \cos \theta + \vc{e}_2 \sin \theta$. A projection via a random Gaussian matrix $A \sim \No(0,1)^{d \times k}$, whose first two columns are denoted $\vc{x}$ and $\vc{y}$, then projects $\vc{v}$ onto $A \vc{v} = \vc{x}$ and $\vc{w}$ onto $A \vc{w} = \vc{x} \cos \theta + \vc{y} \sin \theta$, with $\vc{x}, \vc{y} \sim \No(0,1)^k$. The condition of being contained in a spherical cap then requires an additional normalization, as these vectors $\vc{x}$ and $\vc{y}$ are not necessarily normalized. This finally leads to the given expression for $\rho_k(c,\theta)$, which denotes the parameter $\rho$ for a utopian code consisting of $c$ equal-sized spherical caps. Note that the given parameter $\alpha_k(c)$ is such that $\Pr_{\vc{z} \in \cS^{k-1}}(z_1 \geq \alpha_k(c)) = 1/c$.  
\end{proof}
%

For analyzing actual spherical codes, we obtained conditions of the form $\|\vc{x} - \vc{c}_i\| \geq \|\vc{x} - \vc{c}_j\|$, which then translated to conditions $\langle \vc{x}, \vc{c}_i - \vc{c}_j \rangle \geq 0$, leading to orthant probabilities. Here however we have different conditions of the form $\|\vc{x}/\|\vc{x}\| - \vc{c}_i\| \leq u$ for $\vc{x}$ to be contained in the Voronoi cell (spherical cap) defined by code word $\vc{c}_i$. Assuming w.l.o.g.\ $\vc{c} = \vc{e}_1$, the formulas reduce to those given in \eqref{eq:47}, but these unfortunately cannot be further simplified or expressed in terms of orthant probabilities.

For small values $k = 1, \dots, 6$, and for the angles $\theta \in \frac{\pi}{12} \{1, \dots, 5\}$, we numerically computed values $\rho_k(c, \theta)$, and taking the minimum over $c$, we obtained the values $\rho_k(\theta)$ given in Table~\ref{tab:overview}. The superscripts in Table~\ref{tab:overview} indicate the value $c$ leading to the minimum value $\rho$, i.e.\ the ideal code size for these utopian spherical cap codes. Although these codes are not actually instantiable, these lower bounds may be useful as a guideline for choosing the code size $c$, when trying to come up with the best spherical codes in this framework.

From Table~\ref{tab:overview}, observe that as the angle $\theta$ decreases, and the nearest neighbor is expected to lie very close to the query vector, the optimal code size $c$ increases; for problem instances where the neighbor lies very close to the query vector, a fine-grained partition will likely work better. Also note that as $k$ increases, the optimal code size $c$ for these utopian codes increases as well, and whereas for small $k$ the optimal code size is always less than the number of vertices of the simplex and orthoplex, for larger $k$ the optimal code size increases beyond that of the simplex, and also beyond that of the orthoplex. This is in line with our observation that as $k$ increases, orthoplices start outperforming simplices, and this also hints at the idea that as $k$ increases further, even denser codes (likely with a superlinear number of vertices) will exhibit a better performance (smaller query exponents $\rho$) than partitions based on the cross-polytopes.

\begin{conjecture}[Orthoplices are suboptimal for large $k$] \label{conj:ortho} 
There exists a dimension $k_0 \in \mathbb{N}$ such that, for all dimensions $k \geq k_0$, there exist spherical codes $\cC \subset \cS^{k-1}$ whose query exponents $\rho$ in the project--and--partition framework are smaller than the exponents $\rho$ of the $k$-orthoplex.
\end{conjecture}

An interesting open question, which may well be closely related to other sphere packing and covering problems, is to find higher-dimensional spherical codes that are more suitable for nearest neighbor tasks than using the vertices of simplices or orthoplices. 


\section{Choosing a code in practice}
\label{sec:practice}

With the above framework in mind, and looking at the results from Table~\ref{tab:overview} and Figures~\ref{fig:comp} and \ref{fig:high1}, one might ask: how do I choose the best spherical code for my application? Is this purely a matter of choosing the code with the smallest exponent $\rho$? Or is the practical performance affected by other aspects of the code as well?

To look for an answer to this question, recall that in locality-sensitive hashing (as in Equation~\eqref{eq:lshpars}) there are several parameters that play a role:
\begin{align}
\rho = \frac{\log 1/p_1}{\log 1/p_2} \, , \qquad \ell = \frac{\log n}{\log 1/p_2} \, , \qquad t = O(n^{\rho}) \, . 
\end{align}
Here the constant in $t$ determines the success probability of the scheme. Now, the concrete costs for finding nearest neighbors with locality-sensitive hashing can be summarized as follows:
\begin{description}
\item[Projections:] For a $k$-dimensional code, we use Gaussian projection matrices $\mat{A} \sim \No(0,1)^{k \times d}$, and without further optimizations to $\mat{A}$, computing a product $\mat{A} \vc{x}$ requires $k \cdot d$ multiplications. Generating such a matrix $\mat{A}$ can also be done in time proportional to $k \cdot d$.
\item[Decoding:] Although for some spherical codes faster decoding algorithms may exist, in general the decoding cost for a spherical code of size $c$ in dimension $k$ is proportional to $k \cdot c$ multiplications.
\item[Hash tables:] We need $t = O(n^{\rho})$ hash tables for our data structure.
\item[Hash functions:] For each of the $t$ tables, we need to initialize $\ell$ random hash functions, for a cost of $\ell \cdot t$ hash functions. Each of these requires initializing a random projection matrix $\mat{A} \sim \No(0,1)^{k \times d}$, which can all be generated in time proportional to $\ell \cdot t \cdot k \cdot d$ multiplications.
\item[Preprocessing time:] For each of $n$ vectors, and each of $t$ hash tables, we need to compute the $k$-concatenated hash value, and insert the vector in the corresponding bucket. Timewise, computing all hash values for all vectors takes time proportional to $\ell \cdot t \cdot k \cdot d \cdot n$.
\item[Memory requirement:] We need to store $t$ hash tables, each of which contains all $n$ vectors, for a total cost of $t \cdot n$ entries in the hash tables.\footnote{Observe that we can keep a single copy of the entire list in memory, and only store pointers of essentially constant space to these vectors in the hash tables.}
\item[Query time -- Hashing:] Given a vector, we need to compute its hash values in all $t$ hash tables, for a cost of $\ell \cdot t \cdot k \cdot d$ multiplications.
\item[Query time -- Comparisons:] In total we expect that, apart from the nearest neighbor, approximately $t$ vectors will collide with a query vector in the $t$ hash tables. Going through these vectors to find the actual nearest neighbor can be done in time proportional to $t \cdot d$.
\end{description}
Here, an important observation is that the projections generally become more costly as the dimension $k$ of the spherical code increases; in higher dimensions, spherical codes with better parameters $\rho$ exist, but at the same time we pay for these more complex spherical codes in the decoding costs. The highest time and space complexities for the different parts of the nearest neighbor data structure can be summarized as follows:
\begin{align}
T_{\text{query}} = O(\ell \cdot t \cdot k \cdot d), \qquad T_{\text{preprocessing}} = O(n \cdot T_{\text{query}}), \qquad S = O(t \cdot n),
\end{align}
For the query time complexity $T_{\text{query}}$, the factor $d$ is constant and does not depend on the choice of the spherical code $\cC$. Similarly, the term $\ell$ contains a factor $\log n$ which does not depend on the choice of $\cC$. Dividing out a factor $n \log_2 n$, the product $t \cdot \ell \cdot k / \log_2 n$ has the following form for isogonal codes $\cC$:
\begin{align}
\frac{t \cdot \ell \cdot k}{\log_2 n} = O\left(n^{\rho} \cdot \frac{k}{\log_2 c}\right).
\end{align}
Now, for large $n$ this product is clearly minimized especially when $\rho$ is smallest. For concrete values of $n$ however, the other main factor that determines the query complexity is the ratio $\log_2(c) / k$. Here $\log_2(c)$ can be seen as the number of hash bits that are extracted from a single hash value $h(\vc{x}) \in \{1, \dots, c\}$, while the factor $k$ comes from the cost of computing a projection onto a $k$-dimensional subspace. In general, codes with a large ratio $\log_2(c) / k$ require fewer random projections to extract the necessary number of hash bits to determine the hash bucket, while small codes in high dimensions are in that sense rather costly. This quantitatively explains why e.g.\ hyperplane hashing is fast in practice, as $\log_2(c) / k = 1$, but cross-polytope hashing has a high cost for computing projections, as $\log_2(c) / k = O(\log(k) / k)$ is almost a linear factor smaller. 

To compare all codes in terms of their extracted number of bits per projection, as well as their query exponents $\rho$, Figure~\ref{fig:to} shows the resulting trade-offs for the spherical codes considered in this paper. As we can see, there is often a concrete trade-off between the query exponent $\rho$ and the ratio of extracted hash bits $\log_2(c) / k$; codes with small $\rho$ generally have a high hash cost, and codes which partition a low-dimensional space in many regions pay the price with a higher value $\rho$. Figure~\ref{fig:to} also paints a different picture as to which codes are best suited in practice, compared to Table~\ref{tab:overview}; the lattice-based codes $D_k$ may be better choices in practice than the pure orthoplex hashing of \cite{andoni15cp}, achieving better values $\rho$ while keeping the extracted hash bits per projection vector larger than for orthoplex hashing. 

For an even more concrete comparison of different codes in practice, Figure~\ref{fig:prac} describes concrete cost estimates for processing queries for different spherical codes, for different parameters $n$. Here the query cost estimate is simplified to $t \cdot \ell \cdot k \cdot d$, without taking into account further multiplicative order terms (which are likely to be roughly the same for different spherical codes). As the figure shows, for large $n$ the query exponent $\rho$ dominates, and codes with small values $\rho$ achieve the best query cost estimates. For smaller $n$ (e.g. $n = 10^5$, which might be the size of the database in applications in cryptanalysis~\cite{laarhoven15crypto, becker16lsf}), the cross-over point with hyperplane hashing quickly goes down, and using full-size cross-polytopes may not be as practical as using hyperplanes to partition the space; the smaller $\rho$ of cross-polytope hashing does not outweight the higher cost of the many random projections. This matches practical results from e.g.~\cite{laarhoven15crypto, becker16cp, becker16lsf, mariano15, mariano17}, which established that in applications with rather small $n$, hyperplane hashing dominates. For other applications, where $n$ is larger, using $D_k$ lattices or e.g. the $2_{21}$-polytope seems to give the best performance.

As a generalization of the (rectified) orthoplex families of hash functions $O_k$ and $D_k$, we finally propose the following hash function families, mapping a vector to the $m$ absolute largest coordinates.
\begin{definition}[$m$-max locality-sensitive hash functions]
Let $1 \leq m \leq k \leq d$. Then we define the $m$-max hash function family $\cH$ as the project--and--partition hash family with code:
\begin{align}
\cC = \left\{\sum_{i=1}^k s_i \vc{e}_i: \vc{s} \in \{-1, 0, 1\}^k, \|\vc{s}\|_1 = m\right\}. 
\end{align} 
For arbitrary $k$ and $m$, this spherical code has size $c = 2^m \binom{k}{m}$.
\end{definition}
As special cases, $m = 1$ corresponds to orthoplex hashing ($O_k$), $m = 2$ corresponds to the rectified orthoplices ($D_k$), and $m = k$ corresponds to hypercube hashing ($C_k$). As Figure~\ref{fig:to} suggests that larger codes ($m = 2$) may work better in practice than smaller codes ($m = 1$), we predict that for large $k$, larger values $m$ will outperform both orthoplex hashing and rectified orthoplex hashing. An interesting open question would be to estimate how $m$ should ideally scale with $k$ to obtain the best results in practice.


\bibliography{Database}

\begin{thebibliography}{10}

\bibitem{achlioptas01}
Dimitris Achlioptas.
\newblock Database-friendly random projections.
\newblock In {\em PODS}, pages 274--281, 2001.
\newblock URL: \url{http://dl.acm.org/citation.cfm?id=375608}, \href
  {http://dx.doi.org/10.1145/375551.375608} {\path{doi:10.1145/375551.375608}}.

\bibitem{albrecht19}
Martin~R. Albrecht, L\'{e}o Ducas, Gottfried Herold, Elena Kirshanova, Eamonn
  Postlethwaite, and Marc Stevens.
\newblock The general sieve kernel and new records in lattice reduction.
\newblock In {\em EUROCRYPT}, pages 717--746, 2019.

\bibitem{andoni06}
Alexandr Andoni and Piotr Indyk.
\newblock Near-optimal hashing algorithms for approximate nearest neighbor in
  high dimensions.
\newblock In {\em FOCS}, pages 459--468, 2006.
\newblock URL:
  \url{http://ieeexplore.ieee.org/xpl/articleDetails.jsp?arnumber=4031381},
  \href {http://dx.doi.org/10.1109/FOCS.2006.49}
  {\path{doi:10.1109/FOCS.2006.49}}.

\bibitem{andoni15cp}
Alexandr Andoni, Piotr Indyk, Thijs Laarhoven, Ilya Razenshteyn, and Ludwig
  Schmidt.
\newblock Practical and optimal {LSH} for angular distance.
\newblock In {\em NIPS}, pages 1225--1233, 2015.
\newblock URL:
  \url{https://papers.nips.cc/paper/5893-practical-and-optimal-lsh-for-angular-distance}.

\bibitem{andoni14}
Alexandr Andoni, Piotr Indyk, Huy~L\^{e} Nguy\^{e}n, and Ilya Razenshteyn.
\newblock Beyond locality-sensitive hashing.
\newblock In {\em SODA}, pages 1018--1028, 2014.
\newblock \href {http://dx.doi.org/10.1137/1.9781611973402.76}
  {\path{doi:10.1137/1.9781611973402.76}}.

\bibitem{andoni17}
Alexandr Andoni, Thijs Laarhoven, Ilya Razenshteyn, and Erik Waingarten.
\newblock Optimal hashing-based time-space trade-offs for approximate near
  neighbors.
\newblock In {\em SODA}, pages 47--66, 2017.
\newblock \href {http://dx.doi.org/10.1137/1.9781611974782.4}
  {\path{doi:10.1137/1.9781611974782.4}}.

\bibitem{andoni15}
Alexandr Andoni and Ilya Razenshteyn.
\newblock Optimal data-dependent hashing for approximate near neighbors.
\newblock In {\em STOC}, pages 793--801, 2015.
\newblock \href {http://dx.doi.org/10.1145/2746539.2746553}
  {\path{doi:10.1145/2746539.2746553}}.

\bibitem{andoni16lb}
Alexandr Andoni and Ilya Razenshteyn.
\newblock Tight lower bounds for data-dependent locality-sensitive hashing.
\newblock In {\em SOCG}, pages 1--15, 2016.
\newblock \href {http://dx.doi.org/10.4230/LIPIcs.SoCG.2016.9}
  {\path{doi:10.4230/LIPIcs.SoCG.2016.9}}.

\bibitem{arya94}
Sunil Arya, David~M. Mount, Nathan~S. Netanyahu, Ruth Silverman, and Angela~Y.
  Wu.
\newblock An optimal algorithm for approximate nearest neighbor searching in
  fixed dimensions.
\newblock In {\em SODA}, pages 573--582, 1994.
\newblock URL: \url{http://dl.acm.org/citation.cfm?id=314464.314652}.

\bibitem{annbench2}
Martin Aumueller, Erik Bernhardsson, and Alexander Faithfull.
\newblock {ANN} benchmarks -- available online at
  \url{http://sss.projects.itu.dk/ann-benchmarks/}, 2017.
\newblock URL: \url{http://sss.projects.itu.dk/ann-benchmarks/}.

\bibitem{aumueller17}
Martin Aumueller, Erik Bernhardsson, and Alexander Faithfull.
\newblock {ANN}-benchmarks: A benchmarking tool for approximate nearest
  neighbor algorithms.
\newblock In {\em SISAP}, pages 34--49, 2017.
\newblock \href {http://dx.doi.org/10.1007/978-3-319-68474-1_3}
  {\path{doi:10.1007/978-3-319-68474-1_3}}.

\bibitem{baernstein76}
Albert Baernstein and B.A. Taylor.
\newblock Spherical rearrangements, subharmonic functions, and $\ast$
  -functions in $n$ -space.
\newblock {\em Duke Math. J.}, 43(2):245--268, 06 1976.
\newblock URL: \url{https://doi.org/10.1215/S0012-7094-76-04322-2}, \href
  {http://dx.doi.org/10.1215/S0012-7094-76-04322-2}
  {\path{doi:10.1215/S0012-7094-76-04322-2}}.

\bibitem{becker16lsf}
Anja Becker, L\'{e}o Ducas, Nicolas Gama, and Thijs Laarhoven.
\newblock New directions in nearest neighbor searching with applications to
  lattice sieving.
\newblock In {\em SODA}, pages 10--24, 2016.
\newblock \href {http://dx.doi.org/10.1137/1.9781611974331.ch2}
  {\path{doi:10.1137/1.9781611974331.ch2}}.

\bibitem{becker16cp}
Anja Becker and Thijs Laarhoven.
\newblock Efficient (ideal) lattice sieving using cross-polytope {LSH}.
\newblock In {\em AFRICACRYPT}, pages 3--23, 2016.
\newblock \href {http://dx.doi.org/10.1007/978-3-319-31517-1_1}
  {\path{doi:10.1007/978-3-319-31517-1_1}}.

\bibitem{annbench1}
Erik Bernhardsson.
\newblock {ANN} benchmarks -- available online at
  \url{https://github.com/erikbern/ann-benchmarks}, 2016.
\newblock URL: \url{https://github.com/erikbern/ann-benchmarks}.

\bibitem{bishop06}
Christopher~M. Bishop.
\newblock {\em Pattern Recognition and Machine Learning}.
\newblock Springer-Verlag, 2006.
\newblock URL: \url{https://www.springer.com/us/book/9780387310732}.

\bibitem{brouwer02}
Andries~E. Brouwer.
\newblock Lattices - \url{https://www.win.tue.nl/~aeb/latt/lattices.pdf}, 2002.
\newblock URL: \url{https://www.win.tue.nl/~aeb/latt/lattices.pdf}.

\bibitem{chandrasekaran18}
Karthekeyan Chandrasekaran, Daniel Dadush, Venkata Gandikota, and Elena
  Grigorescu.
\newblock Lattice-based locality sensitive hashing is optimal.
\newblock In {\em ITCS}, 2018.

\bibitem{charikar02}
Moses~S. Charikar.
\newblock Similarity estimation techniques from rounding algorithms.
\newblock In {\em STOC}, pages 380--388, 2002.
\newblock URL: \url{http://dl.acm.org/citation.cfm?doid=509907.509965}, \href
  {http://dx.doi.org/10.1145/509907.509965} {\path{doi:10.1145/509907.509965}}.

\bibitem{cheng68}
M.C. Cheng.
\newblock The clipping loss in correlation detectors for arbitrary input
  signal-to-noise ratios.
\newblock {\em IEEE Transactions on Information Theory}, 14(3):3, 1968.

\bibitem{cheng69}
M.C. Cheng.
\newblock The orthant probabilities of four {Gaussian} variables.
\newblock {\em The Annals of Mathematical Statistics}, 40(1):152--161, 1969.

\bibitem{christiani17}
Tobias Christiani.
\newblock A framework for similarity search with space-time tradeoffs using
  locality-sensitive filtering.
\newblock In {\em SODA}, pages 31--46, 2017.
\newblock \href {http://dx.doi.org/10.1137/1.9781611974782.3}
  {\path{doi:10.1137/1.9781611974782.3}}.

\bibitem{diaconis87}
Persi Diaconis and David Freedman.
\newblock A dozen de {F}inetti-style results in search of a theory.
\newblock {\em Annales de l'IHP Probabilit\'{e}s et statistiques},
  23(2):397--423, 1987.

\bibitem{dubiner10}
Moshe Dubiner.
\newblock Bucketing coding and information theory for the statistical
  high-dimensional nearest-neighbor problem.
\newblock {\em IEEE Transactions on Information Theory}, 56(8):4166--4179, Aug
  2010.
\newblock \href {http://dx.doi.org/10.1109/TIT.2010.2050814}
  {\path{doi:10.1109/TIT.2010.2050814}}.

\bibitem{duda00}
Richard~O. Duda, Peter~E. Hart, and David~G. Stork.
\newblock {\em Pattern Classification (2nd Edition)}.
\newblock Wiley, 2000.
\newblock URL: \url{https://dl.acm.org/citation.cfm?id=954544}.

\bibitem{eshghi08}
Kave Eshghi and Shyamsundar Rajaram.
\newblock Locality sensitive hash functions based on concomitant rank order
  statistics.
\newblock In {\em KDD}, pages 221--229, 2008.
\newblock \href {http://dx.doi.org/10.1145/1401890.1401921}
  {\path{doi:10.1145/1401890.1401921}}.

\bibitem{genz19}
Alan Genz, Frank Bretz, Tetsuhisa Miwa, Xuefei Mi, Friedrich Leisch, Fabian
  Scheipl, and Torsten Hothorn.
\newblock {\em {mvtnorm}: Multivariate normal and $t$ distributions}, 2019.
\newblock R package version 1.0-10.
\newblock URL: \url{https://CRAN.R-project.org/package=mvtnorm}.

\bibitem{indyk98}
Piotr Indyk and Rajeev Motwani.
\newblock Approximate nearest neighbors: Towards removing the curse of
  dimensionality.
\newblock In {\em STOC}, pages 604--613, 1998.
\newblock \href {http://dx.doi.org/10.1145/276698.276876}
  {\path{doi:10.1145/276698.276876}}.

\bibitem{jiang06}
Tiefeng Jiang.
\newblock How many entries of a typical orthogonal matrix can be approximated
  by independent normals?
\newblock {\em The Annals of Probability}, 34(4):1497--1529, 2006.
\newblock \href {http://dx.doi.org/10.1214/009117906000000205}
  {\path{doi:10.1214/009117906000000205}}.

\bibitem{kennedy17}
Christopher Kennedy and Rachel Ward.
\newblock Fast cross-polytope locality-sensitive hashing.
\newblock In {\em ITCS}, pages 1--16, 2017.
\newblock \href {http://dx.doi.org/10.4230/LIPIcs.ITCS.2017.53}
  {\path{doi:10.4230/LIPIcs.ITCS.2017.53}}.

\bibitem{laarhoven15crypto}
Thijs Laarhoven.
\newblock Sieving for shortest vectors in lattices using angular
  locality-sensitive hashing.
\newblock In {\em CRYPTO}, pages 3--22, 2015.
\newblock \href {http://dx.doi.org/10.1007/978-3-662-47989-6_1}
  {\path{doi:10.1007/978-3-662-47989-6_1}}.

\bibitem{laarhoven15nns}
Thijs Laarhoven.
\newblock Tradeoffs for nearest neighbors on the sphere.
\newblock {\em arXiv:1511.07527 [cs.DS]}, pages 1--16, 2015.
\newblock URL: \url{https://arxiv.org/abs/1511.07527}.

\bibitem{laarhoven17hypercube}
Thijs Laarhoven.
\newblock Hypercube {LSH} for approximate near neighbors.
\newblock In {\em MFCS}, pages 1--20, 2017.
\newblock \href {http://dx.doi.org/10.4230/LIPIcs.MFCS.2017.7}
  {\path{doi:10.4230/LIPIcs.MFCS.2017.7}}.

\bibitem{orthantmo}
Thijs Laarhoven.
\newblock {MathOverflow}: Normal multivariate orthant probabilities --
  \url{https://mathoverflow.net/questions/334212/normal-multivariate-orthant-probabilities},
  2019.
\newblock URL:
  \url{https://mathoverflow.net/questions/334212/normal-multivariate-orthant-probabilities}.

\bibitem{lee14}
Yongjae Lee and Woo~Chang Kim.
\newblock Concise formulas for the surface area of the intersection of two
  hyperspherical caps.
\newblock {\em KAIST Technical Report}, 2014.

\bibitem{mariano15}
Artur Mariano, Thijs Laarhoven, and Christian Bischof.
\newblock Parallel (probable) lock-free {HashSieve}: a practical sieving
  algorithm for the {SVP}.
\newblock In {\em ICPP}, pages 590--599, 2015.
\newblock \href {http://dx.doi.org/10.1109/ICPP.2015.68}
  {\path{doi:10.1109/ICPP.2015.68}}.

\bibitem{mariano17}
Artur Mariano, Thijs Laarhoven, and Christian Bischof.
\newblock A parallel variant of {LDSieve} for the {SVP} on lattices.
\newblock In {\em PDP}, pages 23--30, 2017.
\newblock \href {http://dx.doi.org/10.1109/PDP.2017.60}
  {\path{doi:10.1109/PDP.2017.60}}.

\bibitem{may15}
Alexander May and Ilya Ozerov.
\newblock On computing nearest neighbors with applications to decoding of
  binary linear codes.
\newblock In {\em EUROCRYPT}, pages 203--228, 2015.
\newblock \href {http://dx.doi.org/10.1007/978-3-662-46800-5_9}
  {\path{doi:10.1007/978-3-662-46800-5_9}}.

\bibitem{plackett54}
R.L. Plackett.
\newblock A reduction formula for normal multivariate integrals.
\newblock {\em Biometrika}, 41(3--4):351--360, 1954.

\bibitem{falconn}
Ilya Razenshteyn and Ludwig Schmidt.
\newblock {FALCONN} -- available online at \url{https://falconn-lib.org/},
  2016.
\newblock URL: \url{https://falconn-lib.org/}.

\bibitem{riesz30}
Frederic Riesz.
\newblock Sur une inegalite integrale.
\newblock {\em Journal of the London Mathematical Society}, s1-5(3):162--168,
  1930.
\newblock URL:
  \url{https://londmathsoc.onlinelibrary.wiley.com/doi/abs/10.1112/jlms/s1-5.3.162},
  \href
  {http://arxiv.org/abs/https://londmathsoc.onlinelibrary.wiley.com/doi/pdf/10.1112/jlms/s1-5.3.162}
  {\path{arXiv:https://londmathsoc.onlinelibrary.wiley.com/doi/pdf/10.1112/jlms/s1-5.3.162}},
  \href {http://dx.doi.org/10.1112/jlms/s1-5.3.162}
  {\path{doi:10.1112/jlms/s1-5.3.162}}.

\bibitem{shakhnarovich05}
Gregory Shakhnarovich, Trevor Darrell, and Piotr Indyk.
\newblock {\em Nearest-Neighbor Methods in Learning and Vision: Theory and
  Practice}.
\newblock MIT Press, 2005.
\newblock URL: \url{http://ttic.uchicago.edu/~gregory/annbook/book.html}.

\bibitem{shannon59}
Claude~E. Shannon.
\newblock Probability of error for optimal codes in a {G}aussian channel.
\newblock {\em Bell System Technical Journal}, 38(3):611--656, 1959.
\newblock URL:
  \url{http://www.alcatel-lucent.com/bstj/vol38-1959/bstj-vol38-issue03.html}.

\bibitem{sheppard99}
William~F. Sheppard.
\newblock On the application of the theory of error to cases of normal
  distribution and normal correlation.
\newblock {\em Philosophical Transactions of the Royal Society A},
  192:101--167, 1899.

\bibitem{sloanepackings}
Neil~J.A. Sloane.
\newblock Spherical codes: nice arrangements of points on a sphere in various
  dimensions.
\newblock URL: \url{http://neilsloane.com/packings/}.

\bibitem{steck62}
G.P. Steck.
\newblock Orthant probabilities for the equicorrelated multivariate normal
  distribution.
\newblock {\em Biometrika}, 49(3--4):433--445, 1962.

\bibitem{terasawa07}
Kengo Terasawa and Yuzuru Tanaka.
\newblock Spherical {LSH} for approximate nearest neighbor search on unit
  hypersphere.
\newblock In {\em WADS}, pages 27--38, 2007.
\newblock \href {http://dx.doi.org/10.1007/978-3-540-73951-7_4}
  {\path{doi:10.1007/978-3-540-73951-7_4}}.

\bibitem{terasawa09}
Kengo Terasawa and Yuzuru Tanaka.
\newblock Approximate nearest neighbor search for a dataset of normalized
  vectors.
\newblock In {\em IEICE Transactions on Information and Systems}, volume~92,
  pages 1609--1619, 2009.
\newblock URL: \url{http://search.ieice.org/bin/summary.php?id=e92-d_9_1609}.

\bibitem{vladut19}
Serge Vladut.
\newblock Lattices with exponentially large kissing numbers.
\newblock {\em Moscow J. Comb. Number Th.}, 8:163--177, 2019.

\end{thebibliography}


\end{document}